\documentclass[a4paper,fleqn]{cas-dc}
\usepackage[utf8]{inputenc}
\usepackage[T1]{fontenc}

\usepackage[numbers]{natbib}

\usepackage{hyperref}
\usepackage{amsmath,amsthm,amssymb,latexsym,amscd,amsfonts,enumerate}
\usepackage{color, fancyhdr, graphicx}
\usepackage{hyperref}
\usepackage{appendix}
\usepackage[english]{babel}
\usepackage{subfigure}
\usepackage{helvet}
\newtheorem{theorem}{Theorem}
\newtheorem{remark}{Remark}
\newtheorem{assumption}{Assumption}

\newtheorem{lemma}{Lemma}
\newtheorem{problem}{Problem}
\usepackage{algorithm}
\usepackage{algorithmicx}
\usepackage{color, fancyhdr, graphicx}
\usepackage{graphicx}
\usepackage{tabularx}
\usepackage{array,tabularx}
\usepackage{multirow}
\usepackage{hhline}
\usepackage{dsfont}
\usepackage{xcolor}
\usepackage{lineno} % number the lines
\definecolor{huy}{rgb}{0.2, 0.6, 0.8}
\graphicspath{{./figs/}}

\theoremstyle{thmstyleone}%
\begin{document}
\let\WriteBookmarks\relax
\def\floatpagepagefraction{1}
\def\textpagefraction{.001}
\shorttitle{A model free approach for CT OTC with unknown user-define cost and constrained control input via advantage function}
\shortauthors{Duc Cuong Nguyen et~al.}
%\begin{frontmatter}

\title [mode = title]{A model free approach for continuous-time optimal tracking control with unknown user-define cost and constrained control input via advantage function}                      
%\tnotemark[1,2]

%\tnotetext[1]{This document is the results of the research
%   project funded by the National Science Foundation.}

%\tnotetext[2]{The second title footnote which is a longer text matter
%   to fill through the whole text width and overflow into
%   another line in the footnotes area of the first page.}

\author[1]{Duc Cuong Nguyen}[]
\credit{Writing – original draft – review \& editing, Software, Data curation}
\author[2]{Quang Huy Dao}[]

\credit{Writing – original draft – review \& editing, Software, Conceptualization}

\author[2]{Phuong Nam Dao}[orcid=0000-0002-8333-5572]
\credit{Writing – original draft – review \& editing, Supervision, Methodology}
\cormark[1]
%\fnmark[1]
\ead{nam.daophuong@hust.edu.vn}
%\ead[url]%{www.cvr.cc, cvr@sayahna.org}

%\credit{Conceptualization of this study, Methodology, Software}

%\address[1]{Korea, ...}

%\fnmark[2]
%\ead{nam.daophuong@hust.edu.vn}
%\ead[URL]{nam.daophuong@hust.edu.vn}

%\credit{Data curation, Writing - Original draft preparation}
\address[1]{Faculty of Engineering, Lund University, Lund, Sweden}
\address[2]{School of Electrical and Electronic Engineering, Hanoi University of Science and Technology, Hanoi, Vietnam}

%\address[2]{Hanoi University of Science and Technology, Hai Ba Trung District, Vietnam}

%\cortext[cor1]{Corresponding author}
%\cortext[cor2]{Principal corresponding author}
%\fntext[fn1]{This is the first author footnote. but is common to third
%  author as well.}
%\fntext[fn2]{Another author footnote, this is a very long footnote and
%  it should be a really long footnote. But this footnote is not yet
%  sufficiently long enough to make two lines of footnote text.}

%\nonumnote{This note has no numbers. In this work we demonstrate $a_b$
%  the formation Y\_1 of a new type of polariton on the interface
%  between a cuprous oxide slab and a polystyrene micro-sphere placed
%  on the slab.
%  }

\begin{abstract} 
This paper presents a pioneering approach to solving the linear quadratic regulation (LQR) and linear quadratic tracking (LQT) problems with constrained inputs using a novel off-policy continuous-time Q-learning framework. The proposed methodology leverages a novel concept of the Advantage function for linear continuous systems, enabling solutions to be obtained without the need for prior knowledge of the reward matrix weights, state resetting, or assuming the existence of a predefined admissible controller. This framework includes multiple algorithms (Algs) tailored to address these control problems under model-free conditions, without requiring any knowledge about system dynamics. Two distinct implementation methods are explored: the first processes state and input data over a fixed time interval, making it well-suited for LQR problems, while the second method operates over multiple intervals, offering a practical solution for tracking problems with constrained inputs. The convergence of the proposed algorithms is verified theoretically. Finally, the simulation results of the F-16 aircraft system are presented for the two problems to validate the effectiveness of the proposed method.

\end{abstract}
 
%\begin{graphicalabstract}
%\includegraphics{figs/grabs.pdf}
%\end{graphicalabstract}

%\begin{highlights_i}
%\item Research highlights item 1
%\item Research highlights item 2
%\item Research highlights item 3
%\end{highlights_i}


\maketitle
%\linenumbers
\section{Introduction}
Reinforcement Learning Control (RLC), a subbranch within data science, has become increasingly popular for learning optimal control behavior in dynamical systems \cite{dao2024nonlinear}. In dynamical systems, there are two widely used approaches for RLC to simultaneously implement the trajectory tracking problem and the consideration of minimizing the performance index: the linear quadratic regulator (LQR) and the linear quadratic tracking (LQT) \cite{zhao2023linear,pang2021robust}. The traditional LQR method is capable of simultaneously investigating based on the tracking error model, which is difficult to find the explicit representation in several practical systems, such as mobile robotics \cite{nguyen2025robust}. To fulfill this challenge, the direct LQT approach is implemented with the tracking error integrated into the performance index and several modifications of the Ricatti equation are investigated \cite{zhao2023linear,pang2021robust}. In \cite{zhao2023linear}, the On-Policy and Off-Policy RL methods are developed in the LQT problem in two phases using the integration of the discount factor into the kernel matrix. In addition, taking into account a matrix that contains both the control policy and the value function, a data-driven policy iteration (PI) algorithm is developed for continuous-time linear systems in the presence of unknown disturbances, and then the control policy for each iteration can be obtained from the matrix extraction mentioned above \cite{pang2021robust}. To address the challenge of mismatched disturbances, the zero-sum differential game and $L_2$ gain consideration are employed by the Nash equilibrium solution and Hamilton-Jacobi Issac (HJI) equation, respectively \cite{wang2025prescribed,li2025event}. However, these works are implemented on the basis of the knowledge of the model \cite{wang2025prescribed,li2025event}. Unlike PI algorithms, the value iteration (VI) method can achieve the value function after each step without solving the equations and the requirement of an admissible initial control policy \cite{liu2025value}. Several approaches to VI-based RL algorithms have been investigated to address dynamic uncertainties, such as using the neural network approximation (NN) for the value function and control policy \cite{liu2025value}, data collection for the computation of the linear optimal output regulation problem (LOORP) with exo-system consideration \cite{jing2025incremental}, inexact VI procedure for updating the value function after one step \cite{lai2025robust}. Moreover, to overcome the difficulty of input constraints, the $\tanh^{-\top}$ function is integrated in the $L_2-$ gain of the controlled system \cite{liu2025value}. On the other hand, the inexact VI procedure is also extended for a discrete-time stochastic system with disturbance considered as a probability distribution \cite{lai2025robust}. In \cite{Guo2025nearly}, the framework of integral RL with Actor-Critic NN and Quasi-SMC method is developed for nonlinear continuous time systems. Although these above methods can handle the dynamic uncertainties in developing RL algorithms, to address the problem of complete uncertainties, Off-Policy RL algorithms are necessary to employ by applying the control behavior to the model and then collecting the information of input, output variables for computation \cite{zhang2024off,jha2025off,treesatayapun2025model,huang2025inverse,gao2025resilient}. In addition, Off-Policy RL methods are also extended for several situations, such as $H_\infty$ Fault-Tolerant Control (FTC) \cite{zhang2024off}, safe optimal control \cite{jha2025off}, barrier functions \cite{treesatayapun2025model}, inverse optimal control \cite{huang2025inverse}, resilient control \cite{gao2025resilient}. In \cite{jha2025off}, the invariance of the safe set is satisfied by establishing a generalized safety-aware Hamilton-Jacobi-Bellman (HJB). However, although the Off-Policy RL methods guarantee the consideration of complete dynamic uncertainty, the calculation in each step requires simultaneous implementation of the control policy and the value function, which makes the complexity of computation. Thus, the efficient treatment of complete uncertainties plays a crucial role in boosting the practicality of RL algorithms. In addition, the development of RL algorithms subject to constraints is a challenge.           

One of the most popular model-free algorithms is the Q-learning method investigated to estimate the state action Q-function, which is able to easily achieve the control policy after each iteration step without solving the optimization problem. For implementation of learning the optimal Q-function, the Q-learning algorithm will use dynamic programming (DP) with Bellman equation for the Q-function, which is obtained from the value function in the RL algorithm next time. Hence, this can be conveniently implemented in discrete-time systems \cite{wang2023discounted,lopez2023efficient,xiao2025multistep,li2025optimal,xu2025q,yuan2025evolution,song2022model,zhao2025neural}. In the case of an existing discount factor in the performance index, the Q function is chosen from the Bellman function, an additional term obtained from the performance index next time, and the discount factor \cite{wang2023discounted}. This is able to implement the VI algorithm for the LQT problem with a discount factor \cite{wang2023discounted}. For the case of LQR problem, robustness properties and the consideration of persistence of excitation (PE) condition and the deadbeat control scheme are developed with data-based Q-learning \cite{lopez2023efficient}. On the other hand, the traditional method of establishing the Q-function from the Bellman function at the next time can be also implemented for multi-agent systems (MAS) and cyber-physical systems \cite{xiao2025multistep, li2025optimal}. In addition, the Q-learning algorithm can be integrated with an evolutionary algorithm to solve the LQT problem \cite{yuan2025evolution}. In \cite{zhao2025neural}, the Q-function is established from the augmented utility function, which is generated by the control barrier function (CBF), to complete the safe Q learning algorithm for the LQT problem. In \cite{xu2025q, song2022model}, in view of the relation between Q-functions at two consecutive time instants and the Neural Network (NN) approximation, a general Q-learning and an iterative deterministic Q-learning based control schemes are developed for nonlinear switched systems and nonlinear input-affine DTSs, respectively. However, due to the continuity of the time domain, it is hard to establish the Q-function from the value function in Continuous-Time Systems (CTS) for developing Q-learning algorithms \cite{Possieri,VAMVOUDAKIS201714,cui2025q,yu2025optimal,perrusquia2022solution}. For developing a Q-learning algorithm in linear time invariant CTS, the Q-function is designed from the Bellman function at the same time and the Hamiltonian term \cite{VAMVOUDAKIS201714}. Moreover, the weights of the Q-function and optimal control are trained by minimizing the squared norm of errors to \cite{VAMVOUDAKIS201714}. This Q-learning algorithm is extended for uncertain non-linear systems with mismatched perturbations by the NN approximation \cite{cui2025q}. A different approach to the Q-learning algorithm is developed for perturbed nonlinear systems with nonzero-sum game consideration by employing a time-varying Q-function and assessing the $H_2/H_{\infty}$ control \cite{yu2025optimal}. As a further study, when the development of model-free Q-learning algorithms in CTS and actuator saturation are simultaneously considered, whether we can find for LQR, LQT problems with efficient computation is
an interesting but challenging topic. On the other hand, the implementation of F-16 control systems still only considers Q-learning for DTS \cite{dao2025h} as well as conventional non-linear control methods \cite{f16original}.  

With the above discussions, the research gaps are presented as follows: $1)$ The development of off-policy Model-Free RL algorithms in CTS by special functions for convenient computation without solving the optimization problem after each step. $2)$ The handling of the satisfaction of the admissible control policy in the initial step of Q-learning algorithms and how to efficiently obtain the admissible control policy? $3)$ The implementation of off-policy Model-Free RL algorithms for CTS under actuator saturation and their development for robotic systems. To fill in the above research gaps, this article develops four off-policy Q-learning algorithms for CTS in the cases of LQR and LQT with actuator saturation. The main contributions are summarized below.
\begin{enumerate}
    \item[1)] Unlike Q-learning methods for CTS presented in \cite{Possieri,VAMVOUDAKIS201714,cui2025q,yu2025optimal,perrusquia2022solution}, this study introduces a novel off-policy model-free framework based on the advantage function for linear continuous-time systems (CTSs) to estimate the terms $\Lambda_{xx}, \Lambda_{xu}, \Lambda_{uu}$ derived from the Q-function instead of being utilized solely for its optimality properties. Furthermore, we provide a detailed analysis of the convergence property of the proposed Alg~\ref{Alg 2} for the LQR problem based on Theorem~\ref{theorem 2}.  
    
    \item[2)] It is different from the policy iteration (PI) technique \cite{zhang2024off,jha2025off,treesatayapun2025model,huang2025inverse,gao2025resilient} requiring that the initial control policy satisfies an admissible condition, which is eliminated by introducing new approaches for both the LQR case and the LQT problem of optimal tracking control with constrained inputs. Furthermore, in the LQR case, the admissible policy can be obtained efficiently by solving an appropriate Linear Matrix Inequalities (LMI), which is computationally more efficient than existing methods that rely on solving a BMI \cite{Possieri}.  
    
    \item[3)] The proposed methodologies are further extended to accommodate tracking control under input constraints, demonstrating their applicability to constrained control scenarios. This has been effectively verified in an F-16 aircraft system by simulation results. 

\end{enumerate}

The remainder of this article is organized as follows. Some mathematical preliminaries, remarks, and the original algorithm are stated in Section~\ref{sec:headings}. The main contributions with Off-Policy Q-learning Algorithms for LQR and Optimal Tracking Problems are presented in Section~\ref{Section 3}. Illustrative examples of F-16 with simulation figures are given in Section~\ref{Section 4}, and conclusions are drawn in Section~\ref{Section 5}.

\textit{Notation:}
Throughout this paper, \( \mathbb{R} \) and \( \mathbb{N} \) represent sets of real numbers and non-negative integers, respectively. The operators \( \mathrm{vec}(\cdot) \), \( \otimes \), \( \mathrm{vec}_S(\cdot) \), and \( \otimes_S \) correspond to the vectorization operator, the standard Kronecker product, the symmetric vectorization operator and the symmetric Kronecker product, respectively. The identity matrix is denoted by \( I \), and its dimension is determined as appropriate. And $\mathrm{diag}(v)$ denotes a diagonal matrix whose diagonal elements are given by the entries of the vector $v \in \mathbb{R}^n$. For a symmetric matrix $S \in \mathbb{R}^{n \times n}$, the maximum and minimum eigenvalues are denoted by $\lambda_{\max}$ and $\lambda_{\min}$, respectively. The spectral norm of a matrix $A \in \mathbb{R}^{m \times n}$ is defined as $\|A\| = \sqrt{\lambda_{\max}(A^\top A)}$, and the Euclidean norm of a vector $v \in \mathbb{R}^n$ is given by $\|v\| = \sqrt{v^\top v}$. Finally, we denote $\nabla f(x(t))$ as the partial derivative of the function $f(x(t))$ with respect to the variable $x(t)$.

\section{Problem Statement and Preliminaries}
\label{sec:headings}
Consider a linear time invariant continuous-time system as presented by the following dynamic equation:
\begin{equation} \label{sys1}
    \dot{x}(t) = Ax(t) +Bu(t)
\end{equation}
where $x(t) \in \mathbb{R}^{n} $ is the state vector and $u(t) \in \mathbb{R}^{m}$ is the control input of the system. In addition, $A \in \mathbb{R}^{n \times n}$, $B \in \mathbb{R}^{n \times m}$ are the constant matrix and are assumed to be an unknown matrix. 
In this paper, two optimal control problems~\ref{Prob 1},~\ref{Prob. 2} for linear continuous-time systems \eqref{sys1} will be investigated as follows: 
\begin{problem}\label{Prob 1}
The first problem involves designing a state feedback controller $u^{*}(x)$ for model \eqref{sys1} to minimize the infinite-horizon quadratic cost over the trajectory of the closed-loop system
\begin{equation} \label{idxcost1}
    J(x(0),u)=\int^{\infty}_0 r\left(x(\tau),u(\tau)\right) d\tau
\end{equation}
where $r\left(x(t),u(t)\right)=M_m(x(t)) + u(t)^\top  Ru(t)$, $M_m(x(t))= x^\top(t) Mx(t)$, $M \in \mathbb{R}^{n \times n}$ is a semi positive definite matrix and $R \in \mathbb{R}^{m \times m}$ is a positive definite matrix.
\end{problem} 

The dynamic systems~\eqref{sys1} discussed in this article are subject to appropriate assumptions, which are described as follows:
\begin{assumption}\label{Assumption 01}
 The pair $\left(A,\sqrt{M}\right)$ is observable and the pair (A,B) is stabilizable.  
\end{assumption}  
\begin{assumption}\cite{Possieri}\label{Assumption 02}
     Although $M$ and $R$ in \eqref{idxcost1} are two unknown matrices, the reward signal $r(x(t), u(t))$ is directly measurable.
\end{assumption}

The observability condition in the assumption~\ref{Assumption 01} is necessary to ensure the uniqueness and existence of the solution to the optimization problem~\ref{Prob 1}. Meanwhile, Assumption~\ref{Assumption 02} refers to the establishment of certain problems related to the optimal output feedback as $y(t)=Cx(t)$, with an unknown matrix $C$, as well as predictive control problems of the economic model \cite{Possieri}. The extension of the constraint consideration in problem~\ref{Prob 1} is shown in the following problem.
\begin{problem}\label{Prob. 2}
Similarly to problem~\ref{Prob 1}, the objective of the second problem is to design an optimal feedback control input that minimizes some given cost function on the trajectory of the system, subject to input constraints $| u_i(t)| \leq \lambda, j= 1,...,m$.
Motivated by the concept of using the barrier cost function in \cite{MODARES20141780} to ensure that the input satisfies the constraints, the cost function is expressed as follows:
\begin{equation} \label{cost_function2}
    \begin{aligned}
        J&=\int^{\infty}_0 r\left(x(\tau),u(\tau)\right) d\tau\\
    &= \int^{\infty}_0 \bigg(M_m\left(x(\tau)\right)+ U\left(u(\tau)\right)\bigg) d\tau
    \end{aligned}
\end{equation}
where $U(u) = 2 \int_0^u \left(\lambda \beta^{-1} (v/\lambda)\right)^\top  R dv$, $\beta(\cdot)$ is hyperbolic tangent function $\tanh(\cdot)$ and $R$ is assumed to be diagonal matrix. In this establishment, $R$ is chosen as a positive definite diagonal matrix. $M$ and $R$ are also considered unknown matrices, while the reward signal $r(x(t), u(t))$ is measurable.
\end{problem}

In this section, the model-free Q-learning method \cite{Possieri}, which serves as a solution to Problem~\ref{Prob 1}, can be revisited along with some discussions. Building upon this foundation, a novel framework that is explicitly tailored to the LQR problem will be introduced. Subsequently, the second problem is explored through the lens of optimal tracking control and solved again using a sophisticated approach inspired by the proposed methodologies. 

According to the optimal control literature, the optimal feedback controller of problem~\ref{Prob 1} is derived by solving the Riccati equation given as:
\begin{equation}\label{Riccati}
A^{\top}P^* + P^*A + M - P^*BR^{-1}B^{\top}P^* = 0
\end{equation}
The unique solution, $P^* \leq 0$, represents a critical component of the so-called optimal value function, $V = x^\top P^* x$, and is also utilized to compute the optimal controller gain
\begin{equation}\label{opt-con}
    u^*(t)= -K^*x(t)= - R^{-1}B^\top P^* x(t) 
\end{equation}
However, as the dimensionality of the system state increases, solving equation \eqref{Riccati} becomes an increasingly formidable challenge. To address this issue with numerical efficiency, an iterative procedure is proposed to approximate the solution $P^*$, as outlined in Kleinman's Theorem~\ref{Theorem 1}.
\begin{theorem} \label{Theorem 1}\cite{Kleinman} 
Suppose that $\hat{K}_0 \in \mathbb{R}^{n \times m}$ is a stabilizing feedback gain matrix such that all eigenvalues of $A + B\hat{K}_0$ have negative real parts in the complex plane. Starting with the initial value $\hat{K}_0$, define a sequence of symmetric, positive definite matrices $P_i$ and feedback gains $\hat{K}_i$ iteratively as follows:
\begin{enumerate}
    \item Solve the Lyapunov equation for finding $P_i$:  
\begin{equation}  \label{value_evaluation}
A_i^\top P_i + A_i P_i = -M - \hat{K}_i^\top R \hat{K}_i,  
\end{equation}  
where $A_i = A - B\hat{K}_i$, with $P_i = P_i^\top > 0$.
    \item Update the feedback gain using:  
\begin{equation}  \label{policy_improvement}
\hat{K}_{i+1} = -R^{-1} B^\top P_i.  
\end{equation}  
\end{enumerate}
Then, the following properties hold:  
\begin{enumerate}  
\item The matrix $A + B\hat{K}_i$ is Hurwitz for all $i$.  
\item The sequence of solutions satisfies $P^* \leq P_{i+1} \leq P_i$, where $P^*$ is the optimal solution.  
\item The sequence of gains and solutions converge as:  
\begin{align*}  
        \lim_{i \to \infty} \hat{K}_i &= K^*, \\  
        \lim_{i \to \infty} P_i &= P^*,  
    \end{align*}  
    where $K^*$ is the optimal feedback gain.  
\end{enumerate}  
\end{theorem}
\begin{proof}
    See the work in \cite{Kleinman}.
\end{proof}
The on-policy method presented in \cite{Possieri} provides a model-free implementation of Theorem~\ref{Theorem 1} by leveraging the properties of the Q-function as follows:  
 \begin{equation}\label{Q func}
 \begin{aligned}
         Q(x,u) &= \hat{V}(x) + \nabla \hat{V}(x)(Ax + Bu) + x^{\top}Mx + u^{\top}Ru\\
         &=\begin{bmatrix}
             x\\
             u
         \end{bmatrix}^\top \begin{bmatrix}
    \hat{P}+M+A^\top \hat{P}+ \hat{P}A & \hat{P}B\\
    B^\top \hat{P} & R
    \end{bmatrix}
    \begin{bmatrix}
        x\\
        u
    \end{bmatrix}\\
    &= x^\top H_{xx}x+2x^\top H_{xu}u+u^\top H_{uu}u
     \end{aligned}
 \end{equation}
where $\hat{P}=\hat{P}^\top$, $H_{xx} = \hat{P} + M + \hat{P}A + A^T \hat{P}$, $ H_{xu} = \hat{P}B$ and $ {H_uu} = R$. 
The Q function \eqref{Q func} was first introduced in \cite{VAMVOUDAKIS201714} for a different purpose. In its original formulation, the Q function was defined specifically for the case where the value function is optimal, leveraging the optimality of the Hamiltonian function \cite{VAMVOUDAKIS201714}. In contrast, an alternative approach proposed in \cite{Possieri} utilizes the Q function differently. By partitioning the Q function \eqref{Q func}, it is shown that the Q function possesses a certain unique characteristic, since it encapsulates the quantities $H_{xx}$, $ H_{xu}$, $H_{uu}$, which remain conserved along any trajectory of the system (see \cite{Possieri}, Lemma 1). This property of the Q-function enables the introduction of a method to determine these quantities directly from system trajectory data by solving the following equation:
\begin{equation} \label{Q equation}
\begin{aligned}
        \xi_0(t)+ \mathrm{vec}_{S}^\top &\left(\hat{P} \right)\psi_0(t)=  \mathrm{vec}_{S}^\top \left(\hat{H}_{xx} \right)\phi_0(t) + \\
        & +2\mathrm{vec}^\top \left(\hat{H}_{xu}\right)\mu_0(t) + \mathrm{vec}_{S}^\top \left(\hat{H}_{uu}\right)\nu_0(t)
\end{aligned}
\end{equation}
where $\phi_0(t)=\int_{t_0}^t x(\tau)\otimes_{S}x(\tau)d\tau$, $\mu_0(t)=\int_{t_0}^t u(\tau)\otimes x(\tau)d\tau$, $\nu_0(t)=\int_{t_0}^t u(\tau)\otimes_{S}u(\tau)d\tau$,  $\xi_0(t)=\int_{t_0}^t  r(x(\tau),u(\tau))d\tau$, $\psi_0(t)=\int_{t_0}^t x(\tau)\otimes_{S}x(\tau)d\tau+ x(t)\otimes_{S}x(t)-x(t_0)\otimes_{S}x(t_0)$.
Once the solution to equation \eqref{Q equation} is obtained, the feedback gain in Step 2 of Theorem~\ref{Theorem 1} can be computed directly from the Q-function, without requiring any prior knowledge of the system dynamics, as follows:
\begin{equation}
\hat{K}= \hat{H}_{uu}^{-1} \hat{H}_{xu}^\top  
\end{equation}
On the other hand, to perform Step 1 of Theorem~\ref{Theorem 1} in a model-free manner, the Integral Reinforcement Learning (IRL) approach is employed to compute the value function using only the collected data: 
\begin{equation} \label{IRL}
     \mathrm{vec}_S^\top \left(\hat{P} \right)\psi_{1}(t)= \mathrm{vec}^\top \left(M+\hat{K}^\top R \hat{K} \right)\delta_1(t)
\end{equation}
where $\psi_{1}(t)=x(t_0)\otimes_{S}x(t_0)-x(t)\otimes_{S}x(t)$, $\delta_1(t)= \int_{t_0}^t x(\tau)\otimes x(\tau)d\tau$.

Note that there exists a method to initialize the gain of the initial feedback controller that satisfies the admissibility assumption in Theorem~\ref{Theorem 1}, based solely on knowledge of the Q function. Without requiring any information on the dynamics of the system, this method involves solving the following Bilinear Matrix Inequality (BMI) \cite{Possieri}:
    \begin{equation} \label{BMI}
     \hat{H}_{xx} - W + \hat{H}_{xu} \hat{K} + \hat{K}^\top  \hat{H}_{xu}^\top  < 0, \quad W > 0 
  \end{equation}
  s.t
  \begin{equation}\label{constraintBMI}
  \begin{aligned} 
 &\left[\mathrm{ \mathrm{vec}}_{S}\left(\hat{H}_{xx}\right), \mathrm{ \mathrm{vec}}\left(\hat{H}_{xu}\right), \mathrm{ \mathrm{vec}}_{S}\left(\hat{H}_{uu}\right)\right] \\
&\quad\quad= \Omega_0^\dagger \left(\Xi_0 + \Psi_0 \mathrm{vecs}(W)\right)   
  \end{aligned}
  \end{equation}
where $\Omega_0 = \begin{bmatrix}
    C_0(t_1)^\top \\
    C_0(t_2)^\top \\
    \vdots \\
    C_0(t_N)^\top
\end{bmatrix},
\Xi_0 = \begin{bmatrix}
    \xi(t_1)^\top\\
    \xi(t_2)^\top \\
    \vdots \\
    \xi(t_N)^\top
\end{bmatrix}, 
\Psi_0 = \begin{bmatrix}
    \psi_0(t_1)^\top \\
    \psi_0(t_2)^\top \\
    \vdots \\
    \psi_0(t_N)^\top
\end{bmatrix}$, $C_0(t)= \left[\phi(t),2\mu(t),\nu(t) \right], t_1<t_2< \hdots < t_N.$

To summarize the Q learning method \cite{Possieri}, the Alg~\ref{Alg 1} is presented by combining the equations \eqref{Q equation} through \eqref{constraintBMI} as follows.
\begin{algorithm} 
    \caption{On-policy Q-learning}\label{Alg 1}
\begin{algorithmic}[1]
\State Collect trajectory data $\Omega_0$, $\Xi_0$, and $\Psi_0$ from the system under excitation by the behavior policy $u$.
\State Solve $K_0$ from the BMI depicted in \eqref{BMI} and \eqref{constraintBMI}.
\State \textbf{loop}:
\State Collect data $\bar{\delta}_{1}^i= \begin{bmatrix}
    \delta_1(t_{1,i})^\top \\
    \delta_1(t_{2,i})^\top \\
    \vdots \\
    \delta_1(t_{N,i})^\top
\end{bmatrix}$, $\bar{\psi}_{1}^i= \begin{bmatrix}
    \psi_{1}(t_{1,i})^\top \\
    \psi_{1}(t_{2,i})^\top \\
    \vdots \\
    \psi_{1}(t_{N,i})^\top
\end{bmatrix} $ from the system controlled by the policy $u_i = -K_i x$. This can be achieved by either resetting the system's state or adding probing noise to the control input.
\State Solve the following equation to obtain $P_i$:
$ \mathrm{vec}_S^\top (P_i) \bar{\psi}_{1}^i= \mathrm{vec}^\top (M+K_i^\top RK_i)\bar{\delta}_{1}^i$.
\State Let $   \left[ \mathrm{vec}_S \left(H_{xx(i)}\right),  \mathrm{vec}\left(H_{xu(i)}\right),  \mathrm{vec}_S \left(H_{uu(i)}\right)\right] = \Omega_0^{\dagger}\left(\Xi_0 + \Psi_0  \mathrm{vec}_S \left(P_i \right) \right)$.
\State Update the control policy $K_{i+1}=H_{uu(i)}^{-1}H_{xu(i)}^\top $.
\State Let $i \leftarrow i+1$
\State \textbf{until} $\|P_i-P_{i-1}\| < \epsilon_t$, $\epsilon_t$ is a small real positive number.
\State Return $P^*= P_i$, $u^*(x)=-K_{i+1}x$.
\end{algorithmic}
\end{algorithm}

\begin{remark}\label{remark 1}
In this Alg~\ref{Alg 1}, it is important to note that in Step 4, the data are collected from the system driven by the exact control policy $u_i = -K_ix$, calculated in Step 7, making this Alg an on-policy method. However, this approach has a significant drawback, which is not always guaranteed that the collected data are sufficiently rich for calculating the value function $P_i$ in Step 5. This issue was addressed in \cite{Possieri} by resetting the system state at random to ensure that the collected data satisfies the Persistence of Excitation (PE) condition. However, resetting the system state is often impractical in real-world applications. An alternative solution, commonly adopted in several methods, is to introduce probing noise $e_n(t)$ into the control input. Although this approach can help the data satisfy the PE condition, it introduces a bias into the estimated value function $\hat{P}$. Specifically, in Alg~\ref{Alg 1}, if probing noise is added to the control input during data collection in Step 4, the value function computed in Step 5 will be modified as follows:
\begin{equation} 
\begin{aligned}
 \mathrm{vec}_S^\top \left(\hat{P} \right) \psi_{1}(t) =& \mathrm{vec}^\top \left(M+\hat{K}^\top R \hat{K} \right) \delta_1(t) 
 \\
& + \int_{t_0}^t e_n^\top(\tau) R e_n(\tau) d\tau    
\end{aligned}
\end{equation}
Compared to the original equation \eqref{IRL}, the modified equation includes an additional residual term, $\int_{t_0}^t e_n^\top(\tau) R e_n(\tau) , d\tau$, which depends on the probing noise $e_n(t)$. This residual term introduces bias into the estimation of $\hat{P}$, potentially compromising the accuracy of the Alg.
\end{remark}
 
Based on the analysis in Remark~\ref{remark 1}, a model-free off-policy Q-learning method will be proposed in the next section. This approach eliminates the need to reset the system state and avoids introducing bias into the estimation caused by adding probing noise.

\section{Main result}\label{Section 3}
In this section, several novel off-policy Q-learning methods with respect to continuous time will be introduced to address the Problems~\ref{Prob 1},~\ref{Prob. 2}. Furthermore, this off-policy Q-learning algorithm is improved for the optimal tracking problem with constrained input. Then, the convergence to the optimal control of these off-policy Q-learning algorithms is analyzed.  

\subsection{Off-policy Q-learning Approach for LQR Problem}
First, a new definition of the Advantage function will be constructed within the framework of the LQR problem.
From $Q(x,u)$ and $\hat{V}(x)$ defined in \eqref{Q func}, construct the Advantage function as follows:
\begin{equation} \label{advantage_func}
\begin{aligned}
    &A(x,u)=Q(x,u)-\hat{V}(x)\\
    &=\begin{bmatrix}
        x\\
        u
    \end{bmatrix}^\top  \begin{bmatrix}
    \hat{P}+M+A^\top \hat{P}+\hat{P}A & \hat{P}B\\
    B^\top \hat{P} & R
    \end{bmatrix}
    \begin{bmatrix}
        x\\
        u
    \end{bmatrix}-x^\top \hat{P}x\\
    &= x^\top \Lambda_{xx}x+2x^\top \Lambda_{xu}u+u^\top \Lambda_{uu}u
\end{aligned}
\end{equation}
where $\hat{P}=\hat{P}^\top$, $\Lambda_{xx} = M + \hat{P}A + A^T \hat{P}$, $ \Lambda_{xu} = \hat{P}B$ and $\Lambda_{uu} = R$.

Note that with this new definition, the advantage function can take a formulation similar to the Hamiltonian function $\mathcal{H}(x,u,\nabla V^*)$. However, similar to the Q-function defined in \eqref{Q func}, the advantage function will serve a different purpose. In this work, the advantage function will be used to estimate the components $\Lambda_{xx}, \Lambda_{xu}, \Lambda_{uu}$ derived from the Q-function instead of being utilized solely for its optimality properties, as the case being considered with the Hamiltonian function.
Then with the same derivation as in \cite{Possieri}, by considering the integral of the advantage function along the trajectory driven under the behavior policy $u$, yields:

\begin{equation}\label{advantage-equation}
\begin{aligned}
    \xi_{0,1}+ &\mathrm{vec}_{S}^\top \left(\hat{P} \right)\psi_{0,1}= \mathrm{vec}_{S}^\top \left(\hat{\Lambda}_{xx} \right)\phi_{0,1} + \\
    & +2\mathrm{vec}^\top \left(\hat{\Lambda}_{xu} \right)\mu_{0,1} + \mathrm{vec}_{S}^\top \left(\hat{\Lambda}_{uu} \right)\nu_{0,1}
\end{aligned}
\end{equation}
where $\phi_{0,1}=\int_{t_0}^{t_1}x(\tau)\otimes_{S}x(\tau)d\tau$, $\mu_{0,1}=\int_{t_0}^{t_1} u(\tau)\otimes x(\tau)d\tau$, $\nu_{0,1}=\int_{t_0}^{t_1} u(\tau)\otimes_{S}u(\tau)d\tau$,  $\xi_{0,1}=\int_{t_0}^{t_1}  r(x(\tau),u(\tau))d\tau$, $\psi_{0,1}= x(t_1)\otimes_{S}x(t_1)-x(t_0)\otimes_{S}x(t_0)$, $t_1>t_0$.
Once the Advantage function is obtained from \eqref{advantage-equation}, the policy update step, as described in Theorem~\ref{Theorem 1}, can be executed without any prior knowledge of the system dynamic: $\hat{K} = \hat{\Lambda}_{uu}^{-1} \hat{\Lambda}_{xu}^\top$. Subsequently, to avoid resetting the system state for satisfying the PE condition when calculating the value function $\hat{P}$, we adopt the off-policy method proposed in \cite{JIANG20122699}. First, the system dynamics \eqref{sys1} are rewritten as follows:

\begin{equation} \label{off-sys}
    \dot{x}= Ax+Bu_i+B(u-u_i)
\end{equation}
where $u_i = -K_i x$ and $K_i = R^{-1}B^\top P_i$, with $P_i$ representing the current estimate of the optimal solution $P^*$ as described in Theorem~\ref{Theorem 1}. Substituting the dynamics \eqref{off-sys} into Steps 1 and 2 of Theorem~\ref{Theorem 1}, we obtain: 
\begin{equation}\label{off-value}
\begin{aligned}
    x^\top (t_1) &P_ix(t_1)-x^\top (t_0)P_ix(t_0)= \\
    &\int_{t_0}^{t_1}\left(-x^\top(\tau) Mx(\tau)-x(\tau)K_i\right)^\top RK_ix(\tau) d\tau& \\
    &+2\int_{t_0}^{t_1}(u(\tau)+K_{i}x(\tau))^\top B^\top P_ix(\tau)d\tau
\end{aligned}
\end{equation}
Substituting (\ref{off-value}) into (\ref{advantage-equation}), yields:
\begin{equation}
\begin{aligned}
& \mathrm{vec}_{S}^\top (R)\nu_{0,1} - \mathrm{vec}^\top \left(K^\top _iRK_i \right)\delta_{0,1}+\\
& \quad + 2\mathrm{vec}^\top\left( P_i B \right) \mu_{0,1} +2 \mathrm{vec}^\top \left(K_i^\top B^\top P_i \right) \delta_{0,1}\\
&= \mathrm{vec}_{S}^\top \left(\Lambda_{xx(i)} \right)\phi_{0,1} +2\mathrm{vec}^\top \left(\Lambda_{xu(i)} \right)\mu_{0,1}\\
&\quad + \mathrm{vec}_{S}^\top \left(\Lambda_{uu(i)} \right)\nu_{0,1}
\end{aligned} 
\end{equation}
where $\delta_{0,1}=\int_{t_0}^{t_1}x(\tau)\otimes x(\tau)d\tau$.
Since equation \eqref{advantage-equation} holds, it follows that $\Lambda_{xx(i)} = M + P_iA + A^T P_i$, $ \Lambda_{xu(i)} = P_iB$ and $\Lambda_{uu(i)} = R$. To align with the policy improvement step in Theorem~\ref{Theorem 1}, the control policy should be satisfied at the current iteration as:$K_i= R^{-1}\Lambda_{xu(i-1)}^\top$ where $\Lambda_{xu(i-1)}$ is not a variable but rather a constant value from the previous iteration of $\Lambda_{xu(i)}$. With this analysis, the following relationship is derived:
\begin{equation} \label{advantage-relation}
\begin{aligned}
        2\mathrm{vec}\left(K_i^\top \Lambda_{xu(i)}^\top \right)\delta_{0,1} &- \mathrm{vec}\left(K_i^\top \Lambda_{xu(i-1)}^\top\right)\delta_{0,1}=\\
    &=\mathrm{vec}_S\left(\Lambda_{xx(i)}\right)\phi_{0,1}
\end{aligned}
\end{equation}
Substituting equation \eqref{advantage-relation} into equation \eqref{advantage-equation} results in the following modified equation for the Advantage function:
\begin{equation}\label{modified_advantage}
   \begin{aligned}
&\xi_{0,1}+  \mathrm{vec}_{S}^\top \left(P_i \right) \psi_{0,1}= \\
&2\mathrm{vec}^\top \left(\Lambda_{xu(i)} \right)\mu_{0,1}+ \mathrm{vec}_{S}^\top \left(\Lambda_{uu(i)}\right)\nu_{0,1}\\
&+\left[ 2\mathrm{vec}^\top \left(K_i^\top \Lambda_{xu(i)}^\top \right) - \mathrm{vec}^\top \left(K_i^\top \Lambda_{xu(i-1)}^\top \right) \right]\delta_{0,1}\\
       \end{aligned}
\end{equation}
Consider the following time instants $0<t_1<t_2< \hdots < t_N$, define the matrices as follows:
\begin{equation}\label{Matrices_Notation}
\begin{aligned}
&\bar{\psi}= \left[ \psi_{0,1},\psi_{1,2}, \hdots, \psi_{N-1,N}  \right]^\top, \bar{\xi}= \left[ \xi_{0,1},\xi_{1,2}, \hdots, \xi_{N-1,N}  \right]^\top\\
&\bar{\delta}= \left[ \delta_{0,1},\delta_{1,2}, \hdots, \delta_{N-1,N}  \right]^\top,\bar{\mu}= \left[ \mu_{0,1},\mu_{1,2}, \hdots, \mu_{N-1,N}  \right]^\top\\
&\bar{\nu}= \left[ \nu_{0,1},\nu_{1,2}, \hdots, \nu_{N-1,N}  \right]^\top, \bar{\phi}= \left[ \phi_{0,1},\phi_{1,2}, \hdots, \phi_{N-1,N}  \right]^\top\\
&\Delta_i= \left[ -\bar{\psi},2 \left(\bar{\mu}+\bar{\delta}\left(I_{n}\otimes K_i^\top \right) \right),\bar{\nu} \right]\\
&\Theta_i= \left[\mathrm{vec}_{S}^\top\left(P_i\right),\mathrm{vec}^\top\left(\Lambda_{xu(i)}\right), \mathrm{vec}_{S}^\top\left(\Lambda_{uu(i)}\right) \right]^\top\\
&\Xi_i= \bar{\xi}+ \bar{\delta}\mathrm{vec} \left(K_i^\top \Lambda_{xu(i-1)}^\top \right)
%\nonumber
\end{aligned}
\end{equation}
The equation \eqref{modified_advantage} is rewritten similarly at different time intervals $\left[t_1, t_2\right],..., \left[t_{N-1}, t_N\right]$, then we can achieve the following linear matrix equations:
\begin{equation}\label{updaterule1}
     \Delta_i \Theta_i = \Xi_i
\end{equation}
It can be observed that the parameter matrix $\Theta_i$ is determined from \eqref{updaterule1} if the matrix $\delta_i$ possesses full column rank, utilizing the pseudo-inverse:
\begin{equation} \label{pseudo-inv1}
    \Theta_i =   \left(\Delta_i^\top \Delta_i \right)^{-1}\Delta_i^\top \Xi_i
\end{equation}
Since the parameter matrix $\Theta_i$ obtained in \eqref{pseudo-inv1} contains not only the value function at each iteration of Theorem~\ref{Theorem 1} but also includes the components of the Advantage function, \\$\Lambda_{uu(i)}, \Lambda_{xu(i)}$, it follows that according to these components, the step of policy update can be performed as follows:
\begin{equation} \label{policy_update1}
    K_{i+1}= \Lambda_{uu(i)}^{-1}\Lambda_{xu(i)}^\top
\end{equation}
After deriving the key steps \eqref{pseudo-inv1} and \eqref{policy_update1} as the value evaluation and policy improvement procedures outlined in Theorem~\ref{Theorem 1}, our proposed method can now be summarized in Alg~\ref{Alg 2}.

\begin{remark}
    Our proposed algorithm qualifies as an off-policy method because the learning phase operates independently of the data collection phase \cite{JIANG20122699}. During data collection, any persistently exciting input can be used to drive the system without introducing bias into the learning process. 
\end{remark}
The implementation of our algorithm in practice requires the full rank condition of the matrix $\Delta_i$ in \eqref{updaterule1}, which is challenging, since this matrix $\Delta_i$ contains the term $K_i$ changing throughout iterations. To address this, we propose a method to verify the rank of $\Delta_i$ by examining the rank of the collected data alone, as detailed in the lemma below. 

\begin{lemma}\label{Lemma 1}
Assume there exists an integer $l_0>0$ such that, for all $l>l_0$:
\begin{equation}
\begin{aligned} \label{rank_cond1}
    &\mathrm{rank}[I_{xx},I_{xu},I_{uu}]= \frac{1}{2}n(n+1)+mn+\frac{1}{2}m(m+1)\\
\end{aligned}    
\end{equation}
where $I_{xx}= \int_{t_0}^{t_1} x\otimes xd\tau$, $I_{xu}= \int_{t_0}^{t_1} x\otimes ud\tau$ and $I_{uu}=\int_{t_0}^{t_1} u\otimes ud\tau$. Then, the matrix $\Delta_i$ in \eqref{updaterule1} is in full rank.
\end{lemma}

\begin{proof}
    To prove the lemma~\ref{Lemma 1}, it is shown that:
    \begin{equation}\label{trivial_solution}
        \Delta_iX=0
    \end{equation}
has the only solution $X=0$.
The proof begins with a contradiction. Assume $X=\left[Y^\top ,H^\top ,T^\top \right]^\top $ where $Y= \mathrm{vec}_S(p)$, $H= \mathrm{vec}(a_{ux})$, $T= \mathrm{vec}_S(a_{uu})$ and $a=\begin{bmatrix}
    a_{xx}&a_{xu}\\
    a_{xu}^\top &a_{uu}
\end{bmatrix}$, $p$ is a symmetric matrix, representing the solution at each iteration step for the advantage function and value function, respectively. Considering only the first two terms in both vectors \\ $\Delta_{i(1:2)}$, $X_{(1:2)}=\left[Y^\top ,H^\top \right]^\top $, the following equation is derived from the proof in \cite{JIANG20122699}: 
\begin{equation}
  \Delta_{i(1:2)}X_{1:2}=    I_{xx} \mathrm{vec}(O)+2I_{xu} \mathrm{vec}(N)
\end{equation}
where $O=A_i^\top Y+YA_i+K_i^\top \left(B^\top Y-H \right)+\left(YB-H^\top \right)K_i$ and $N= B^TY-H$.
The equation above aligns with (A.2) in \cite{JIANG20122699} when substituting $H=RZ$, where $Z$ corresponds to the solution for $K_{i+1}$ in each iteration.
This result leads to the following expanded representation:
\begin{equation}
    \Delta_iX= I_{xx} \mathrm{vec}(O)+2I_{xu} \mathrm{vec}(N)+I_{uu}T
\end{equation}
Applying \eqref{trivial_solution}, it can be concluded that:
\begin{equation}
\left[I_{xx}, 2I_{xu}, I_{uu} \right]\left[ \mathrm{vec}(O)^\top , \mathrm{vec}(N)^\top ,T\right]^\top =0    
\end{equation}
Under the rank condition stated in \eqref{rank_cond1}, the only valid solution to this equation is $O=0$ and $N=0$, $T=0$. Consequently, $a_{xu}= B^\top Y$ and $Y=0$ are shown by solving the Lyapunov equation $A^\top _iY+YA_i=0$.
In conclusion, $ \mathrm{vec}(a_{xu})=H=0$ and $Y=0$, $T=0$. This implies $X=0$, contradicting the assumption that $X \neq 0$. This completes the proof.
\end{proof}

\begin{theorem}\label{theorem 2}
    Under Assumption~\ref{Assumption 01}, starting from the initial admissible policy $K_0$, and the condition in Lemma~\ref{Lemma 1} is satisfied. Then, the matrices $K_i,i=1,2, \hdots,$ and $P_i,i=1,2, \hdots,$ generated by Alg~\ref{Alg 2} will converge to the optimal value $K^*,P^*$.
\end{theorem} 
 
\begin{proof}
    When the condition in Lemma~\ref{Lemma 1} is satisfied, the matrix $\Delta_i$ is full-rank, then the matrix $P_i$, $\Lambda_{xu(i)}$, $\Lambda_{uu(i)}$ are uniquely determined. On the other hand, the solution $P_i$ also satisfies equation \eqref{value_evaluation} by its definition, then the updated policy follows this form $K_{i+1}=-R^{-1}B^\top P_i$, which indicates that Alg~\ref{Alg 2} has the same procedure as Step 1, 2 in Theorem~\ref{Theorem 1}. From that, Alg~\ref{Alg 2} will inherit the particularly appealing convergence features of the Kleinman algorithm, including the quadratic convergence induced by Newton’s strategy underlying the approach proposed by Kleinman.    
\end{proof}                 
 
\begin{remark}
   The proof of Lemma~\ref{Lemma 1} reveals that if the rank condition \eqref{rank_cond1} is satisfied, the matrix $\Lambda_{uu}$ is accurately estimated to match the true value of $R$ after the first iteration. Specifically, $P_i$, $\Lambda_{xu(i)}$, and $\Lambda_{uu(i)}$ have been demonstrated to be uniquely determined at each iteration step. Considering the first iteration step, we obtain from \eqref{updaterule1} the following equation:
\begin{equation}
\begin{aligned}
  &\Delta_{1(col=1:2)}X+\bar{\nu}\mathrm{vec}_{S}\left(\Lambda_{uu(0)} \right)= \\
  & \quad \quad \quad \bar{\phi} \mathrm{vec} _{S}(M)+ \bar{\delta} \mathrm{vec}\left(K_i^\top \Lambda_{xu(init)}^\top \right) +\bar{\nu}\mathrm{vec}_{S}\left(R \right) 
\end{aligned}
\end{equation}
where $\Delta_{1(col=1:2)}$ is the vector of the first two row of vector $\Delta$ at first iteration, $X=\left[ \mathrm{vec}_{S}\left(P_0 \right), \mathrm{vec}\left(\Lambda_{xu(0)}\right)\right]$ and $\Lambda_{xu(init)}$ is the initial value of $\Lambda_{xu}$.
Referring to \cite{JIANG20122699}, it follows that $\Delta_{1(col=1:2)}X= \bar{\phi} \mathrm{vec}^\top _{S}(M)+  \bar{\delta} \mathrm{vec}\left(K_i^\top \Lambda_{xu(init)}^\top \right) $, leaving the trivial solution $\Lambda_{uu(0)}=R$.
 Thus, the matrix $R$ is accurately determined after the first iteration, allowing it to be excluded from subsequent calculations. The remaining iterations can then proceed following the steps outlined in Alg~\ref{Alg 1} of \cite{JIANG20122699}. This adjustment significantly reduces computational demands and memory usage at each step of the iteration.
\end{remark}
  
\begin{remark}
    At the start of the algorithm, the iteration process has not yet begun, which means that $K$ does not adhere to the iterative rule $K_i= R^{-1}\Lambda_{xu(i-1)}^\top$ and instead is derived from a predefined initial value. This implies that our algorithm requires two initial parameters: $K_0, \Lambda_{xu(init)}$. To ensure that the algorithm remains consistent with its intended structure and does not deviate from its original nature, we propose the setting $\Lambda_{xu(init)}=0$. 
    In this configuration, the first step of the algorithm seamlessly transitions into the value evaluation step $A_0^\top P_0 + A_0 P_0 = -M$,  followed by the improvement step $K_{1}= \Lambda_{uu(0)}^{-1}\Lambda_{xu(0)}^\top=\Lambda_{uu(0)}^{-1}B^TP_0$. Since $\Lambda_{uu(0)}$ provides an accurate estimation of $R$ and $A_0$ is Hurwitz due to the admissible gain $K_0$, the solution $P_0$ is unique and $K_1$ is also an admissible control gain. 
    Thus, the first step can be interpreted as an identification phase for $R$, which does not alter the algorithm's properties, as it only updates the initial admissible gain $K_0$ to a new admissible gain $K_1$ .
\end{remark}

To complete the algorithm in model-free fashion we will introduce a method to obtain the stable controller gain that without using the information about the system model as mentioned in the previous section.
However, unlike the previous method that relied on solving BMI for the stable controller, in this section, we will propose a method that can find a stabilizing controller by only solving LMI as follows:
\begin{lemma}\label{Lemma 2}
    If Assumption~\ref{Assumption 01} is satisfied and the collected data $\bar{\phi}$, $\bar{\mu}$, $\bar{\psi}$ as defined in \eqref{updaterule1}, are used to calculate the variables $S, \hat{K}$ for the following feasible LMI: \begin{equation}\label{LMI}
     S \hat{\Pi} + \hat{\Pi} S -\hat{\Gamma}_{xx}  - \varpi \hat{B}\hat{B}^\top  < 0, \quad S  > 0, \varpi >0 
  \end{equation}
  s.t
  \begin{equation}\label{constraintLMI2_a}
   \begin{aligned}
   &\left[\mathrm{ \mathrm{vec}}_{S}^\top\left(\hat{\Gamma}_{xx}\right), \mathrm{ \mathrm{vec}}^\top\left(\hat{\Gamma}_{xu}\right)\right]^\top =\\
   & \quad \quad \quad =\left(\Delta^\top \Delta \right)^{-1}\Delta^\top \left(  \bar{\psi} \mathrm{vec}_S(S) \right)        
   \end{aligned}
   \end{equation}
    
  \begin{equation} \label{constraintLMI2_b}
   \left[\mathrm{ \mathrm{vec}}_{S}^\top\left(\hat{\Pi}\right), \mathrm{ \mathrm{vec}}^\top\left(\hat{B}\right)\right]^\top = \left(\Delta^\top \Delta \right)^{-1}\Delta^\top   \bar{\varrho} 
\end{equation}
where $\Delta= \left[\bar{\phi},2\bar{\mu}\right]$ is assumed to be full column rank, $\varrho_{0,1}= x(t_1)^\top x(t_1)-x(t_0)^\top x(t_0)$ and $\bar{\varrho}= \left[ \varrho_{0,1},\varrho_{1,2}, \hdots, \varrho_{N-1,N}  \right]^\top$. Then, the feasible solution $\hat{K}= \frac{\varpi}{2}\hat{B}^\top S^{-1}$ is a stabilizing controller.
\end{lemma} 
\begin{proof}
Note that, using the dynamic constraint \eqref{sys1}, the following equations always hold:
\begin{equation}\label{LMIproof1}
    \psi_{0,1}  \mathrm{vec}_{S}^\top\left(S\right) = \mathrm{vec}_{S}^\top\left(\hat{\Gamma}_{xx}\right)\phi_{0,1} + 2\mathrm{vec}^\top\left(\hat{\Gamma}_{xu}\right) \mu_{0,1}
\end{equation} \label{LMIproof2}        
 \begin{equation}
     \varrho = \mathrm{vec}_{S}^\top\left(\hat{\Pi}\right)\phi_{0,1} + 2\mathrm{vec}^\top\left(\hat{B}\right) \mu_{0,1}
 \end{equation}  
where $\hat{\Pi}= A^\top+A$, $\hat{\Gamma}_{xx}= SA+A^\top S$, $\hat{\Gamma}_{xu}= SB$.
Suppose that matrix $\Delta$ is the full column rank, the approximation terms $\hat{\Pi}$, $\hat{B}$, $\hat{\Gamma}_{xx}$, $\hat{\Gamma}_{xu}$  will be uniquely determined from equations \eqref{constraintLMI2_a}, \eqref{constraintLMI2_b}. Then  the LMI \eqref{LMI} is equivalent to :
\begin{equation}
    SA^\top \ + AS -\varpi BB^\top<0, \quad S>0, \varpi>0 
\end{equation}
The LMI above is always feasible if the pair $(A,B)$ is stabilizable (see Theorem 6.4, \cite{duan2013lmis}). Furthermore, it is shown that choosing $K= \frac{\varpi}{2}B^\top S^{-1}$ ensures the stability of the system \eqref{sys1}.
\end{proof}
Using the proposed method, the stabilizing controller can be found by solving an LMI instead of a BMI as in the method \cite{Possieri}. This approach leads to simplification by reducing the complexity of optimization, making it more computationally efficient. Unlike BMI, which often involves non-convexities and may lead to difficulties in finding a global solution, LMI offers a convex formulation that guarantees feasibility and optimality under the given conditions, thus improving both the tractability and reliability of the solution.

\begin{algorithm}
    \caption{Off-policy Q-learning for LQR problem}\label{Alg 2}
\begin{algorithmic}[1]
\State Collect trajectory data $\bar{\phi},\bar{\delta}, \bar{\mu}, \bar{\xi}, \bar{\nu}, \bar{\psi}$ \eqref{Matrices_Notation} from the system under excitation by the behavior policy $u$.
\State Solve $K_0$ from the LMI depicted in Lemma~\ref{Lemma 2} and initiate $\Lambda_{xu(init)}$
\State \textbf{loop}:
\State Let $\Delta_i= \left[ -\bar{\psi},2 \left(\bar{\mu}+\bar{\delta}\left(I_{n}\otimes K_i^\top \right) \right),\bar{\nu} \right], \Xi_i= \bar{\xi}+ \bar{\delta}\mathrm{vec} \left(K_i^\top \Lambda_{xu(i-1)}^\top \right) $
\State Get $P_i$, $\Lambda_{xu(i)}$ and $\Lambda_{uu(i)}$ from the following expression:
$ \Theta_i =   \left(\Delta_i^\top \Delta_i \right)^{-1}\Delta_i^\top \Xi_i$
\State Update the control policy $K_{i+1}=\Lambda_{uu(i)}^{-1}\Lambda_{xu(i)}^\top $
\State Let $i \leftarrow i+1$
\State \textbf{until} $\|P_i-P_{i-1}\| < \epsilon_t$, $\epsilon_t$ is a small real positive number. $\epsilon_t$ is a small real positive number
\State Return $P^*= P_i$, $u^*(x)=-K_{i+1}x$

\end{algorithmic}
\end{algorithm}
 
Before addressing the problem~\ref{Prob. 2}, we will introduce an alternative implementation of Alg~\ref{Alg 2}, which will be crucial for the subsequent method that addresses problem~\ref{Prob. 2}. 

By following the derivation steps of the explorized policy iteration from \cite{LEE20122850}, the resulting formulation is as follows:
\begin{equation}\label{explored_HJB}
    \begin{aligned} 
&x^\top(t)P_i \left[ Ax(t) + B \left(u_i(x) + w(t)\right) \right]  \\
&= -M_m\left(x(t) - u_i^\top(x)\right)Ru_i(x) + x^\top(t)P_i  B w(t)
\end{aligned}
\end{equation}
where $w(t)$ is the exploration noise and in this framework, the behavior control policy is considered as $u(t)= u_i(t)+w(t)$. The above equation is obtained by introducing the exploration noise term $x^\top(t)P_i  B w(t)$ to both sides of the Hamilton-Jacobi-Bellman (HJB) equation. Integrating both sides of \eqref{explored_HJB} yields:
\begin{equation}\label{explored_irl}
    \begin{aligned}
&x^\top(t_0)  P_i x(t_0) + \int_{t_0}^{t_1} x^\top(\tau)  L_i w(\tau) d\tau\\ 
&\quad\quad= \int_{t_0}^{t_1} r \left(x(\tau), u_i(\tau) \right)  d\tau + x^\top(t_1)  P_i x(t_1)
\end{aligned}
\end{equation}
where $L_i= B^\top P_{i-1}$
The following Lemma~\ref{Lemma 3} is introduced to develop the computation in next Alg~\ref{Alg 3}.
\begin{lemma}\label{Lemma 3}
    Define $\Lambda_{xx},\Lambda_{xu}, \Lambda_{uu}$ as in \eqref{advantage_func}, the following equation holds:
    \begin{equation}\label{equiv}
    \begin{aligned}
&x^\top(t_0)  P_i x(t_0) +2 \int_{t_0}^{t_1} x^\top(\tau)  \Lambda_{xu(i)} w(\tau) d\tau  \\ 
+& \int_{t_0}^{t_1} w^\top(\tau)  \Lambda_{uu(i)} w(\tau) d\tau+ 2\int_{t_0}^{t_1} u_i^\top(\tau)  \Lambda_{uu(i)} w(\tau) d\tau\\
= &\int_{t_0}^{t_1} r(x(\tau), u(\tau)) \, d\tau + x^\top(t_1)  P_i x(t_1)     
    \end{aligned}
    \end{equation}
    where $u(t)= w(t)+u_i(t)=w(t)-K_ix$
\end{lemma}
\begin{proof}
Let \eqref{advantage-equation} be derived from the system dynamics with the behavior policy $u$. Subtracting \eqref{explored_irl} from this equation yields:
\begin{equation}
\begin{aligned}
     &\int_{t_0}^{t_1} x^\top(\tau)  L_i w(\tau) d\tau + \int_{t_0}^{t_1} w^\top(\tau)  R w(\tau) d\tau\\
     &+ 2\int_{t_0}^{t_1} u_i^\top(\tau)  R w(\tau) d\tau\\
     =&\int_{t_0}^{t_1} x^\top(\tau)  \Lambda_{xx(i)} x(\tau) d\tau
     + 2\int_{t_0}^{t_1} x^\top(\tau)  \Lambda_{xu(i)} u(\tau) d\tau\\
     +&\int_{t_0}^{t_1} u^\top(\tau)  \Lambda_{uu(i)} u(\tau) d\tau
\end{aligned}
\end{equation}
    Note that the right hand side of the equation can be rewrite as follows:
    \begin{equation}
        \begin{aligned}
            &\int_{t_0}^{t_1} x^\top(\tau)  \Lambda_{xx(i)} x(\tau) d\tau
     + 2\int_{t_0}^{t_1} x^\top(\tau)  \Lambda_{xu(i)} u(\tau) d\tau\\
     &\quad +\int_{t_0}^{t_1} u^\top(\tau)  \Lambda_{uu(i)} u(\tau) d\tau\\
     & = H_i+2\int_{t_0}^{t_1} x^\top(\tau)  \Lambda_{xu(i)} w(\tau) d\tau + \\
     &\quad +\int_{t_0}^{t_1} w^\top(\tau)  \Lambda_{uu(i)} w(\tau) d\tau + \\
     & \quad + 2\int_{t_0}^{t_1} u_i^\top(\tau)(\tau)  \Lambda_{uu(i)} w(\tau) d\tau
        \end{aligned}
    \end{equation}
where 
\begin{equation}
\begin{aligned}
    H_i=& \mathrm{vec}^\top \left(\Lambda_{xx(i)} \right)\phi_{0,1}- 2\mathrm{vec}^\top \left(\Lambda_{xu(i)}K_i\right)\phi_{0,1}\\
    &+ \mathrm{vec}^\top \left(K_i\top\Lambda_{uu(i)}K_i\right)\phi_{0,1}
    \nonumber    
\end{aligned}
\end{equation} 
From the definition of the Advantage function \eqref{advantage_func}, it is evident that $H_i$ corresponds to the Lyapunov equation \eqref{value_evaluation} at each iteration step, which is equal to zero. Therefore, we can deduce:
\begin{equation}
   \begin{aligned}
        &\int_{t_0}^{t_1} x^\top(\tau)  L_i w(\tau) d\tau + \int_{t_0}^{t_1} w^\top(\tau)  R w(\tau) d\tau\\
        &\quad+2\int_{t_0}^{t_1} u_i^\top(\tau)  R w(\tau) d\tau\\
        &=2\int_{t_0}^{t_1} x^\top(\tau)  \Lambda_{xu(i)} w(\tau) d\tau + \int_{t_0}^{t_1} w^\top(\tau)  \Lambda_{uu(i)} w(\tau) d\tau\\
        &\quad+2\int_{t_0}^{t_1} u_i^\top(\tau)  \Lambda_{uu(i)} w(\tau) d\tau
   \end{aligned} \label{lem3proof1}
\end{equation}
Adding $\int_{t_0}^{t_1} w^\top(\tau)  R w(\tau) d\tau+2\int_{t_0}^{t_1} u_i^\top(\tau)  R w(\tau) d\tau$ to both sides of equation \eqref{explored_irl} gives:
\begin{equation}
\begin{aligned}
    &x(t)^\top  P_i x(t) + \int_{t_0}^{t_1} x^\top(\tau)  L_i w(\tau) d\tau  \\ 
    &\quad+ \int_{t_0}^{t_1} w^\top(\tau)  R w(\tau) d\tau+2\int_{t_0}^{t_1} u_i^\top(\tau)  R w(\tau) d\tau\\
&= \int_{t_0}^{t_1} r(x(\tau), u(\tau)(\tau)) \, d\tau + x_{t+T}^\top  P_i x_{t+T}     
\end{aligned}\label{lem3proof2}
\end{equation}
Substituting \eqref{lem3proof1} into \eqref{lem3proof2} leads to the complement of the proof. 
\end{proof}
In view of the Lemma~\ref{Lemma 3}, we can derive the update rule as follows:
\begin{equation}
\label{adupdate}
  \Delta_i \Theta_i=\bar{\xi}
\end{equation}
where \begin{equation}
\begin{aligned}
&\eta_{0,1}= \int_{t_0}^{t_1} x(\tau)\otimes w(\tau) d\tau, \\
&I_{eu(0,1)}= \int_{t_0}^{t_1} w(\tau) \otimes_S w(\tau)+2u_i(\tau)\otimes_S w(\tau) d\tau,\\
&\bar{\eta}=\left[ \eta_{0,1},\eta_{1,2}, \hdots, \eta_{N-1,N}  \right]^\top,\\
&\bar{I}_{eu}=\left[ I_{eu(0,1)},I_{eu(1,2)}, \hdots, I_{eu(N-1,N)}  \right]^\top,\\
&\Delta_i = \left[-\bar{\psi},2\bar{\eta},\bar{I}_{eu} \right],\\
&\Theta_i= \left[\mathrm{vec}_{S}^\top\left(P_i\right),\mathrm{vec}^\top\left(\Lambda_{xu(i)}\right), \mathrm{vec}_{S}^\top\left(\Lambda_{uu(i)}\right) \right]^\top
\end{aligned}
\end{equation}
To sum up the method, we propose Alg~\ref{Alg 3} below:

\begin{algorithm}
    \caption{Time-Iterative Q-Learning}\label{Alg 3}
\begin{algorithmic}[1]
\State Collect trajectory data $\bar{\phi}, \bar{\mu}, \bar{\xi}, \bar{\nu}, \bar{\psi}$ \eqref{Matrices_Notation} from the system under excitation by the behavior policy $u$.
\State Solve $K_0$ from the LMI depicted in Lemma~\ref{Lemma 2} 
\State \textbf{loop}:
\State Collect trajectory data $\bar{\eta}, \bar{I}_{eu}$ \eqref{adupdate}, $ \bar{\xi}, \bar{\psi}$ \eqref{Matrices_Notation} from the system under excitation by the behavior policy $u$.
\State Let $\Delta_i = \left[-\bar{\psi},2\bar{\eta},\bar{I}_{eu} \right] $
\State Get $P_i$, $\Lambda_{xu(i)}$ and $\Lambda_{uu(i)}$ from the expression following:
$ \Theta_i =   \left(\Delta_i^\top \Delta_i \right)^{-1}\Delta_i^\top \bar{\xi}$
\State Update the control policy $K_{i+1}=\Lambda_{uu(i)}^{-1}\Lambda_{xu(i)}^\top $
\State Let $i \leftarrow i+1$, $t \leftarrow t +\Delta t$
\State \textbf{until} $\|P_i-P_{i-1}\| > \epsilon_t$.
\State Return $P^*= P_i$, $u^*(x)=-K_{i+1}x$

\end{algorithmic}
\end{algorithm}
Since Alg~\ref{Alg 3} follows an online update scheme, where the policy is updated incrementally as new data are collected in real time, the following lemma~\ref{Lemma 4} investigates the stability of the system during the training phase under the control policy induced by Alg~\ref{Alg 3}.
\begin{lemma}\label{Lemma 4} \cite{LEE20122850} If the initial policy $u_0$ is stabilizing, $M > 0$ and assume that matrix $\Delta_i$ is full-rank at each iteration, let $M_i=M+K_i^{\top} RK_i$, $C_i= \|RK_i\|/\lambda_{min}(M_i)$ and $\|w(t)\| \leq w_M$ then the closed-loop system $x = A_i x + Bw$ is uniformly ultimately bounded (UUB) with each ultimate bound $\|x\| \leq 2w_M C_i$ for $i \neq 1, i \in \mathbb{N}$. Moreover, as $i \rightarrow \infty$, $V_i$ and $u_i$ converge to the optimal solution $V^*$ and $u^*$, respectively.
\end{lemma}
\begin{proof}
See the proof in \cite{LEE20122850}.
\end{proof} 
 
\subsection{Off-policy Q-learning for optimal tracking problem with constrained input}
In this section, the main focus will be on solving the problem~\ref{Prob. 2} along with the optimal tracking consideration, utilizing and expanding on the theoretical results and analyses that were systematically presented in the earlier sections of this paper. To present the problem in a more specific analysis, this section will focus on designing an optimal feedback signal to drive the system \eqref{sys1} to follow a given reference signal, as defined by the following formula.
\begin{equation}
    \dot{x}_r(t)= A_rx_r(t)
\end{equation}
and the control input is restricted by $|u_i(t)| \leq \lambda$.
Define $\tilde{x}(t) = x(t)-x_r(t)$ as the tracking error, then the augmented system can be rewritten as:
\begin{equation}\label{aug_sys}
    \dot{X}(t)=FX(t)+Gu(t)
\end{equation}
where the augmented state variables vector \\ $X(t)=\left[\tilde{x}(t),x_r(t)\right]^\top $ and $F=\begin{bmatrix}
    A&A-A_r\\
    0&A_r
\end{bmatrix}$, $G=[B, 0]^\top $. 
For the cost function \eqref{cost_function2} described in Problem~\ref{Prob. 2}, but within the framework of optimal tracking, this cost function \eqref{cost_function2} can be modified as follows:
\begin{equation} \label{cost_cons}
    J=\int_0^T  e^{-\gamma (\tau-t)} r(X(t),u(t))d\tau
\end{equation}
where $r(X,u)= \tilde{Q}(X(t))+U(u(t))$, $\tilde{Q}(X)=  \begin{bmatrix}
\tilde{x}\\
x_r^\top 
\end{bmatrix}^\top 
\begin{bmatrix}
    M&0\\
    0&0
\end{bmatrix}$ $\begin{bmatrix}
\tilde{x}\\
x_r^\top 
\end{bmatrix}$ and $U(u) = 2 \int_0^u (\lambda \beta^{-1} (v/\lambda))^\top  R dv$, with $\gamma$ being the discount factor. Under Assumption~\ref{Assumption 01}, the discounted component is added to the cost function \eqref{cost_cons} to ensure that it remains bounded for certain values of gamma, even when the reference signal is unstable \cite{Modares2014linear}. Without this addition, the optimization problem would lose its significance as the cost function would become unbounded.

Before deriving the model-free solution for this optimal tracking control, it is essential to revisit the concept of optimality and the properties associated with it. Specifically, the optimal Hamiltonian for the optimal tracking problem \cite{MODARES20141780} is formulated as follows:
\begin{equation} \label{opt_cons}
\begin{aligned}
       \mathcal{H}(x,u^*,\nabla V^*)&=r(X,u^*) - \gamma V^* + (\nabla V^*)^\top \left( FX + Gu^* \right)\\
       &= X^{\top}  Q X - \gamma V^{*}(X) + (\nabla V^*)^\top FX\\
       &+ \lambda^2 \bar{R} \ln\left(1 - \tanh^2\left(\frac{1}{2}(\lambda R)^{-1} G^\top  \nabla V^* \right)\right) \\
       &=0
\end{aligned}
\end{equation}
where the optimal feedback control is 
\begin{equation} \label{opt_fb_cons}
    u^*=-\lambda \tanh\left(\frac{1}{2}(\lambda R)^{-1} G^\top  \nabla V^{*} \right)
\end{equation}
The optimal feedback \eqref{opt_fb_cons} has been proven to be the solution to equation \eqref{opt_cons}, and it ensures the asymptotic stability of the system \eqref{aug_sys} in the limit as the discount factor approaches zero \cite{MODARES20141780}. From \eqref{opt_cons} and \eqref{opt_fb_cons}, the extended version of Kleinman's algorithm can be derived as follows:
\begin{equation} \label{value_eva_cons}
    \begin{aligned}
        \mathcal{H}\left(x,u_i,\nabla V_i\right)=r(X,u_i) & - \gamma V_i \\
        & + (\nabla V_i)^\top \left( FX + Gu_i \right)=0
    \end{aligned}
\end{equation}
After solving the value evaluation step above, the policy update step will be performed as follows:
\begin{equation}\label{policy_up_cons}
     u_{i+1}=-\lambda \tanh\left(\frac{1}{2}(\lambda R)^{-1} G^\top  \nabla V_i \right)
\end{equation}

As shown in \cite{MODARES20141780}, the extended version of Kleinman's algorithm defined in \eqref{value_eva_cons}–\eqref{policy_up_cons} is guaranteed to converge, as formalized in the following lemma. 
\begin{lemma}\label{Lemma 5}\cite{MODARES20141780} 
    Consider the optimal tracking problem with constrained input for the system \eqref{sys1}, subject to the cost function defined in equation \eqref{cost_cons}. Assume that \(V^*\) is smooth and positive definite solutions to the Hamilton-Jacobi tracking equations, as expressed in \eqref{opt_cons} with optimal control inputs $u^*(X)=$ $-\lambda$ $\tanh\left(\frac{1}{2}(\lambda R)^{-1} G^\top  \nabla V^* \right)$. The iterative procedure \eqref{value_eva_cons} and \eqref{policy_up_cons} will generate series of $V_i,i=0,1,...$ and $u_i,i=0,1,...$ that will converge to the optimal solution $V^*$, $u^*$ and $u^*$ also make the augmented system \eqref{aug_sys} asymptotically stable with the Lyapunov function $V^*$ at the limit of $\gamma \rightarrow 0$.
\end{lemma}
Note that taking the time derivative of the value function $V_i$, which is the solution of equation \eqref{opt_cons} at each iteration, along the trajectory of the augmented system \eqref{aug_sys} generated by the behavior policy $u(t)$, yields:
\begin{equation}\label{general_devV}
    \dot{V_i}= \nabla V^\top _i FX+ \nabla V^\top _i G(u-u_i) + \nabla V_iGu_i
\end{equation}
where $u_i= -\lambda \tanh\left(\frac{1}{2}(\lambda R)^{-1} G^\top  \nabla V_{i-1} \right)$.
From \eqref{opt_cons} and \eqref{general_devV}, we propose the IRL form of the off-policy Bellman equation for the optimal tracking problem as follows:
\begin{equation}\label{off_policy_cons}
\begin{aligned}
&e^{-\gamma T} V_{i+1}(X(t+T)) - V_{i+1}(X(t))\\
&=\int_{t}^{t+T} e^{-\gamma (\tau-t)} \left( -M_m(x) - U(u_i)\right) d\tau \\
&+\int_t^{t+T}e^{-\gamma (\tau-t)} \left(\nabla V^\top _i G (u - u_i) \right) d\tau
\end{aligned}
\end{equation}
 From this, it can be observed that equation \eqref{off_policy_cons} serves as the value evaluation step, while the policy improvement step will have a form similar to equation \eqref{policy_up_cons}.
Since this is a nonlinear optimization problem, the solution for the value function $V(x)$ may be difficult to obtain or may not have an exact mathematical form. Therefore, according to \cite{MODARES20141780}, we can use a neural network to approximate the optimal value function as follows:
\begin{equation}\label{eqn56}
    V^*(x)= (W^*)^\top \Phi(x)+\varepsilon(x)
\end{equation}
The error $\varepsilon(x)$ is the approximation error that will reach zero if the number of independent elements of the basis function $\Phi(x)$ reaches $\infty$. When the number of neurons is sufficient, the current estimate of the optimal value function can be approximated with arbitrary precision and can be represented as $\hat{V}=\hat{W}^\top \Phi(x)$. Under this representation, the value function $V_i(x)$ at each iteration of the algorithm described in Lemma~\ref{Lemma 5} can be expressed as follows:
\begin{equation} \label{network}
    V_i(x) = W^\top _i\Phi(x)
\end{equation}
Then the off-policy Bellman equation \eqref{off_policy_cons} for the off-policy framework can be transformed into the following.
\begin{equation}\label{integral_offpolicy_cons}
\begin{aligned}
e^{-\gamma T} \Phi^\top &(X(t+T))W_{i+1} - \Phi^\top (X(t))W_{i+1}\\
&=\int_{t}^{t+T} e^{-\gamma (\tau-t)} \left( -M_m(x) - U(u_i)\right) d\tau \\
&+\int_t^{t+T}e^{-\gamma (\tau-t)} \left(\nabla V^\top _i G (u - u_i) \right) d\tau
\end{aligned}
\end{equation}
According to \eqref{policy_up_cons} and \eqref{eqn56}, the formulation of the policy update can be expressed as:
\begin{equation}
\begin{aligned}
    u_{i+1}&=-\lambda \tanh\left(\frac{1}{2}R^{-1}G^\top \nabla\Phi^\top  W_i \right)\\
        &=-\lambda \tanh\left(\frac{1}{\lambda}K_{i+1} \mathrm{vec}\left(\nabla\Phi^\top \right)\right)
\end{aligned}
\end{equation}
where $K_{i+1}=\frac{1}{2}R^{-1}(W_i^\top \otimes G^\top )$.
The Advantage function can be formulated by the fact that $ U(u)= 2 \lambda \left(\beta^{-1} \left(\frac{u}{\lambda}\right)\right)^\top  R u + \lambda^2 \bar{R} \ln\left( \underline{\bold{1}} - \left(\frac{u}{\lambda}\right)^2\right)$, where $\bar{R}$ is the row vector with each element being the diagonal element of matrix $R$ and $\underline{\bold{1}}$ is the column matrix with the ones:
\begin{equation}
\begin{aligned}
    A(x,u)=& Q(x,u)-V(x)\\
    =&\Lambda_{xx}(x)+  \mathrm{vec}^\top \left(\Lambda_{xu}\right)\left(u\otimes  \mathrm{vec}\left(\nabla\Phi^\top \right)\right)\\
    &+ \mathrm{vec}^\top (\Lambda_{uu})\bigg(u \otimes 2 \lambda \beta^{-1}\left( \frac{u}{\lambda}\right)\\
    & \quad \quad \quad+ \ln \left( \underline{\bold{1}} - \left(\frac{u}{\lambda}\right)^2\right) \otimes \lambda^2 \underline{\bold{1}}\bigg)
    \end{aligned}
\end{equation}
where $\Lambda_{xx}(x)= \tilde{Q}(X)-\gamma V_i + \nabla V_i^\top  FX$, $\Lambda_{xu}= W_i^\top \otimes G^\top $, $\Lambda_{xu}= R$.
Then the integral form of the advantage function is expressed as:
\begin{equation} \label{integral_advan_cons}
\begin{aligned}
    &e^{-\gamma T} \Phi^\top (X(t+T))W_{i+1} - \Phi^\top (X(t))W_{i+1}\\
    &+\int_{t}^{t+T} e^{-\gamma (\tau-t)} \left( M_m(x) + U(u)\right) d\tau\\
    &=\int_0^T  e^{-\gamma (\tau-t)} \biggl( \Lambda_{xx}(x)+  \mathrm{vec}^\top (\Lambda_{xu})\left(u\otimes  \mathrm{vec}\left(\nabla\Phi^\top \right)\right)+ \\
    & \mathrm{vec}^\top (\Lambda_{uu})\left(u \otimes 2 \lambda \beta^{-1}\left( \frac{u}{\lambda}\right)+ \ln \left( \underline{\bold{1}} - \left(\frac{u}{\lambda}\right)^2\right) \otimes \lambda^2 \underline{\bold{1}} \right) \biggr)d\tau
\end{aligned}
\end{equation}
Substituting \eqref{integral_offpolicy_cons} into \eqref{integral_advan_cons}, we obtain the following equation:
\begin{equation}
\begin{aligned}
    &\int_0^T  e^{-\gamma (\tau-t)} \left(- U(u_i)- \nabla V^\top _i G u_i \right)\\
    &\quad\quad= \int_0^T  e^{-\gamma (\tau-t)} \Lambda_{xx}(x)
\end{aligned}
\end{equation}
Then we can derive the following equation for advantage function and value function:
\begin{equation} \label{final_advantge_cons}
   \begin{aligned}
\tilde{\xi}_{0,1}+\tilde{\psi}_{0,1} W_i&+\tilde{\omega}_{0,1(i)}=- \mathrm{vec}^\top (\Lambda_{xu(i)})\tilde{\phi}_{0,1(i)}+\\  &+ \mathrm{vec}^\top (\Lambda_{xu(i)})\tilde{\mu}_{0,1} + \mathrm{vec}^\top (\Lambda_{uu(i)})\tilde{\nu}_{0,1}
       \end{aligned}
\end{equation}
where $\tilde{\psi}_{0,1}= e^{-\gamma T}\Phi^\top (X(t_1))-\Phi^\top (X(t_0))$, \\
$\tilde{\omega}_{0,1(i)}= \int_{t_0}^{t_1}  e^{-\gamma (\tau-t)}\int_0^{u_i} (\lambda \beta^{-1} (v/\lambda))^\top  \Lambda_{uu(i-1)}d\tau$, \\
$\tilde{\xi}_{0,1}=\int_{t_0}^{t_1} e^{-\gamma (\tau-t)} r(X,u)d\tau$,\\
$\tilde{\phi}_{0,1(i)}= \int_{t_0}^{t_1}  e^{-\gamma (\tau-t)} u_i\otimes  \mathrm{vec}(\nabla\Phi^\top )d\tau$, \\
$\tilde{\nu}_{0,1}= \int_{t_0}^{t_1}  e^{-\gamma (\tau-t)} \biggr(u \otimes 2 \lambda \beta^{-1}\left( \frac{u}{\lambda}\right)+ \ln \left( \underline{\bold{1}} - \left(\frac{u}{\lambda}\right)^2\right) \otimes$ $ \lambda^2 \underline{\bold{1}}\biggr)d\tau$, \\
$\tilde{\mu}_{0,1}= \int_{t_0}^{t_1}  e^{-\gamma (\tau-t)} u\otimes  \mathrm{vec}(\nabla\Phi^\top )d\tau$.

Define the following matrices:
\begin{equation}
\begin{aligned}
&\tilde{\psi}= \left[ \tilde{\psi}_{0,1},\tilde{\psi}_{1,2}, \hdots, \tilde{\psi}_{N-1,N}  \right]^\top, \tilde{\xi}= \left[ \tilde{\xi}_{0,1},\tilde{\xi}_{1,2}, \hdots, \tilde{\xi}_{N-1,N}  \right]^\top,\\
&\tilde{\nu}= \left[ \tilde{\nu}_{0,1},\tilde{\nu}_{1,2}, \hdots, \tilde{\nu}_{N-1,N}  \right]^\top,\tilde{\mu}= \left[ \tilde{\mu}_{0,1},\tilde{\mu}_{1,2}, \hdots, \tilde{\mu}_{N-1,N}  \right]^\top,\\
& \tilde{\phi}_i= \left[ \tilde{\phi}_{0,1(i)},\tilde{\phi}_{1,2(i)}, \hdots, \tilde{\phi}_{N-1,N(i)}  \right]^\top,\\
&\tilde{\omega}_i= \left[ \tilde{\omega}_{0,1(i)},\tilde{\omega}_{1,2(i)}, \hdots, \tilde{\omega}_{N-1,N(i)}  \right]^\top,\\
&\Delta_i= \left[ -\tilde{\psi},\tilde{\mu}-\tilde{\phi}_i,\tilde{\nu} \right],\\
&\Theta_i= \left[W^\top_i,\mathrm{vec}^\top\left(\Lambda_{xu(i)}\right), \mathrm{vec}^\top\left(\Lambda_{uu(i)}\right) \right]^\top,\\
&\Xi_i= \tilde{\xi}+ \tilde{\omega}_i,
\nonumber
\end{aligned}
\end{equation}
Then we can rewrite \eqref{final_advantge_cons} as linear matrix equations:
\begin{equation} \label{least_square1}
\Delta_i \Theta_i=\Xi_i
\end{equation}
The policy update step is then expressed in the following form:
\begin{equation}\label{alg4_policy_improve}
K_{i+1}=\frac{1}{2} \Lambda_{uu(i)}^{-1}\Lambda_{xu(i)}
\end{equation}
Then to construct the algorithm based on this update rule we will propose an iteration scheme together with time as described in Alg~\ref{Alg 4}. In this algorithm, in every time interval, the data term $\omega, \phi$ is calculated with the policy $u_i= -\lambda \tanh\left(\frac{1}{\lambda}K_{i} \mathrm{vec}\left(\nabla\Phi^\top \right)\right)$.
\begin{algorithm}
    \caption{Constrained input Off-policy Q-learning}\label{Alg 4}
\begin{algorithmic}[1]
\State Initiate admissible control policy gain $K_0$, $\epsilon_t$ is a small real positive number.
\State \textbf{loop}:
\State Collect trajectory data $\tilde{\phi}_i, \tilde{\mu}, \tilde{\xi}, \tilde{\nu}, \tilde{\psi}, \tilde{\omega}_i$ from the system under excitation by the behavior policy $u$.
\State Let $\Delta_i= \left[ -\tilde{\psi},\tilde{\mu}-\tilde{\phi}_i,\tilde{\nu} \right]$, $\Xi_i= \tilde{\xi}+ \tilde{\omega}_i$
\State Obtain $W_i$, $\Lambda_{xu(i)}$ and $\Lambda_{uu(i)}$ from the expression following:
$ \Theta_i =   \left(\Delta_i^\top \Delta_i \right)^{-1}\Delta_i^\top \Xi_i$
\State Update the control policy $K_{i+1}=\frac{1}{2}\Lambda_{uu(i)}^{-1}\Lambda_{xu(i)}^\top $
\State Let $u_{i+1}(x)= -\lambda \tanh\left(\frac{1}{\lambda}K_{i+1} \mathrm{vec}\left(\nabla\Phi^\top(x) \right)\right)$
\State Let $i \leftarrow i+1$, $t \leftarrow t +\Delta t$
\State \textbf{until} $\|W_i-W_{i-1}\| < \epsilon_t$
\State Return $V^*(x)= W_i^\top\Phi(x)$ and $u^*(x)=-\lambda \tanh\left(\frac{1}{\lambda}K_{i+1} \mathrm{vec}\left(\nabla\Phi^\top(x) \right)\right)$
\end{algorithmic}
\end{algorithm}
\begin{remark}
  The Alg~\ref{Alg 4} is required by the full-rank condition of the data $\Delta_i$ collected at each time step. However, in a practical system, several selections of activation will make this rank condition unfulfilled, so we can treat \eqref{least_square1} as a least-square problem. The Alg~\ref{Alg 4} in this study performs iteratively along with time, which means that the time intervals \( t_0 \) to \( t_N \) in each iteration are different. This approach is highly practical, as it allows for continuous data collection while the system is running without the need to reset it. However, this is not the only way to implement Alg~\ref{Alg 4} in practice. Alternatively, Alg~\ref{Alg 4} can be iteratively executed over a single time interval to avoid prolonged system excitation and to reduce the wear on the actuators. However, this approach requires additional memory to store continuous signals of \( x(t) \) and \( u(t) \) throughout the interval \( t_0 \) to \( t_N \) and the vector data $\Delta_i$ at each iteration may experience a rank reduction, which is detrimental to the least-squares method unless new data are sampled.
\end{remark}
\begin{theorem}\label{theorem 3}
Consider the optimal tracking control problem with constrained input for the augmented system \eqref{aug_sys} and the infinite-horizon cost function defined in \eqref{cost_cons}, under the Assumption~\ref{Assumption 01}. Assume that the data collection term \( \Delta_i \), as defined in \eqref{least_square1}, is full-rank, and let the initial policy \( u = -\lambda \tanh\left(K_0 \mathrm{vec}\left(\nabla \Phi^\top(x)\right)\right) \) be an admissible control policy. Under these conditions, the sequence \( W_i, i\in\mathbb{N}\) generated by Alg~\ref{Alg 5} will converge to the optimal solution \( W^* \) as $i \to \infty$.
\end{theorem}
\begin{proof}
First, dividing both sides of equation \eqref{off_policy_cons} by \( T \) and taking the limit yields:
\begin{equation}\label{lim_offpolicy_cons}
\begin{aligned}
        &\lim_{T\to 0} \frac{e^{-\gamma T} V_i(X(t+T)) - V_i(X(t))}{T} \\
    =& -\lim_{T\to 0} \frac{1}{T} \int_{t}^{t+T} e^{-\gamma(\tau-t)} \left( \tilde{Q}(X) + U(u_i) \right) d\tau  \\
    &+\lim_{T\to 0} \frac{1}{T}\int_t^{t+T}e^{-\gamma (\tau-t)} \left(\nabla V^\top _i G (u - u_i) \right) d\tau
\end{aligned}
\end{equation}
Note that the left hand side of \eqref{lim_offpolicy_cons}, by applying L'Hopital's rule, can be expressed as:
    \begin{equation}\label{left_hand}
    \begin{aligned}
        &\lim_{T\to 0} \frac{e^{-\gamma T} V_i(X(t+T)) - V_i(X(t))}{T} \\
        &= \lim_{T\to 0} \left[-\gamma e^{-\gamma T} V_i(X(t+T)) + e^{-\gamma T} V_i(X(t+T))\right] \\
        &= -\gamma V_i + \nabla V_i(F + Gu_i + G(u-u_i)) 
    \end{aligned}
    \end{equation}
Similarly, by applying the same calculation to the right-hand side of \eqref{lim_offpolicy_cons}, we can express the second term as:
\begin{equation} \label{right_hand}
\begin{aligned}
        &\lim_{T\to 0} \frac{1}{T} \int_{t}^{t+T} e^{-\gamma(\tau-t)} \left( \tilde{Q}(X) + U(u_i) \right) d\tau \\
        &+\lim_{T\to 0} \frac{1}{T}\int_t^{t+T}e^{-\gamma (\tau-t)} \left(\nabla V^\top _i G (u - u_i) \right) d\tau \\
    &=\tilde{Q}(X) + U(u_i)+ G(u-u_i)
\end{aligned}
\end{equation}
From \eqref{left_hand} and \eqref{right_hand} we obtain the following expression:
\begin{equation} \label{proof_irl}
   \tilde{Q}(X) + U(u_i)+ \nabla V_i \left( FX + Gu_i \right)=0
\end{equation}
The equation above is equivalent to the value evaluation step \eqref{opt_cons}, and since Alg~\ref{Alg 4} employs the same policy improvement step as in \eqref{opt_fb_cons}, this demonstrates that Alg~\ref{Alg 4} is equivalent to the iterative algorithm described in Alg~\ref{Alg 5}. If we represent the value function as \( V_i(x) = W_i^\top\Phi (x)\), which is uniquely determined in each iteration from the full-rank data \( \Delta_i \), then starting from the initial admissible gain \( K_0 \), the Alg~\ref{Alg 4} inherits the convergence property, as demonstrated in Lemma~\ref{Lemma 5}.
\end{proof}
On the other hand, it is also worth mentioning the extended version of Alg~\ref{Alg 3}, which solves the optimal tracking problem with constrained input and serves as an alternative way of implementing Alg~\ref{Alg 4}. To derive this algorithm, we begin by modifying the off-policy HJB equation \eqref{explored_HJB} into the form of the tracking off-policy HJB equation as follows:
\begin{equation} \label{tracking_explorized_HJB}
     r(X,u_i) - \left(\nabla V_i\right) (Gw) =\gamma V_i - \left(\nabla V_i\right) \left( FX + G(u_i+w) \right)
\end{equation}
Taking into account \eqref{tracking_explorized_HJB}, we can propose a similar IRL form as shown in \eqref{off_policy_cons}, given by:
\begin{equation}\label{IRL_cons}
\begin{aligned}
e^{-\gamma T} &V_{j+1}(X(t+T)) - V_{j+1}(X(t))=\\
&=\int_{t}^{t+T} e^{-\gamma (\tau-t)} \left( -M_m(x) - U(u_i)\right) \\
&\quad\quad\quad+\int_t^{t+1}e^{-\gamma (\tau-t)} \left( \nabla V^\top _i G w \right) d\tau
\end{aligned}
\end{equation}
Adding $U(u)$ to both sides of \eqref{IRL_cons}, and applying the definition of the advantage function, it follows that:
\begin{equation}\label{advantage_cons}
\begin{aligned}
\tilde{\xi}_{0,1}+\tilde{\psi}_{0,1} W_j &= \mathrm{vec}^\top \left(\Lambda_{xu(i)} \right)\tilde{\eta}_{0,1}\\
&+ \mathrm{vec}^\top \left(\Lambda_{uu(i)}\right)\left(\tilde{\nu}_{0,1}-\tilde{\nu}_{0,1(i)} \right)
\end{aligned}
\end{equation}
where $\tilde{\eta}_{0,1}= \int_{t_0}^{t_1} e^{-\gamma (t-\tau)} w\otimes vec(\nabla\Phi^T)d\tau$,\\
$\tilde{\nu}_{0,1(i)}= \int_{t_0}^{t_1}  e^{-\gamma (\tau-t)} \biggr(u_i \otimes 2 \lambda \beta^{-1}\left( \frac{u_i}{\lambda}\right)+ \ln \left( \underline{\bold{1}} - \left(\frac{u_i}{\lambda}\right)^2\right)$ $ \otimes \lambda^2 \underline{\bold{1}}\biggr)d\tau$.\\
The expression above can also be rewritten in the following form:
\begin{equation}
\label{advupdate}
\Delta_i \Theta_i=\tilde{\xi}
\end{equation}
where \\
$\Delta_i= \left[ -\tilde{\psi},\tilde{\eta},\tilde{\nu}_i-\tilde{\nu} \right]$\\
$\Theta_i= \left[W^\top_i,\mathrm{vec}^\top\left(\Lambda_{xu(i)}\right), \mathrm{vec}^\top\left(\Lambda_{uu(i)}\right) \right]^\top$\\
$\tilde{\eta}= \left[ \tilde{\eta}_{0,1},\tilde{\eta}_{1,2}, \hdots, \tilde{\eta}_{N-1,N}  \right]^\top$,\\
$\tilde{\nu}_i= \left[ \tilde{\nu}_{0,1(i)},\tilde{\nu}_{1,2(i)}, \hdots, \tilde{\nu}_{N-1,N(i)}  \right]^\top$\\

Meanwhile, the policy improvement step of the proposed algorithm remains the same as \eqref{alg4_policy_improve}. Since Algorithm~\ref{Alg 5} provides an alternative implementation of Alg~\ref{Alg 4}, with the value evaluation step and the policy improvement step following the same principles as in Alg~\ref{Alg 4}, the convergence properties remain the same as those of Alg~\ref{Alg 5}.
\begin{algorithm}
    \caption{Time-Iterative Off-policy Q-learning subject to Input Constraint}\label{Alg 5}
\begin{algorithmic}[1]
\State Initiate admissible control policy gain $K_0$.
\State \textbf{loop}:
\State Collect trajectory data $  \tilde{\xi}, \tilde{\nu}, \tilde{\psi}, \tilde{\nu}_i, \tilde{\eta}$ from the system under excitation by the behavior policy $u$.
\State Let $\Delta_i= \left[ -\tilde{\psi},\tilde{\eta},\tilde{\nu}_i-\tilde{\nu} \right]$.
\State Get $W_i$, $\Lambda_{xu(i)}$ and $\Lambda_{uu(i)}$ from the expression following:
$ \Theta_i =   \left(\Delta_i^\top \Delta_i \right)^{-1}\Delta_i^\top \tilde{\xi}$.
\State Update the control policy $K_{i+1}=\frac{1}{2}\Lambda_{uu(i)}^{-1}\Lambda_{xu(i)}^\top $.
\State Let $u_{i+1}(x)= -\lambda \tanh\left(\frac{1}{\lambda}K_{i+1} \mathrm{vec}\left(\nabla\Phi^\top(x) \right)\right)$.
\State Let $i \leftarrow i+1$, $t \leftarrow t +\Delta t$.
\State \textbf{until} $\|W_i-W_{i-1}\| < \epsilon_t$, where $\epsilon_t$ is a small real positive number.
\State Return $V^*(x)= W_i^\top\Phi(x)$ and $u^*(x)=-\lambda \tanh\left(\frac{1}{\lambda}K_{i+1} \mathrm{vec}\left(\nabla\Phi^\top(x) \right)\right)$.
\end{algorithmic}
\end{algorithm}

\begin{remark}
    Algorithm~\ref{Alg 4} accumulates a larger volume of data in Step 3 compared to Algorithm~\ref{Alg 5} in the same system time, due to its higher-dimensional data representation. Consequently, Algorithm~\ref{Alg 5} demonstrates superior memory efficiency, making it more suitable for large-scale systems or deployment on resource-constrained microcontrollers. However, Algorithm~\ref{Alg 5} exhibits an inherent limitation: the collected data incorporate the noise term \( w(t)= u(t) - u_i(t) \), which may attenuate the amplitude of \( w(t) \) relative to \( u(t)\). Although the frequency components remain sufficiently rich, the diminished amplitude can induce numerical instability when employing least squares estimation, particularly when the control input constraints are tight. This trade-off should be meticulously evaluated when selecting an appropriate control strategy for a given application.
\end{remark}
Note that, to fully implement Algorithms~\ref{Alg 4} and~\ref{Alg 5} in a completely model-free fashion, a method is required to initialize an admissible gain without relying on the system dynamics model. This heuristic method is based on the intuitive assumption that:
\begin{assumption}\label{hypothesis}
    The stabilizing controller for the constrained input problem (Problem~\ref{Prob. 2}) will be close to the stabilizing controller for Problem~\ref{Prob 1} if the control signal remains within the constraint bounds.
\end{assumption}
Consider the case of the nominal LQT problem, which minimizes the following cost function:
\begin{equation}\label{nominal_cost}
    J=\int_0^T  e^{-\gamma (\tau-t)} \left(
    \tilde{Q}(X)+ u^\top Ru\right) d\tau
\end{equation}
subjected to the dynamic \eqref{aug_sys}. Note that $R$ is assumed to be a positive definite diagonal matrix and the dynamic \eqref{aug_sys} in this case can be rewritten in the following form:
\begin{equation}\label{nominal_sys}
    \dot{X}(t)=ZX(t)+Gu(t)
\end{equation}
where $X(t)=[x(t),x_r(t)]^\top $ and $Z=\begin{bmatrix}
    A&0\\
    0&A_r
\end{bmatrix}$, $G=[B, 0]^\top $ and the reward function can also be rewritten as $\tilde{Q}(X)= X^\top C^\top MCX$, $C=[I,-I]$. Under the assumption that the discount factor $\gamma$ can be chosen such that $\left(A_r-\gamma I\right)$ is stable, the solution for the tracking optimizing problem \eqref{nominal_cost}-\eqref{nominal_sys} can be obtained by solving the LQT Riccati \cite{Modares2014linear}. However, in this case, our interest lies in finding the admissible gain by solving the equation of the observability gramian \cite{rugh1996linear}: 
\begin{equation}\label{Gramian}
    \left(Z- \frac{1}{2}\gamma I\right)^\top S + S\left(Z- \frac{1}{2}\gamma I\right) + C^\top M C =0
\end{equation}
Note that the equation above can be satisfied if the pair $(Z - \gamma I, \sqrt{M}C)$ is observable \cite{rugh1996linear}. Since $(A, \sqrt{M})$ is observable by Assumption~\ref{Assumption 01}, it follows that $(A - \gamma I, \sqrt{M})$ is also observable. Furthermore, $\gamma$ can be chosen such that $\left(A_r - \gamma I\right)$ is stable, which ensures that $(Z -\gamma I, \sqrt{M}C)$ is observable. In \eqref{Gramian}, the solution $S$, which represents the Gramian observability, can also be interpreted as the solution to an optimal control problem with constrained input, as we will show later. Meanwhile, a data-driven method should be used to solve the equation efficiently in terms of computation. Specifically, using the approximation $\hat{\Gamma}_{xx}= Z^\top S + SZ$ in \eqref{LMIproof1} can lead to computational inefficiencies because the matrix $Z$ could be low-rank. Therefore, in some cases where resetting the state of the system is allowed, this issue can be resolved. In an establishment where the system state is reset randomly within a bounded open set containing the origin, setting the control input to zero ($u(t)=0, \forall t \geq 0 $) ensures that the following equation holds:
\begin{equation}
\label{Gramian modelfree1}
     \mathrm{vec}_{S}^\top \left(S \right)\sigma_{0}(t)=\mathrm{vec}_{S}^\top \left(\hat{\Lambda}_{xx} \right)\rho_{0}(t) 
\end{equation}
where $\sigma_{0}(t)=  X(t)\otimes_{S}X(t)-X(t_0)\otimes_{S}X(t_0)$,\\ 
$\rho_{0}(t)=\int_{t_0}^{t}  X(\tau)\otimes_{S}X(\tau)d\tau$.\\
And within the same setup, $M$ can also be approximated from the following equation:
\begin{equation}\label{Gramian modelfree2}
    \bar{\vartheta}= \mathrm{vec}_{S}^\top(M) \bar{\varsigma}
\end{equation}
where $\bar{\vartheta}= \left[ \vartheta^1_{0},\vartheta^2_{0}, \hdots, \vartheta^N_{0}  \right]^\top$, $\bar{\varsigma}= \left[ \varsigma^1_{0},\varsigma^2_{0}, \hdots, \varsigma^N_{0}  \right]^\top$ and $\vartheta^i_0 = \int_{t_0}^{t} r(X(\tau),u(\tau)) d\tau$, $\varsigma^i_0 = \int_{t_0}^{t} \tilde{x}(\tau) \otimes_S  \tilde{x}(\tau) d\tau$ are the data collected from different trajectories when resetting state. Resetting the state of the system multiple times can be assumed that the data term $\bar{\varsigma}$ has the full column rank, and $\hat{M}$ can be uniquely determined.
Hence, using \eqref{Gramian modelfree1} and \eqref{Gramian modelfree2}, the equation \eqref{Gramian} can be rewritten as follows:
     \begin{equation}\label{Gramian modelfree3}
     \begin{aligned}
         \mathrm{vec}_{S}^\top \left(S \right)\sigma_0(t) & - \gamma\mathrm{vec}_{S}^\top (S) \rho_0(t) \\
         & = - \mathrm{vec}_{S}^\top \left(C^\top \hat{M}C \right)\rho_0(t) 
     \end{aligned}
\end{equation}
Define\\
$\chi= \bar{\sigma} - \gamma\bar{\rho}, \quad \bar{\sigma}= \left[ \sigma^1_{0},\sigma^2_{0}, \hdots, \sigma^N_{0}  \right]^\top ,\quad \bar{\rho}= \left[ \rho^1_{0},\rho^2_{0}, \hdots, \rho^N_{0}  \right]^\top$.\\
Suppose that the data term $\chi$ is full column rank by resetting system state at random in a bounded open set containing the origin, $S$ can be uniquely determined from the following:
      \begin{equation} \label{LMI_cons2}
 \mathrm{vec}_S(S)  = - \left( \chi^\top \chi \right)^{-1}\chi^\top\mathrm{vec}_{S}^\top \left(C^\top \hat{M}C \right)\bar{\rho} 
  \end{equation}
Note that if we solve the \eqref{Gramian}, it implies that we also solve the following Lyapunov equation:
\begin{equation}\label{LQT}
    \begin{aligned}
     X^\top  Z^\top  S X  & + X^\top SZX -\gamma X^\top SX +\\
    &+X^\top(C^\top MC + SBR^{-1}B^\top S)X \\
    & -X^\top SBR^{-1}B^\top SX = 0   
    \end{aligned}
\end{equation}
This implies that $S$ is the solution for the value function of the LQT problem with a cost function weighted by $R$ and $Q_I= \left(C^\top MC + B^\top R^{-1}B\right)$. In addition, if the weight $R$ for the value function can be obtained, the admissible gain can also be calculated as $\hat{K}=\hat{B}^\top S$, where $\hat{B}$ is calculated using the equation \eqref{constraintLMI2_b}.
Similarly, the same derivation can be applied to the LQT problem with constrained input as follows:
\begin{equation}\label{cons_quadratic}
    \begin{aligned}
        X^\top  Z^\top  S X + X^\top SZX & -\gamma X^\top SX + \\ & + \tilde{Q}_I(X) + D(X)  =0
\end{aligned}
\end{equation}
where \\
$\tilde{Q}_I(X) = X^\top C^\top MC X - D(X)\geq 0$ \\
$D(X)=\lambda^2 \bar{R} \ln \left(1 - \tanh^2\left(\frac{1}{2\lambda}R^{-1}B^\top SX \right) \right) \leq 0$.\\
This indicates that \( V^*(X) = X^\top S X \) is the solution to the optimal control problem with some non-linear cost function that includes an additional nonlinear term $D(X)$ in $\tilde{Q}(X)$, which also ensures the stability of the system if $D(X)$ is bounded. However, if the solution \( V^*(X) \) of the optimal HJB equation \eqref{opt_cons} is not in quadratic form, the equality \eqref{cons_quadratic} will be modified accordingly.
\begin{equation}
    \begin{aligned}
        &X^\top  Z^\top  S X + X^\top SZX -\gamma X^\top SX + \tilde{Q}_I(X)\\
        &\quad+\lambda^2 \bar{R} \ln \left(1 - \tanh^2\left(-\frac{1}{2\lambda}R^{-1}B^\top SX \right) \right)  =E(X)
\end{aligned}
\end{equation}
where $E(X)$ is some arbitrary nonlinear function. Then, if we consider the quadratic solution \( \hat{V}(X) = X^\top S X \) as an approximation of the non-quadratic optimal solution \( V^* (X)\), and if \( E(X) \) is bounded, it can be inferred that there might exist a bounded nonlinear function \( \epsilon_I(X) \) such that:
\begin{equation}\label{optimal_HJB_robust_cons}
\begin{aligned}
    & \tilde{Q}_I(X)- \gamma V^*(X) + (\nabla V^*(X))^{\top} ZX
\\
&+ \lambda^2 \bar{R} \ln\left(1 - \tanh^2\left(\frac{1}{2}(\lambda R)^{-1} G^\top  \nabla V^* \right)\right) = 0    
\end{aligned}
\end{equation}
where $V^*(X)=\hat{V}(X)+\epsilon_I(X)$.
Equation \eqref{optimal_HJB_robust_cons} suggests that $
\hat{V}(X)$ is the projection of the optimal value function $V^*(X)$ onto the subspace of quadratic functions of the state variable $X$. This weakly supports the assumption~\ref{hypothesis}; however, this assumption~\ref{hypothesis} only holds conditionally on the assumptions mentioned above, such as $E(X)$ being bounded and the existence of a bounded function $\epsilon_I(X)$. Based on the hypothesis, after obtaining $S$ by solving equation \eqref{Gramian modelfree3}, the solution $\hat{V}(X)$ can be considered an approximation to the optimal solution $V^*(X)$ of \eqref{optimal_HJB_robust_cons} in the subspace of quadratic functions of $X$. Despite that, the optimal solution still needs to be searched for around the hypothesis solution $\hat{V}(X)$ through trial and error, due to the fact that some of the conditions for the hypothesis to hold may not be satisfied, in general. Also, even if the choice of activation function for the cost function differs from the quadratic basis, this method can still be leveraged. Since we can project the hypothesis solution $\hat{V}(X)$ onto different subspaces depending on the choice of activation function by using the least squares method, the same search procedure can be applied around the projected solution through trial and error. Although this heuristic search method requires trials and errors, it narrows down the search space for the admissible gain.

\section{Simulation Results}\label{Section 4}         
\subsection{Simulation Setup}
The feasibility of the off-policy Algorithm \ref{Alg 2}, the time-iterative Q-learning Algorithm \ref{Alg 3} and
off-policy Q-learning for the optimal tracking problem with constraint input (Algs. \ref{Alg 4}, \ref{Alg 5}) will be demonstrated by simulations. Furthermore, in the simulation, the on-policy Q-learning algorithms (Algorithm \ref{Alg 1}) \cite{Possieri} and the Online Q-learning algorithms \cite{VAMVOUDAKIS201714} will also be considered to compare with the proposed algorithms in the setup where the state is not allowed to reset during training, as stated in the previous section. An example of a F-16 aircraft autopilot is applied to verify the effectiveness of these algorithms. In light of \cite{f16original}, \cite{dao2025h}, an F-16 model is given with the vector of the three variables $x=[\alpha ,\omega ,\sigma_e ]^T$ driven by the dual engine, where $\alpha ,\omega ,\sigma_e $ are the angle of attack, the pitch rate and the inclination of the elevator, respectively. For feedback control problems, Algorithms \ref{Alg 2}, \ref{Alg 3}, \ref{Alg 4}, and \ref{Alg 5} were trained over episodes $N_{train} = 39$, collecting data for $T = 1 (s)$ in each episode. Furthermore, Algorithms \ref{Alg 4} and \ref{Alg 5} were used to track problems and trained over $N_{train} = 39$ episodes, collecting data for $T = 2.5 (s)$ seconds in each episode. Meanwhile, the behaviour policy used for exploration in this simulation is defined as the sum of sinusoidal signals with random frequencies.

\begin{equation}
    u(t) = \bar{C} \sum_{i=1}^{100} \sin(\omega_i t)
\end{equation}
where $\omega_i $  with $i = 1, \dots, 100$, are randomly selected from $[-500, 500]$, $\bar{C}$ is set to $\frac{50}{15}$ for the feedback case and $\frac{50}{2}$ for the tracking case.

The simulation is implemented on a PC equipped with a six-core 2.50 GHz Intel i5-12400F CPU, 16.0 GB RAM, 64-bit Windows 11 operating system, and Julia 1.10.2.  Open-source implementations of the proposed methods are also publicly available \footnote{\url{https://github.com/Huy-Quang-Dao/OffQLCTS.jl}}.

\subsection{F-16 feedback control problem}
In this simulation case, we consider a linear F16 aircraft plant, where the system dynamics is described by \eqref{sys1}, and the system matrices given by \cite{dao2025h,f16original}:
\begin{equation}
    A=\begin{bmatrix}
-1.01887 & 0.90506 & -0.00215 \\
0.82225 & -1.07741 & -0.17555 \\
0 & 0 & -1
    \end{bmatrix};
    B=\begin{bmatrix}
0 \\
0 \\
1
    \end{bmatrix}
\end{equation}
The performance index is considered with $M=\mathrm{diag}(1,1,1)$ and $R=I$. The Algorithms \ref{Alg 2} and \ref{Alg 3} are verified, along with two other methods: Q-learning on policy \cite{Possieri} and Q-learning in policy \cite{VAMVOUDAKIS201714} to the F-16 model without input constraints. In addition, Algorithms \ref{Alg 4} and \ref{Alg 5} with constrained input \(\lambda = 1.5\) are also taken into account in this case to demonstrate that these algorithms can solve Problem \ref{Prob. 2} without considering tracking control. Subsequently, for the linear system model of the F-16 aircraft, a natural choice for the activation function in the neural network \eqref{network} is $\Phi(x(t))=x(t)\otimes_S x(t)$, as this basis allows the representation of $\hat{V}(x)= x^T\hat{P}x$. 

\begin{figure}
    \centering
    {\includegraphics[width=3.3in]{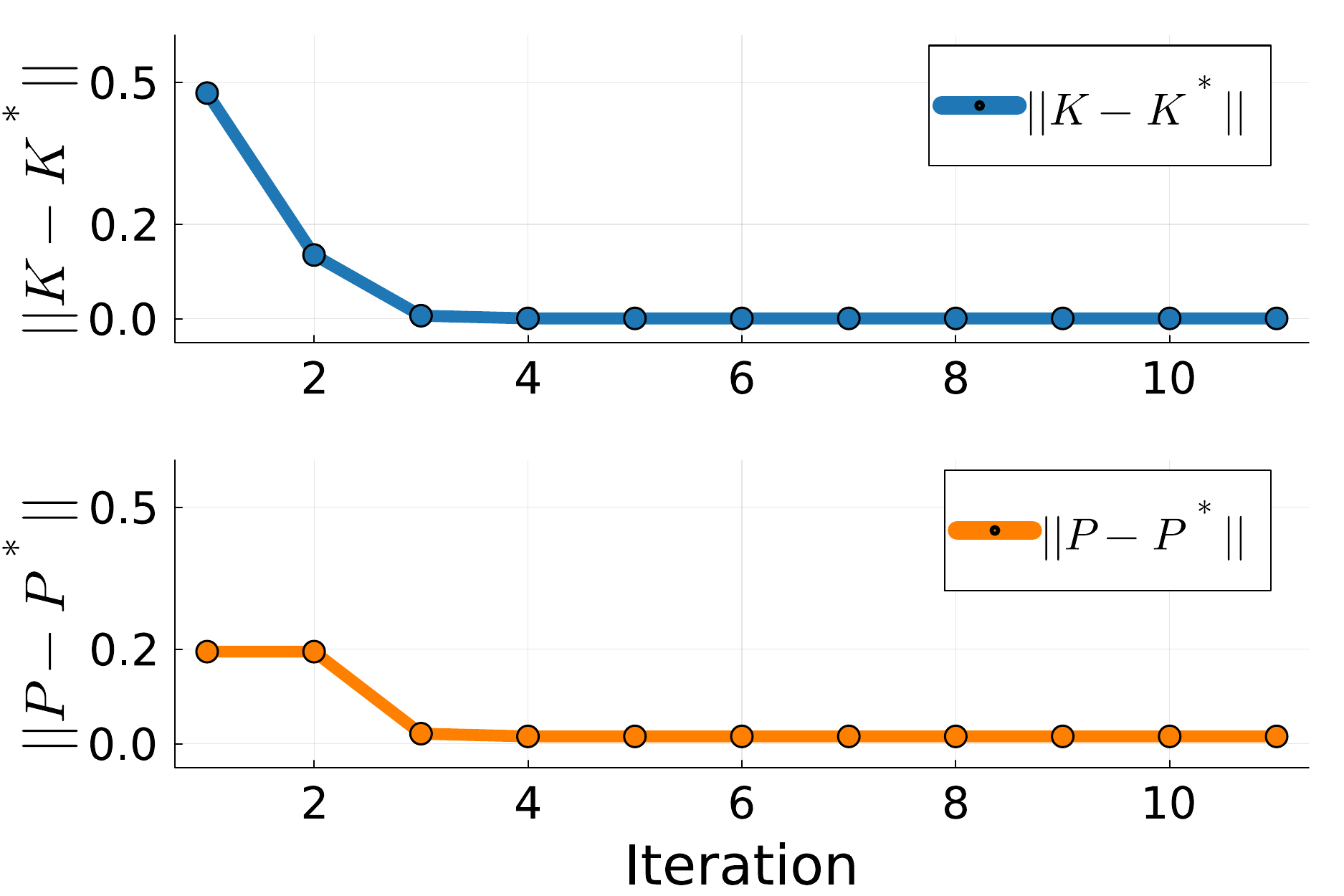}}
    \caption{ $P^i$ and $K_1^i$ using Alg~\ref{Alg 2} for the LQR case.}
    \label{fig:F16_Al2_KP}
\end{figure}

\begin{figure}
    \centering
    {\includegraphics[width=3.3in]{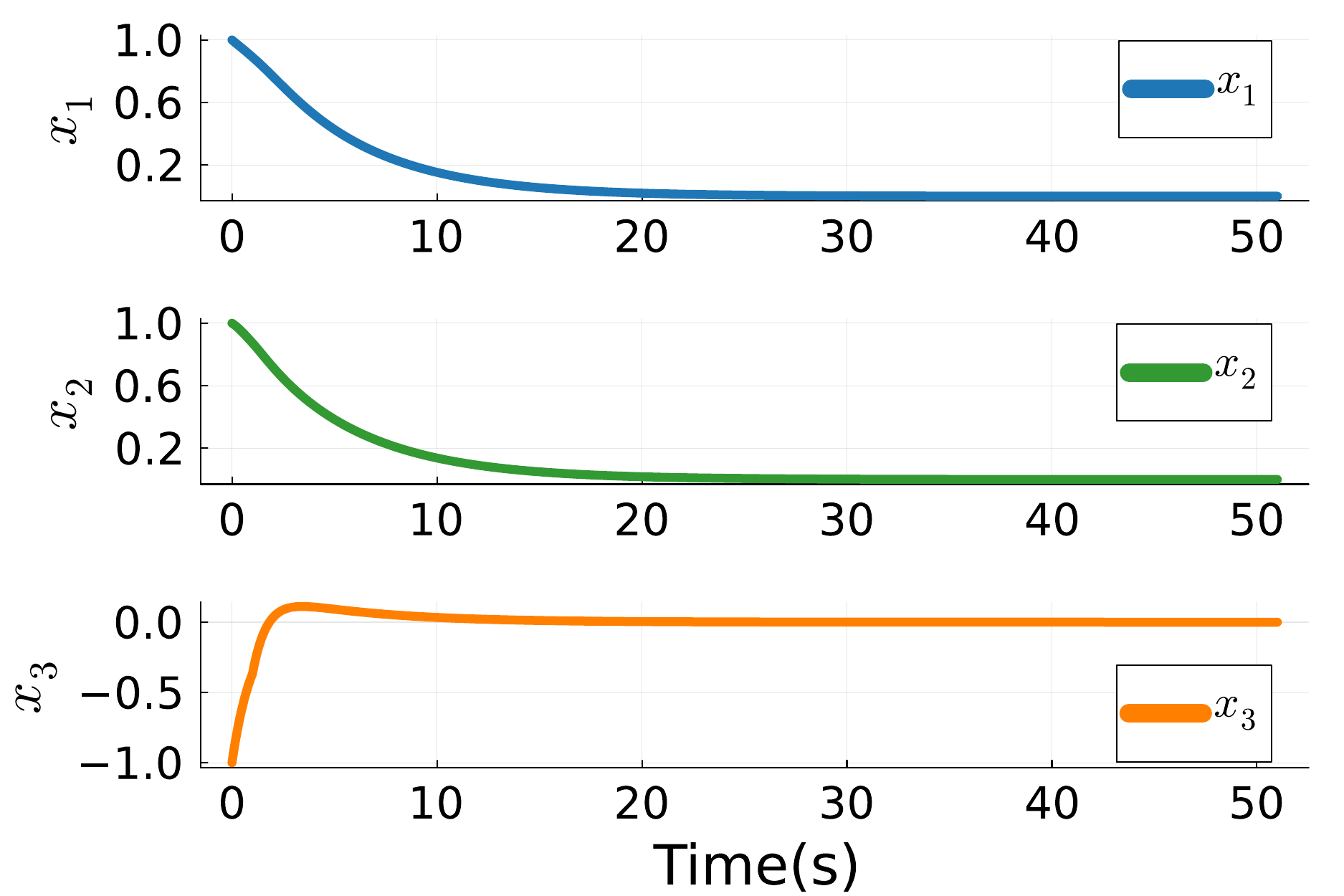}}
    \caption{ F-16 aircraft system state using Alg~\ref{Alg 2} for the LQR case.}
    \label{F16_Al2_state}
\end{figure}

\begin{figure}
    \centering
    {\includegraphics[width=3.3in]{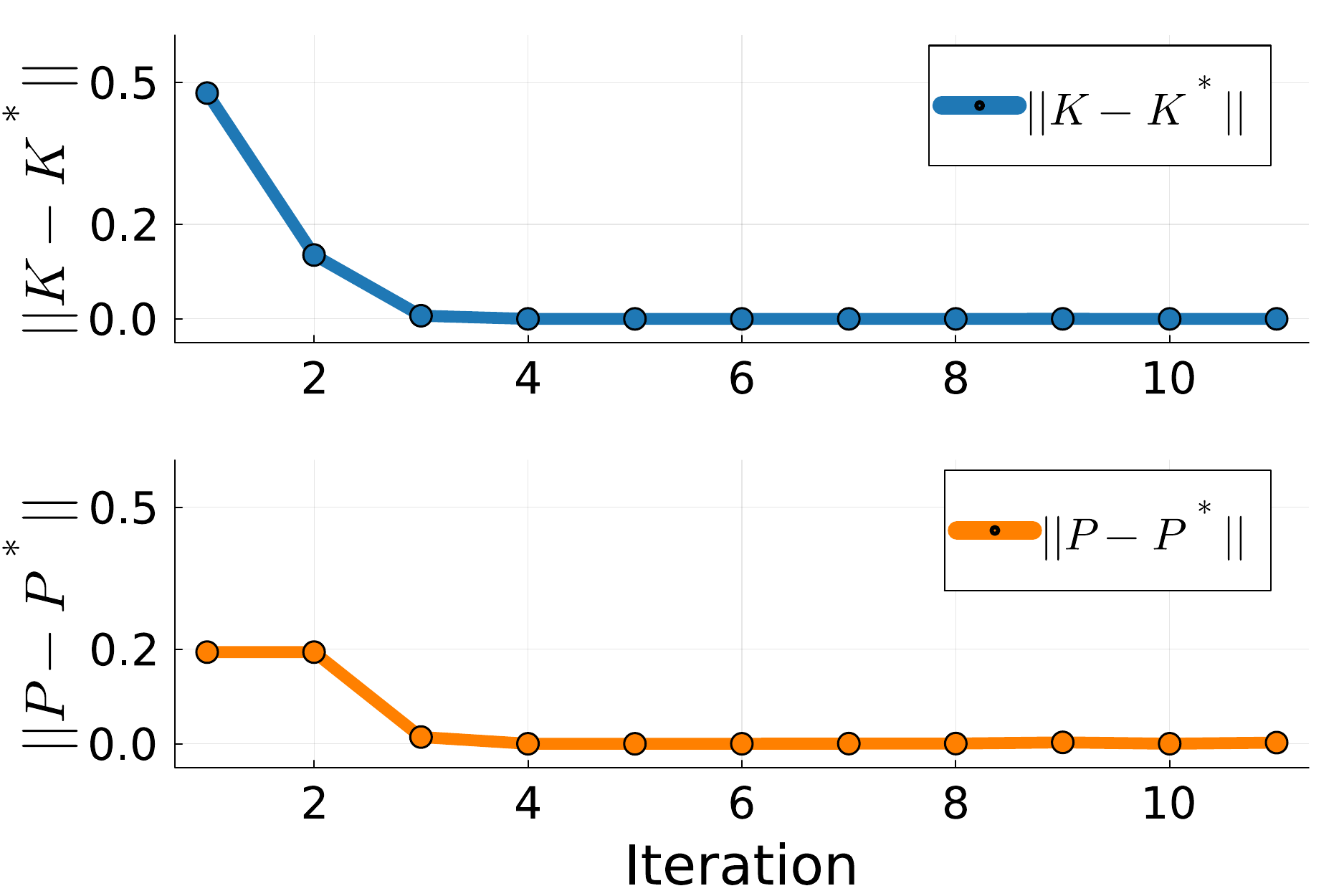}}
    \caption{ $P^i$ and $K_1^i$ using Alg~\ref{Alg 3} for the LQR case.}
    \label{fig:F16_Al3_KP}
\end{figure}

\begin{figure}
    \centering
    {\includegraphics[width=3.3in]{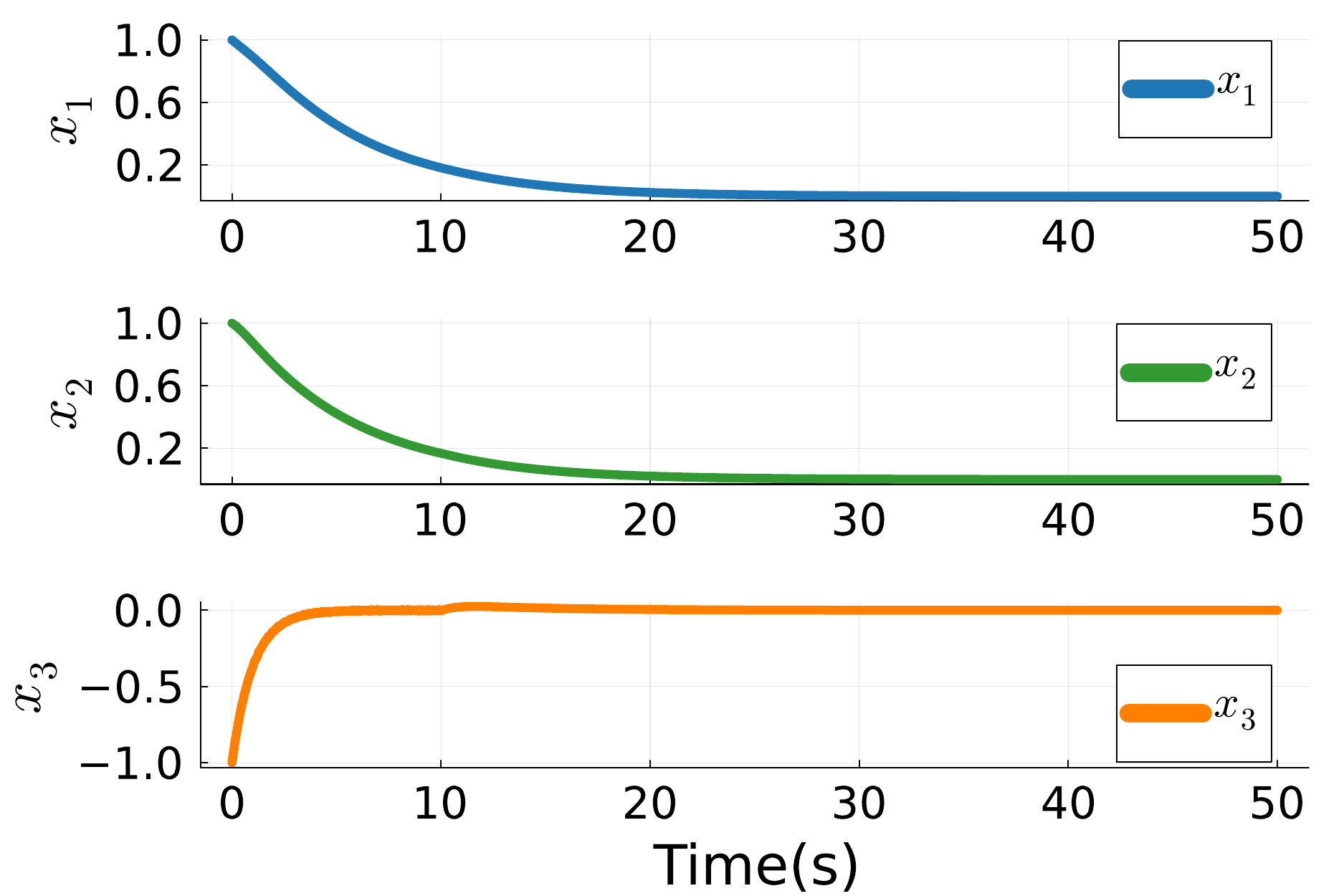}}
    \caption{F-16 aircraft system state using Alg~\ref{Alg 3} for the LQR case.}
    \label{F16_Al2_state}
\end{figure}

\begin{figure}
    \centering
    {\includegraphics[width=3.3in]{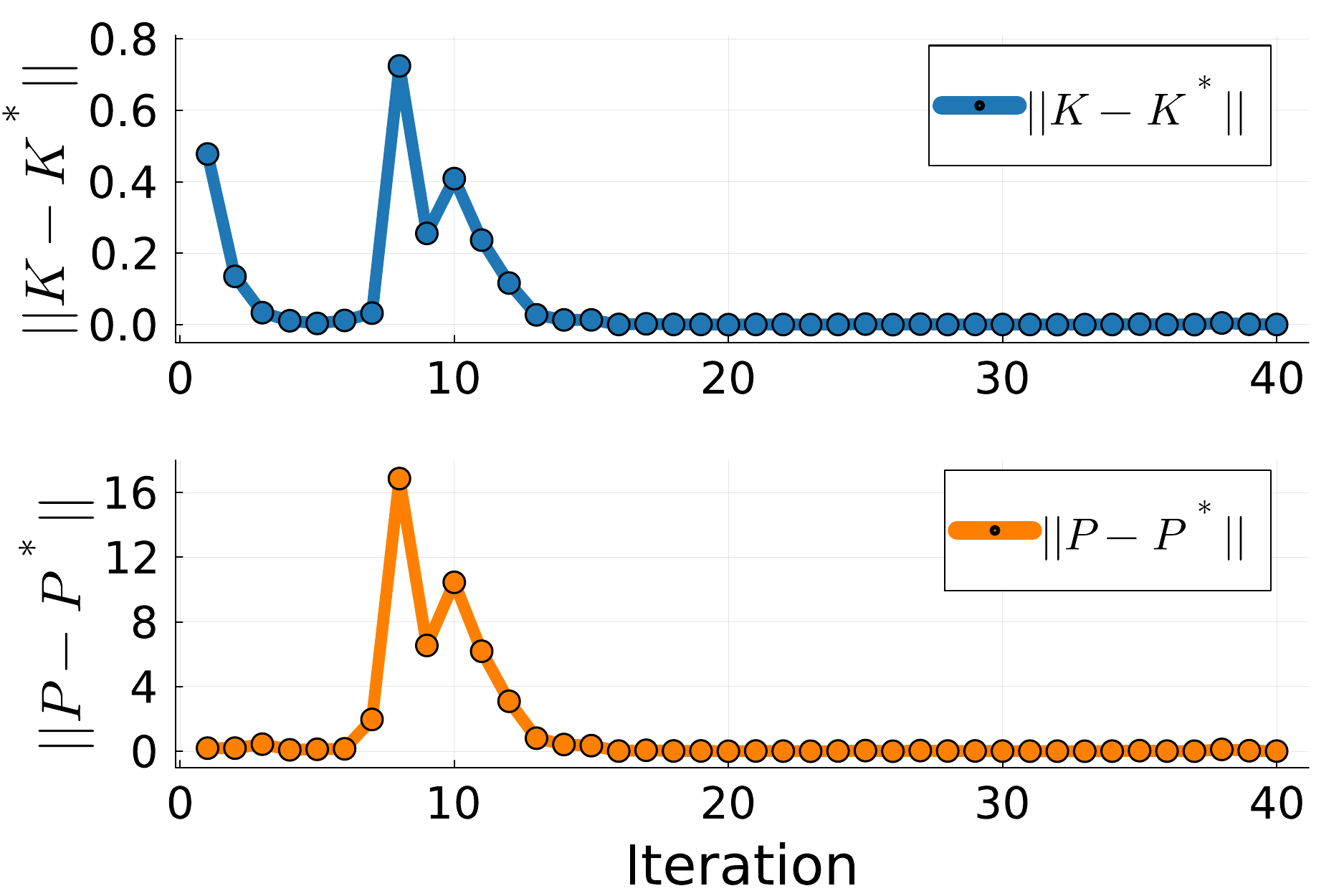}}
    \caption{ $P^i$ and $K_1^i$ using Alg~\ref{Alg 4} for the LQR case.}
    \label{fig:F16_Al4con_KP}
\end{figure}

\begin{figure}
    \centering
    {\includegraphics[width=3.3in]{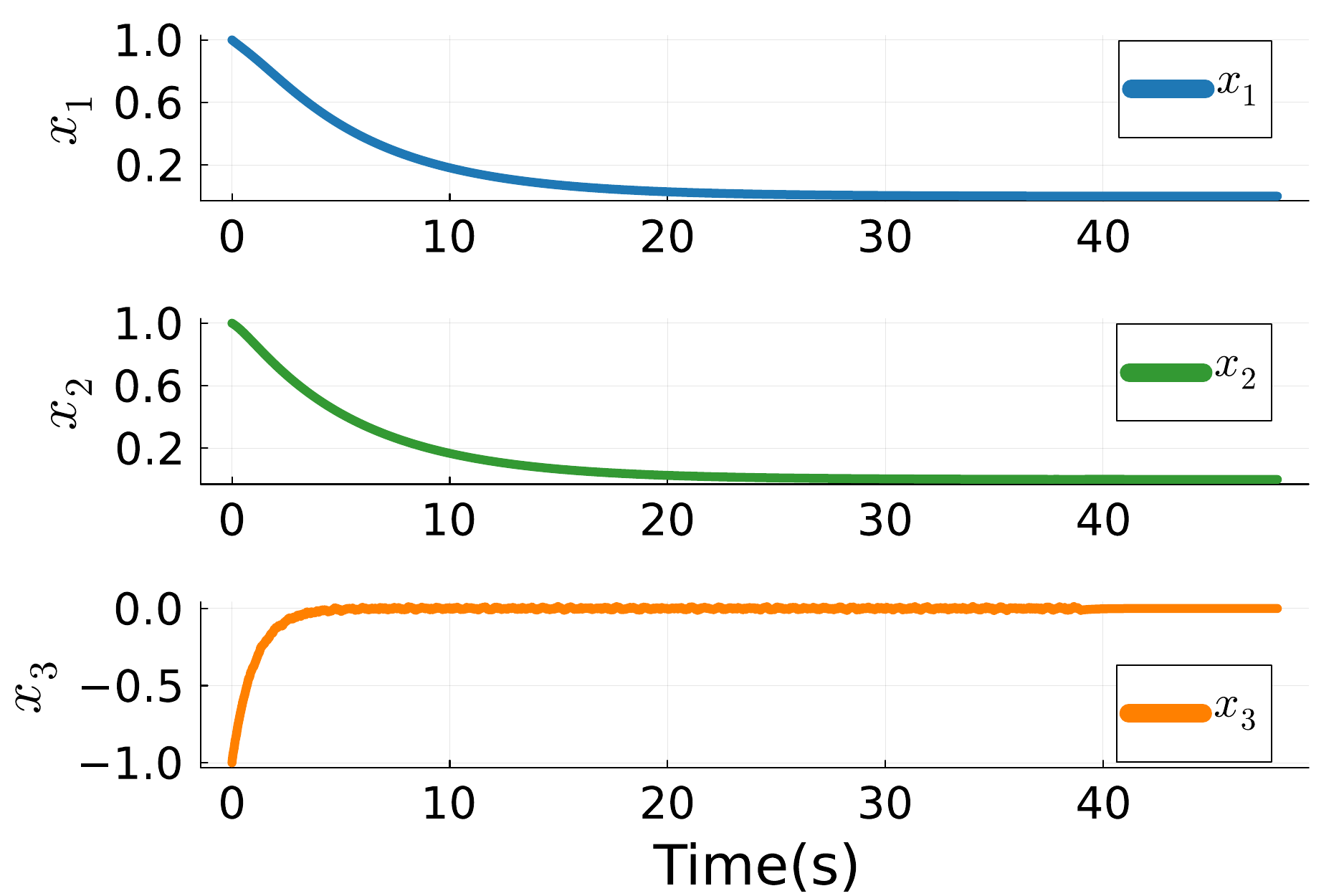}}
    \caption{F-16 aircraft system state using Alg~\ref{Alg 4} for the LQR case.}
    \label{F16_Al4con_state}
\end{figure}

\begin{figure}
    \centering
    {\includegraphics[width=3.3in]{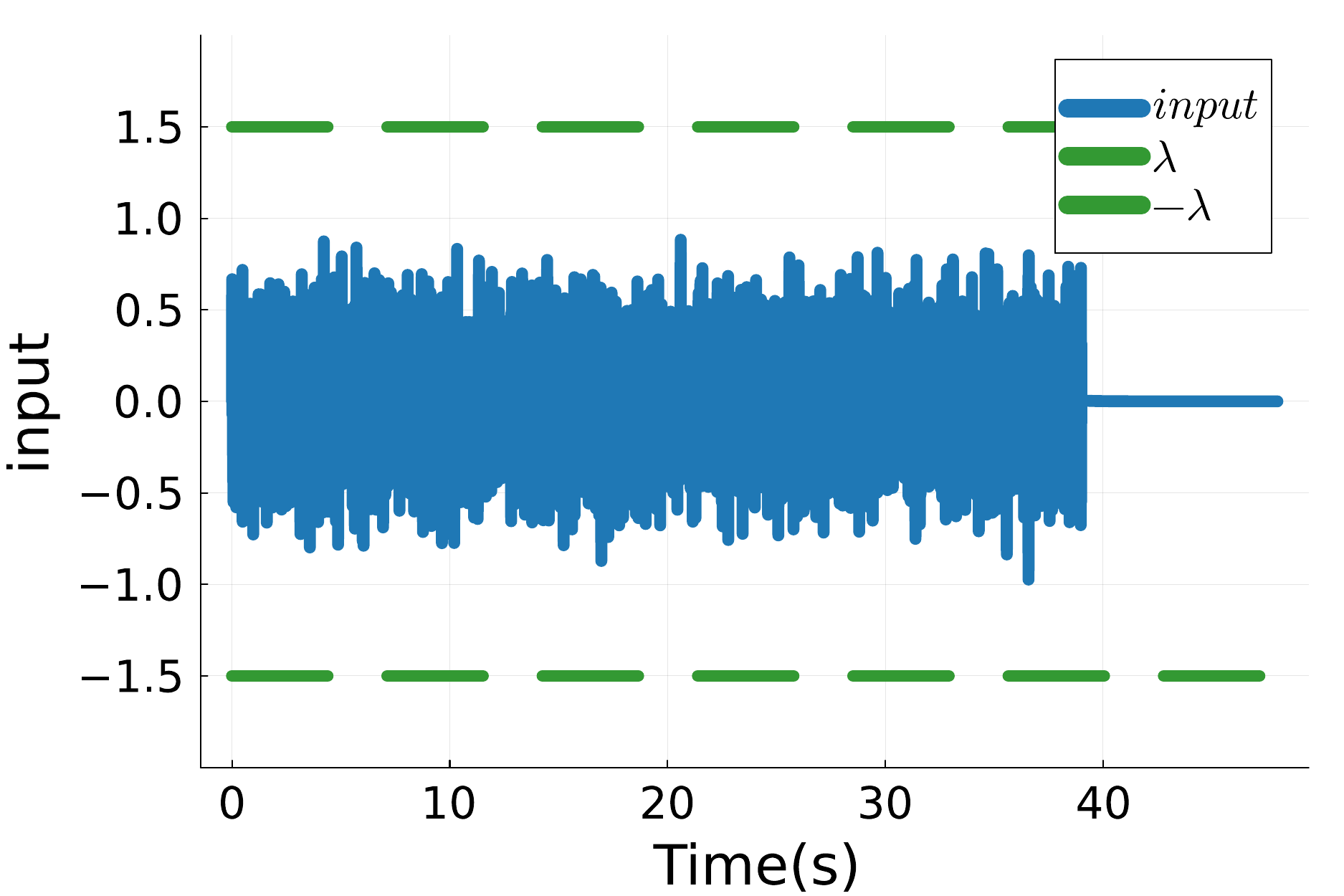}}
    \caption{Input of model using Alg~\ref{Alg 4} for the LQR case.}
    \label{F16_Al4con_input}
\end{figure}

\begin{figure}
    \centering
    {\includegraphics[width=3.5in]{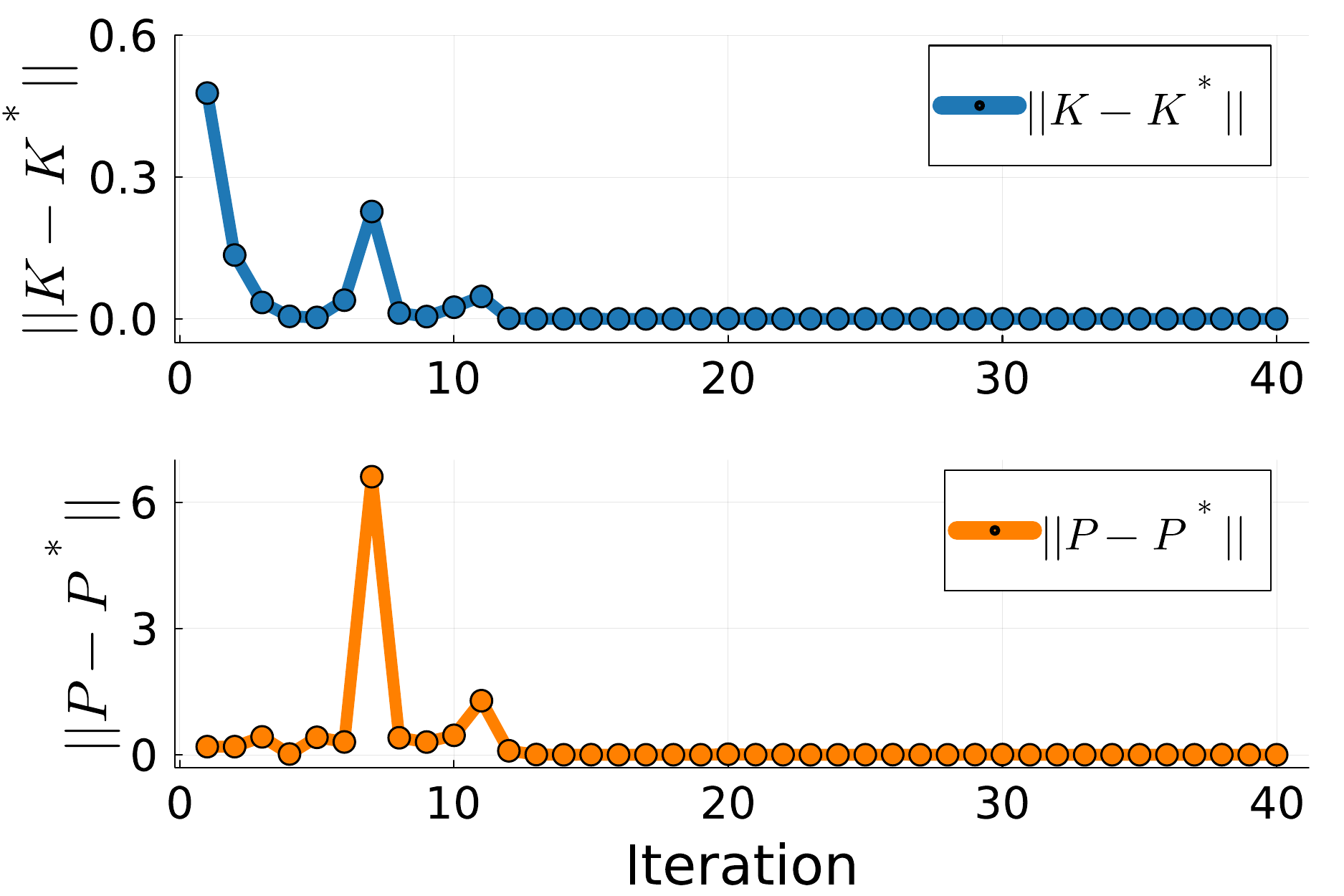}}
    \caption{ $P^i$ and $K_1^i$ using Alg~\ref{Alg 5} for the LQR case.}
    \label{fig:F16_Al5con_KP}
\end{figure}

\begin{figure}
    \centering
    {\includegraphics[width=3.5in]{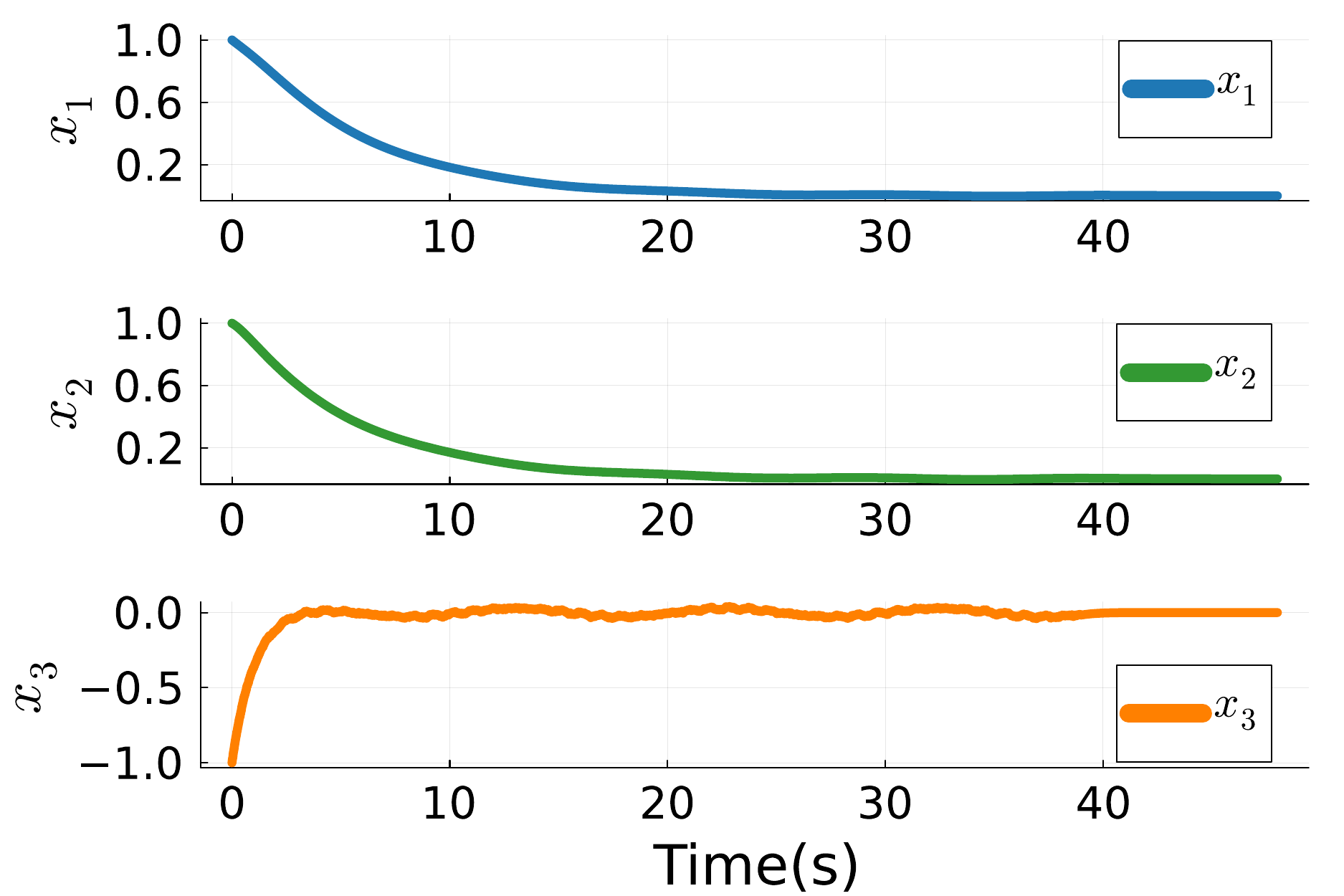}}
    \caption{F-16 aircraft system state using Alg~\ref{Alg 5} for the LQR case.}
    \label{F16_Al5con_state}
\end{figure}

\begin{figure}
    \centering
    {\includegraphics[width=3.3in]{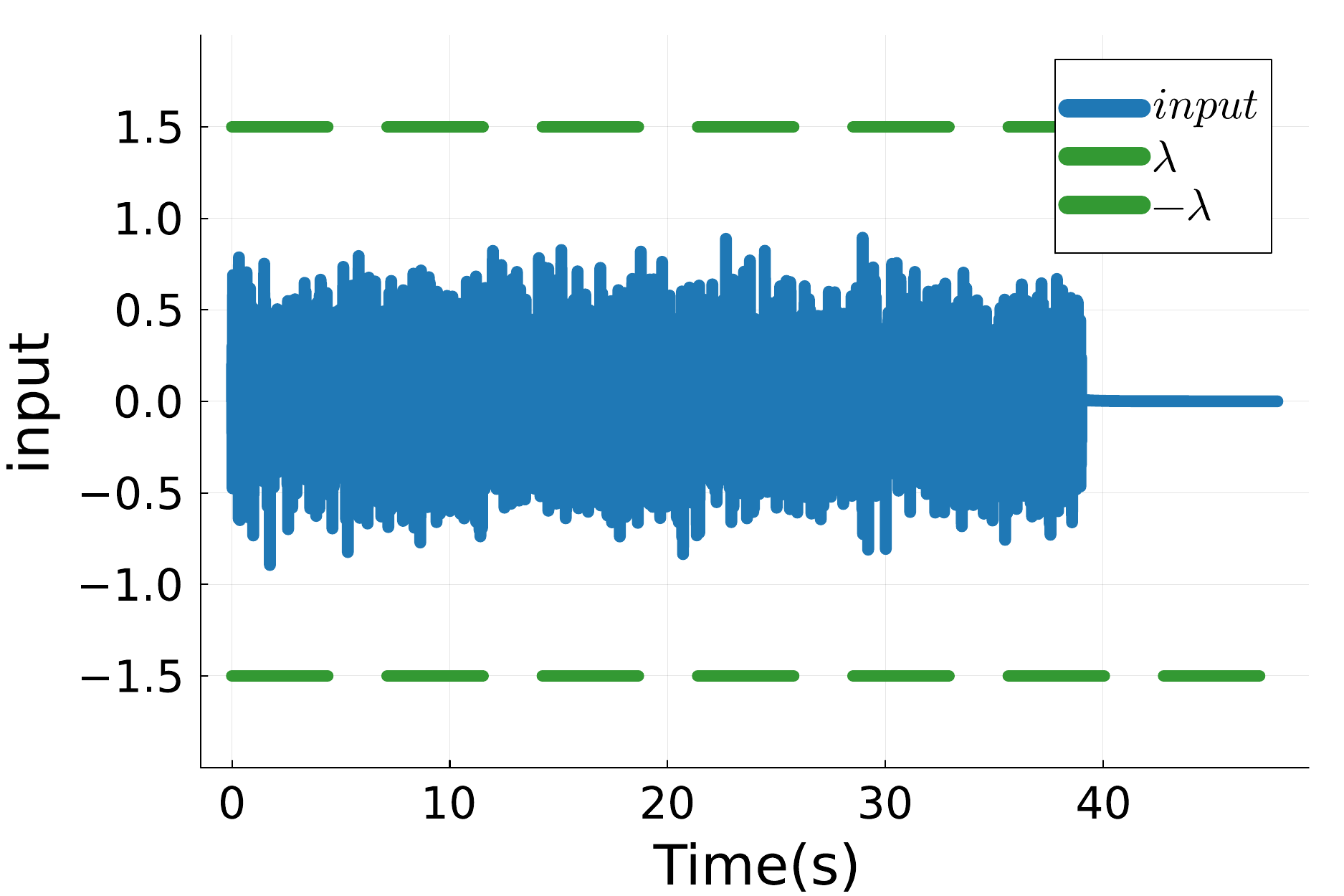}}
    \caption{Input of model using Alg~\ref{Alg 5} for the LQR case.}
    \label{F16_Al5con_input}
\end{figure}

\begin{figure}
    \centering
    {\includegraphics[width=3.3in]{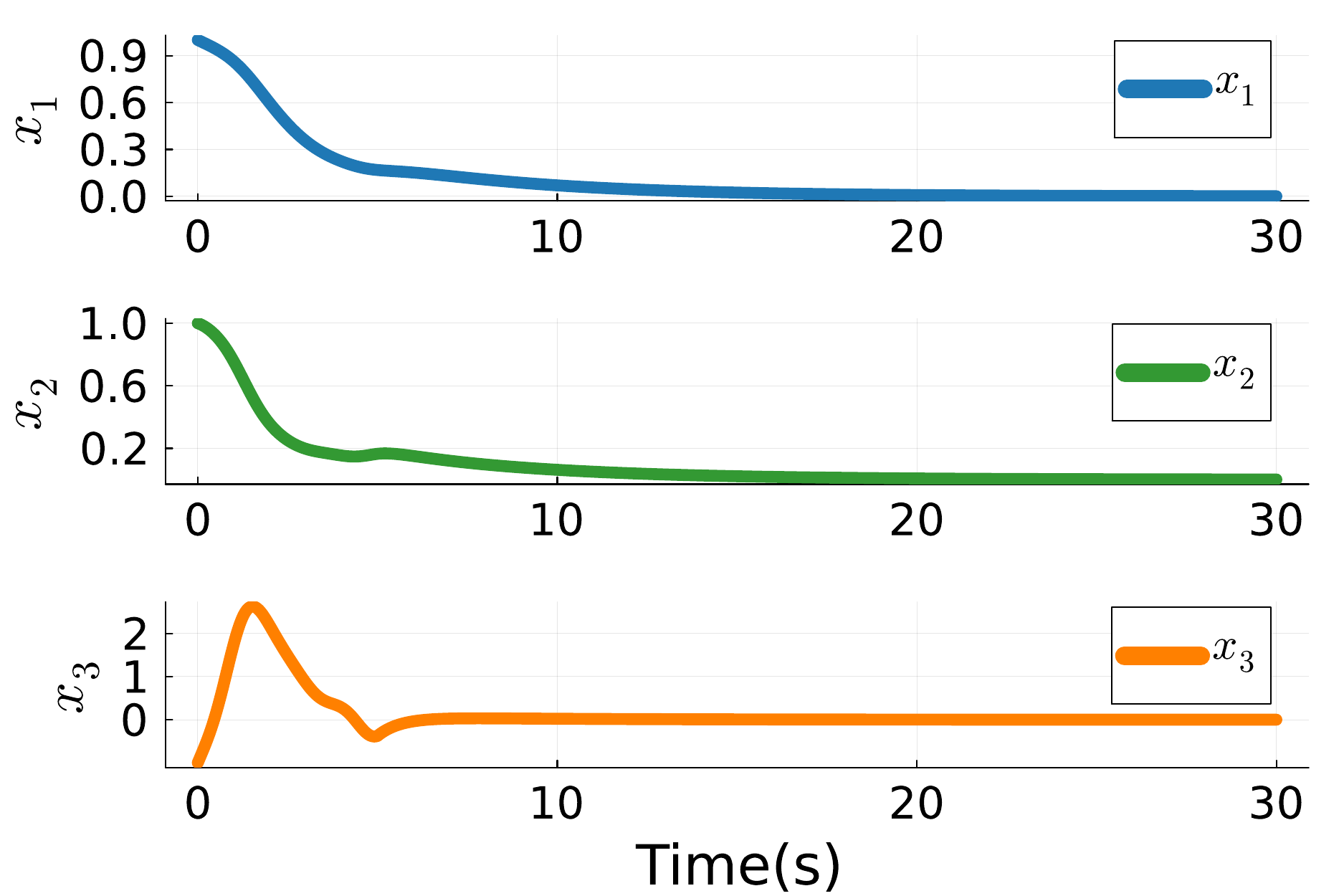}}
    \caption{F-16 aircraft system state using Online Q-learning \cite{VAMVOUDAKIS201714} for the LQR case.}
    \label{fig:comparation_3_state}
\end{figure}

\begin{figure}
    \centering
    {\includegraphics[width=3.3in]{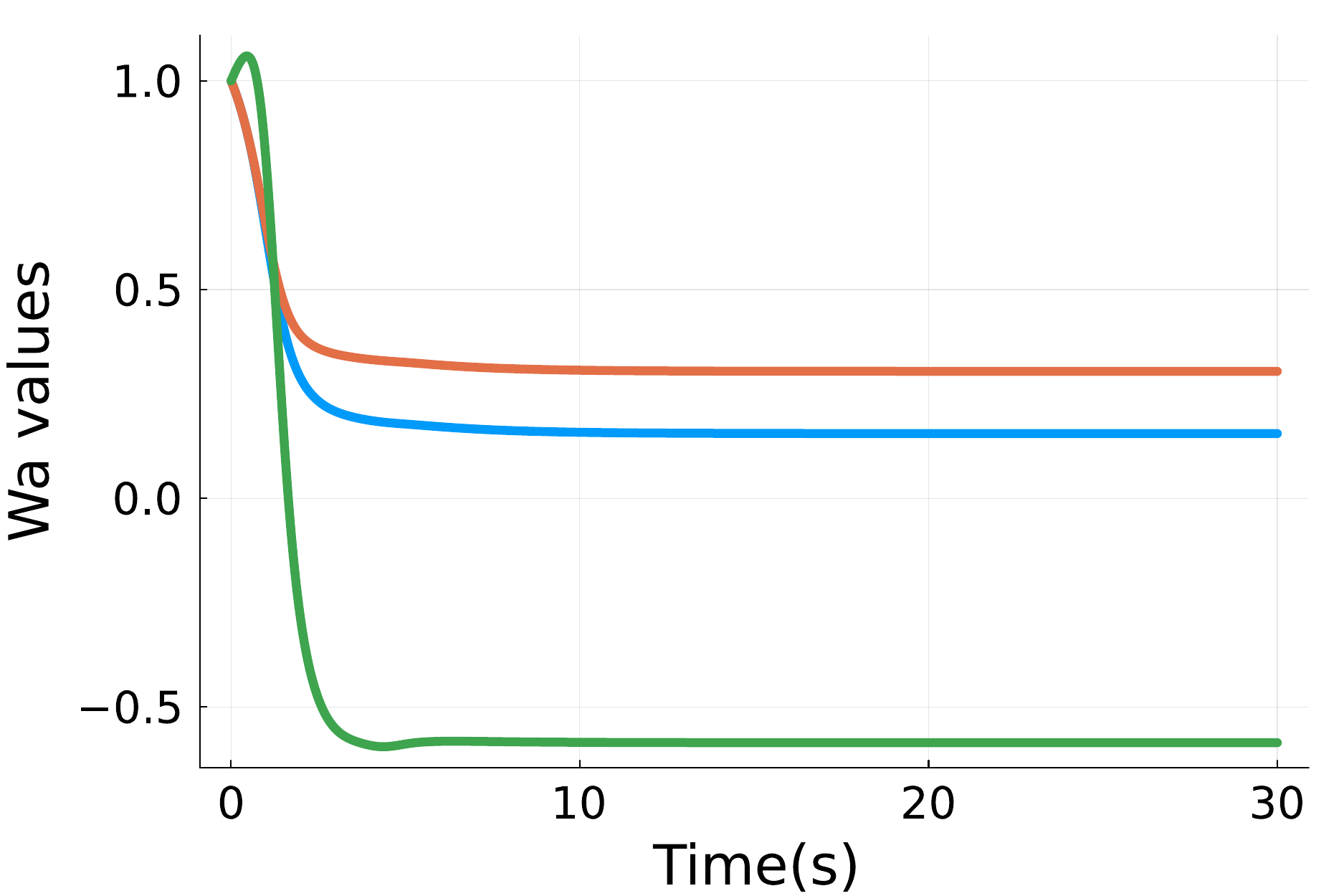}}
    \caption{ Actor weight of Online Q-learning \cite{VAMVOUDAKIS201714} for the LQR case.}
    \label{fig:comparation_3_Wa}
\end{figure}

\begin{figure}
    \centering
    {\includegraphics[width=3.3in]{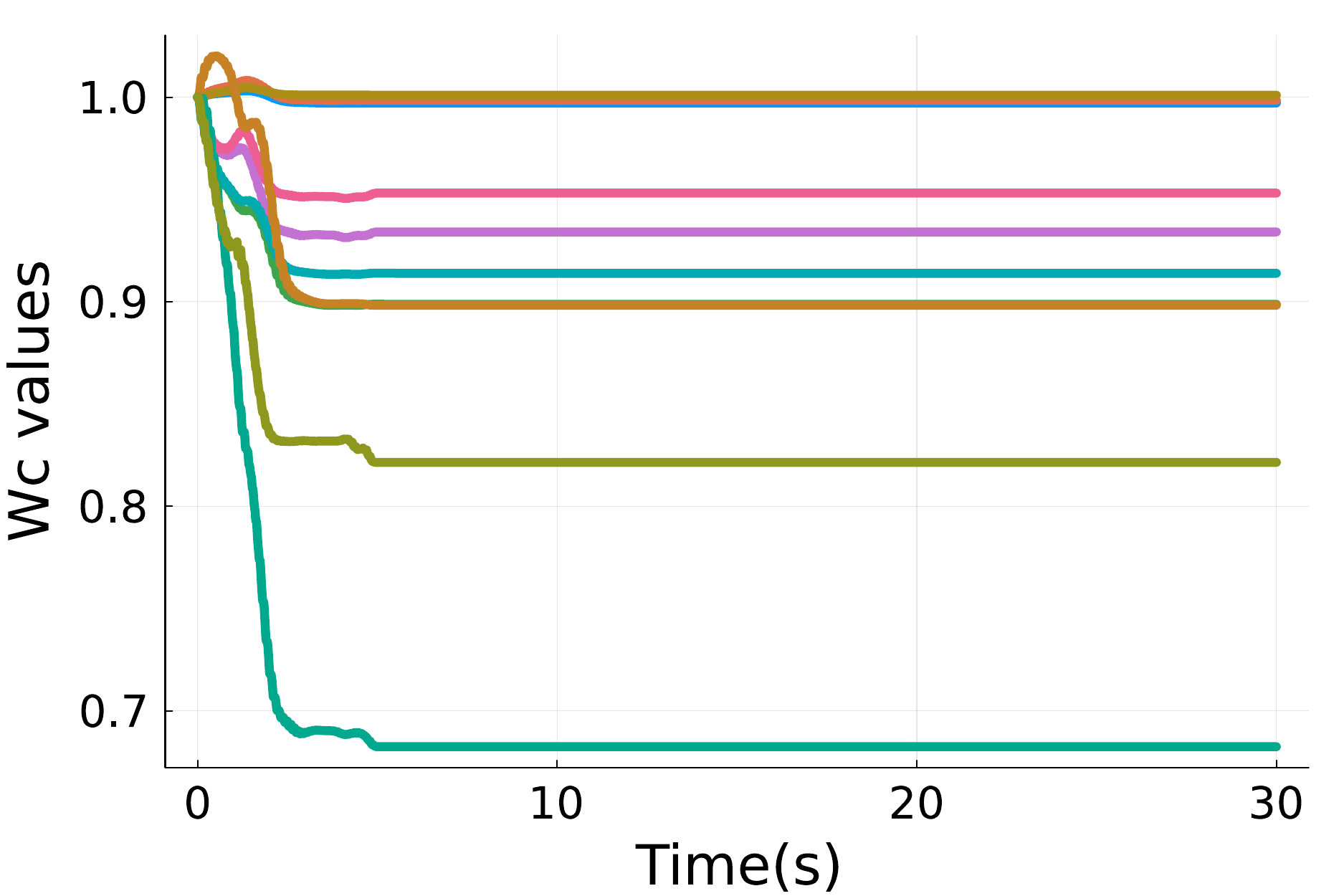}}
    \caption{ Critic weight of Online Q-learning \cite{VAMVOUDAKIS201714} for the LQR case.}
    \label{fig:comparation_3_Wc}
\end{figure}

\begin{figure}
    \centering
    {\includegraphics[width=3.3in]{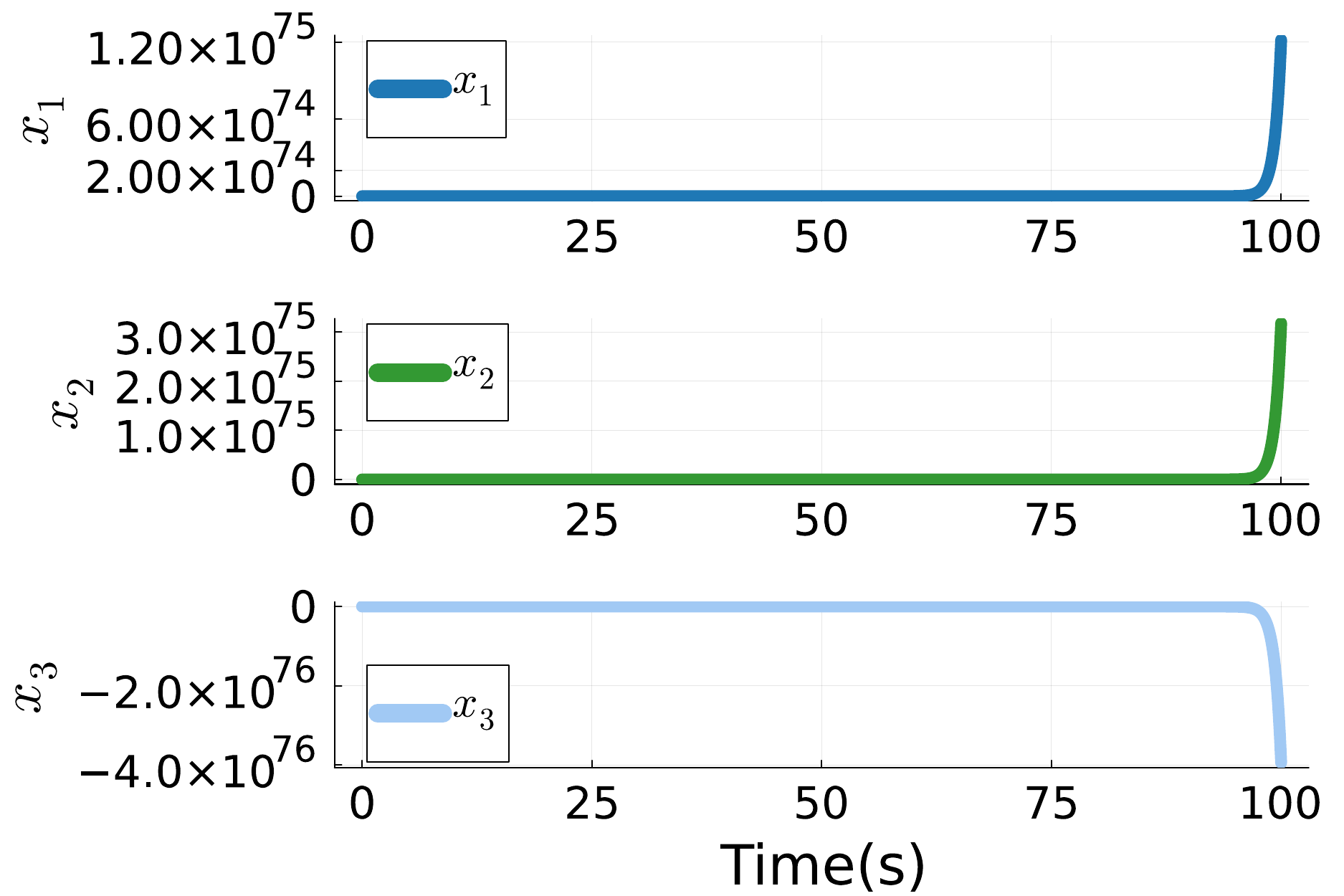}}
    \caption{F-16 aircraft system state using On-policy Q-learning \cite{Possieri} for the LQR case.}
    \label{fig:comparation_4_state}
\end{figure}

\begin{figure}
    \centering
    {\includegraphics[width=3.3in]{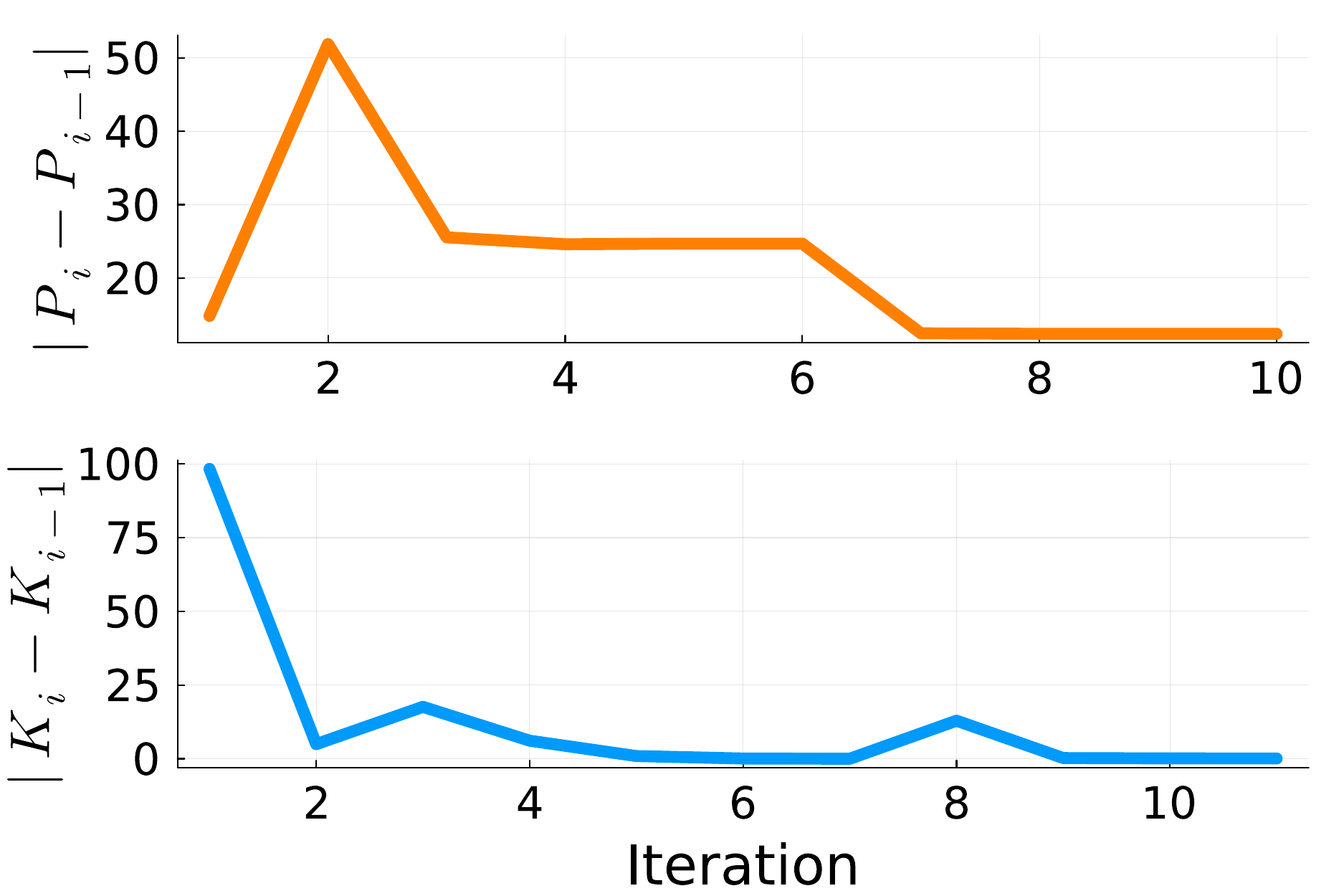}}
    \caption{ $P^i$ and $K_1^i$ using On-policy Q-learning for the LQR case.}
    \label{fig:comparation_4_KP}
\end{figure}

The weights of the value function and the policy for Algorithms~\ref{Alg 2} and~\ref{Alg 3}, shown in Figures~\ref{fig:F16_Al2_KP} and~\ref{fig:F16_Al3_KP}, demonstrate that both algorithms converge in approximately four iterations, with a comparable number of iterations required. Although both methods achieve convergence in a similar number of iterations, as they are merely different implementations of the Kleinman algorithm \textcolor{blue}{\cite{Kleinman}}, Alg~\ref{Alg 1} is more data efficient. Specifically, it collects data within a single second, whereas Alg~\ref{Alg 2} requires four seconds. Both algorithms exhibit higher accuracy and faster convergence compared to the benchmark methods in \cite{Possieri} and \cite{VAMVOUDAKIS201714}. As observed in Figures~\ref{fig:comparation_3_state}-\ref{fig:comparation_4_KP}, the convergence time for Online Q-learning is 6 seconds, while On-policy Q-learning takes 9 seconds. However, both methods exhibit bias in the weights of the value function and the policy. Furthermore, Figure~\ref{fig:comparation_4_state} highlights the instability of the system, indicating that the On-policy method is unsuitable when exploration noise is introduced during training. Furthermore, while Algorithms~\ref{Alg 4} and~\ref{Alg 5} require a longer convergence time, as shown in Figures~\ref{fig:F16_Al4con_KP} and~\ref{fig:F16_Al5con_KP}, they ensure system stability (Figures~\ref{F16_Al4con_state} and \ref{F16_Al5con_state}) and maintain control input within specified constraints (Figures~\ref{F16_Al4con_input} and \ref{F16_Al5con_input}).

\subsection{F-16 tracking control problem}
In this case, the matrices $A$ and $B$ of the F-16 aircraft model are defined similarly to the feedback control problem, and the reference signal $\dot{x}_r(t)=A_rx_r(t)$ is given as:
\begin{align}
    \dot{x}_r(t)=\begin{bmatrix}
-2 & 0 & 0 \\
0 & -2 & 0 \\
0 & 0 & 0
    \end{bmatrix}x_r(t)
\end{align}
where $x_r(0)=[0,0,0.5]^\top$.
As this setup is specifically designed for Algorithms~\ref{Alg 4} and \ref{Alg 5}, we will implement these algorithms for tracking control of the F-16 aircraft with input constraint $\lambda =10$ to evaluate their performance. To assess the effectiveness of the proposed methods, the in-policy Q-learning algorithm \cite{Possieri} and the Online Q-learning algorithm \cite{VAMVOUDAKIS201714} are also simulated for the F-16 to track reference \( x_r(t) \) for comparison. For this case, the activation function for Algorithms~\ref{Alg 4} and \ref{Alg 5} is selected as: $\Phi(X(t))=X(t)\otimes_S X(t)$.

\begin{figure}
    \centering
    {\includegraphics[width=3.3in]{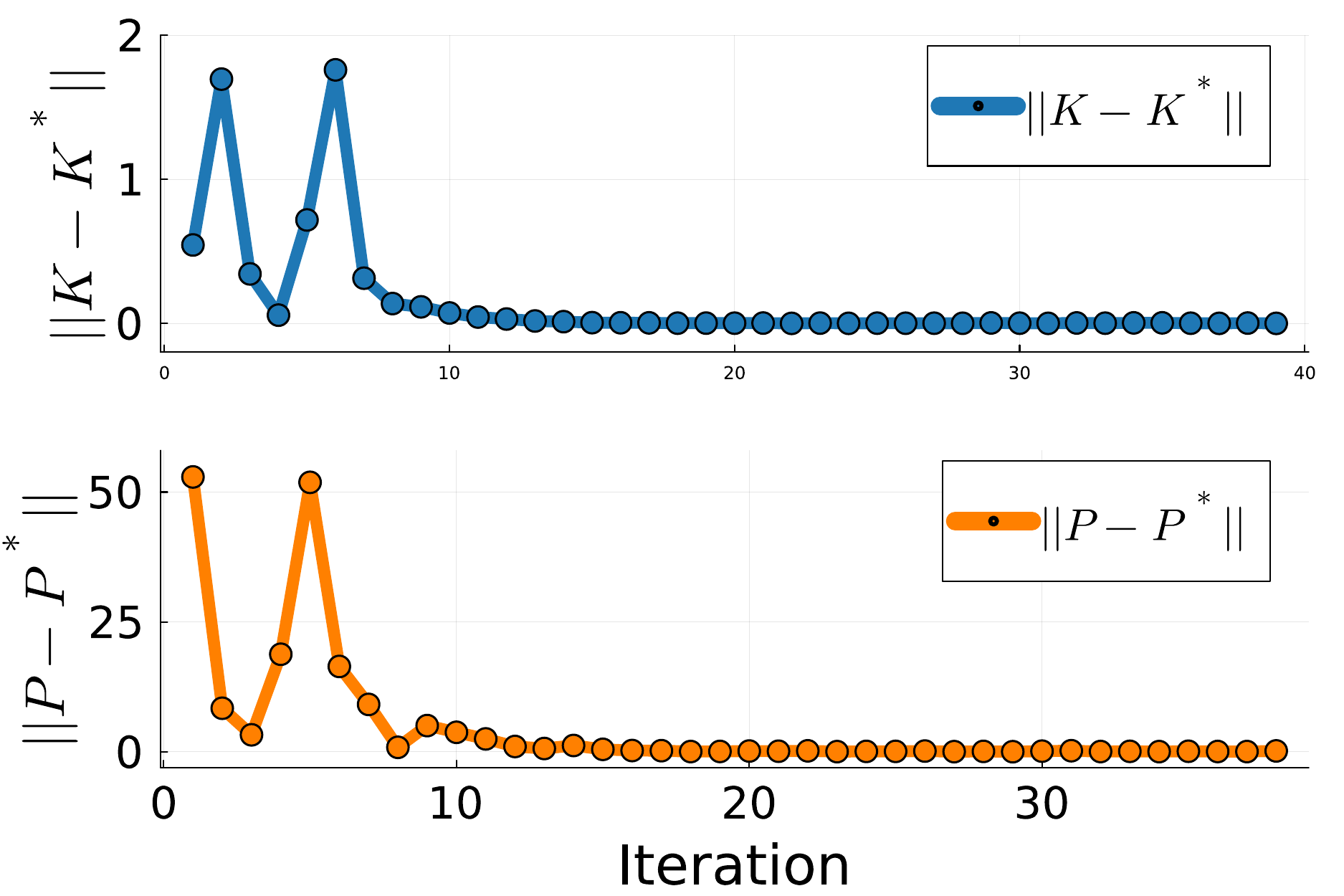}}
    \caption{ $P^i$ and $K_1^i$ using Algorithm 4 for the Tracking Control case.}
    \label{fig:F16_Al4con_KP_track}
\end{figure}

\begin{figure}
    \centering
    {\includegraphics[width=3.3in]{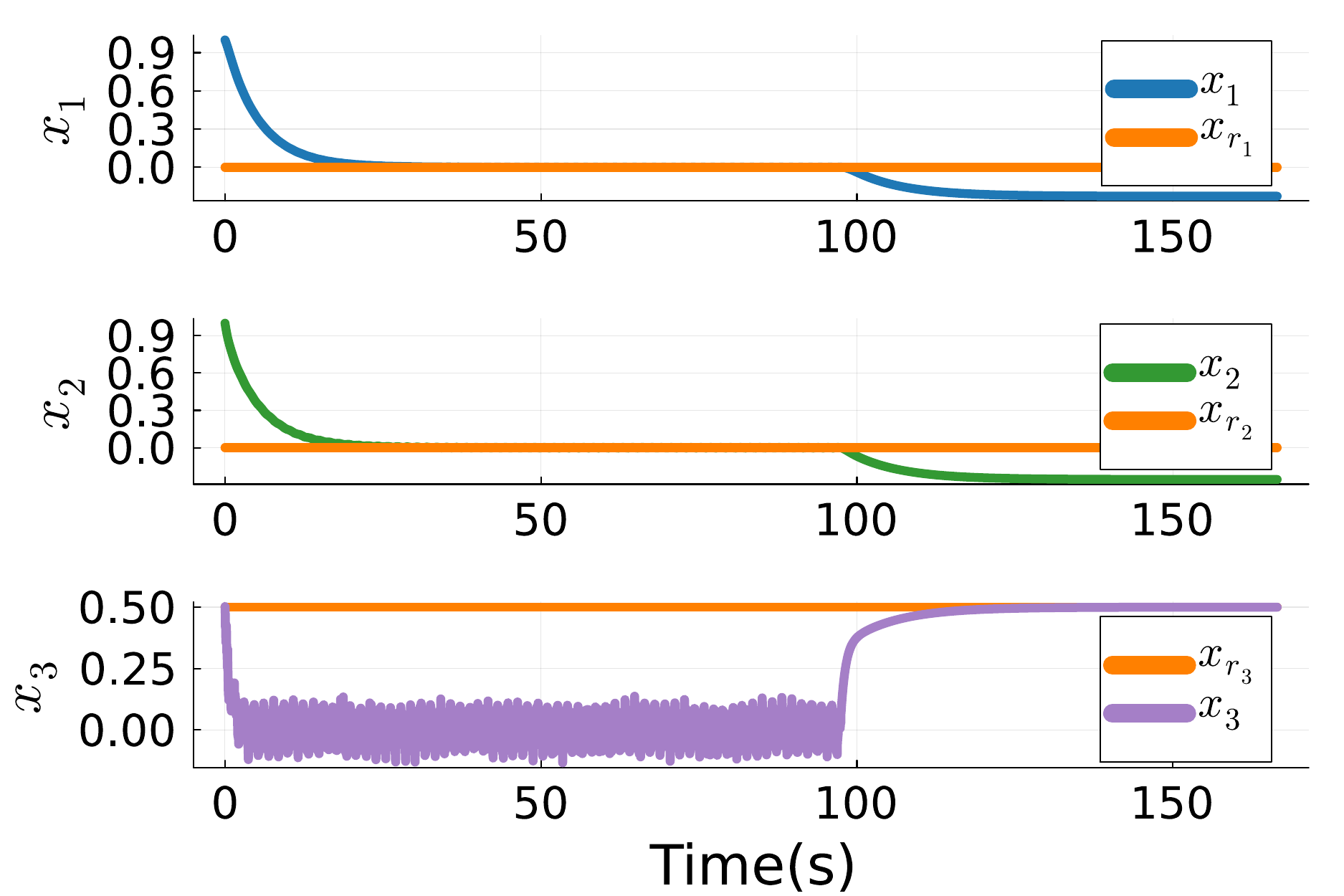}}
    \caption{ F-16 aircraft system state using Algorithm 4 for the Tracking Control case.}
    \label{F16_Al4con_state_track}
\end{figure}

\begin{figure}
    \centering
    {\includegraphics[width=3.3in]{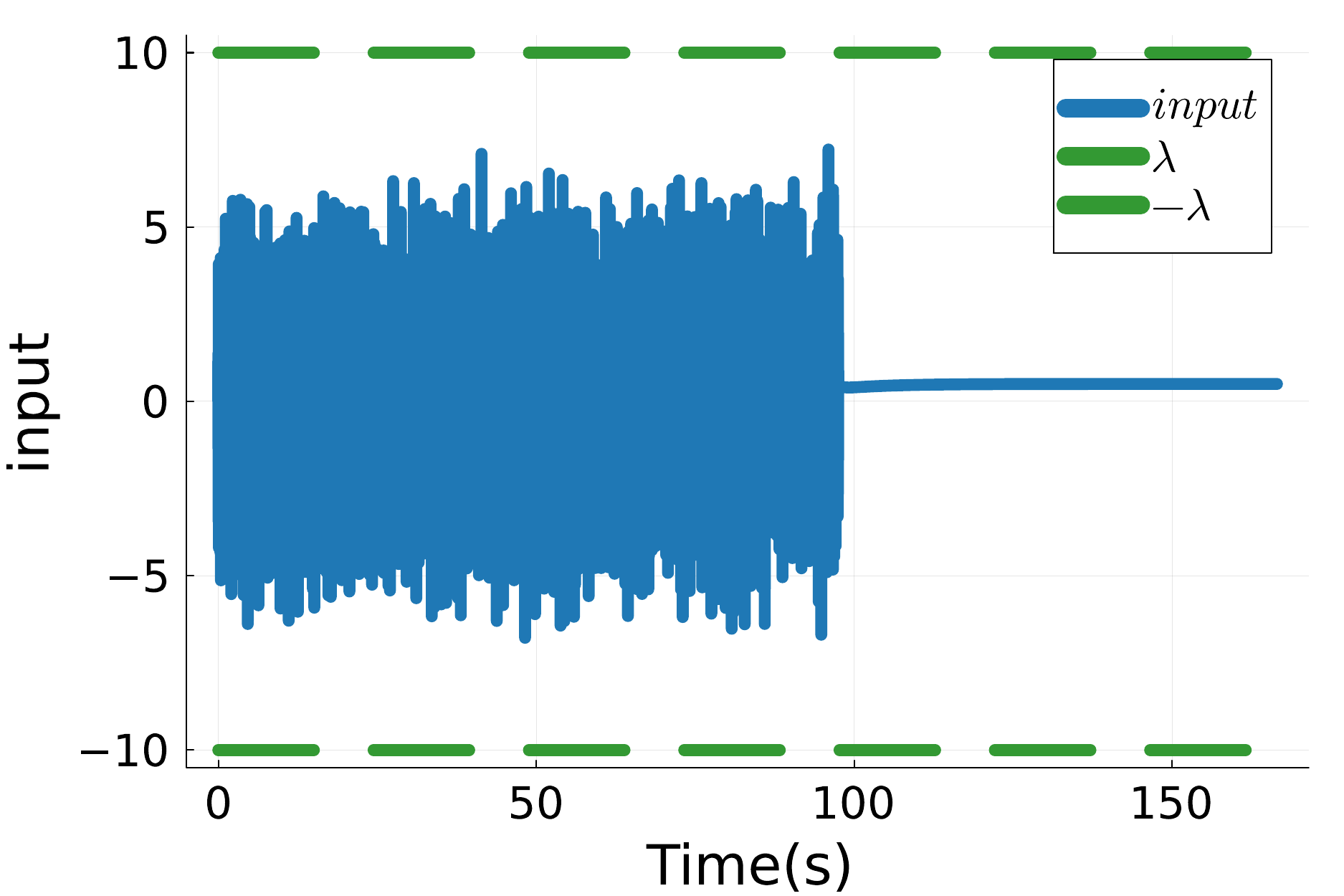}}
    \caption{Input of model using Algorithm 4 for the Tracking Control case.}
    \label{F16_Al4con_input_track}
\end{figure}

\begin{figure}
    \centering
    {\includegraphics[width=3.3in]{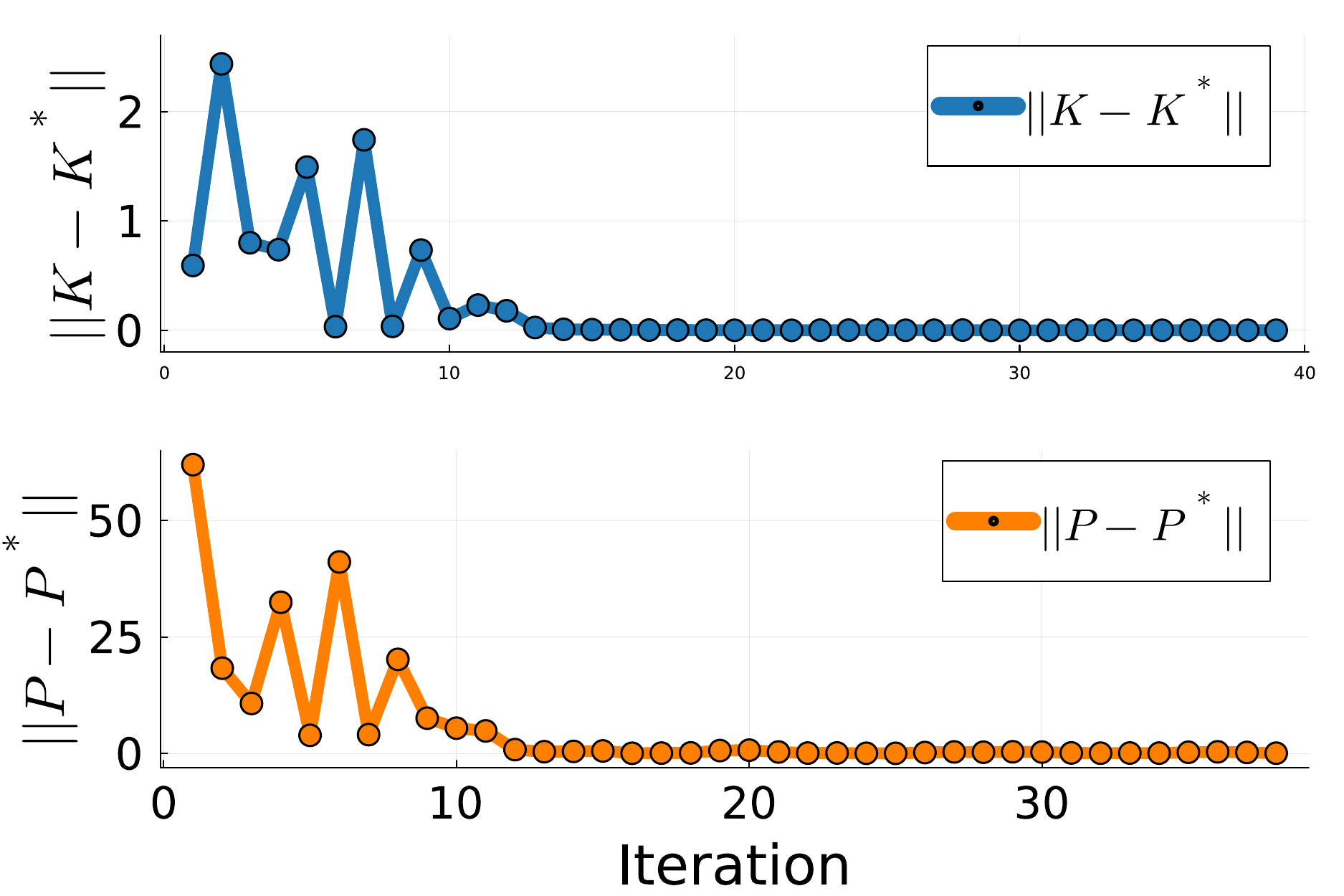}}
    \caption{ $P^i$ and $K_1^i$ using Algorithm 5 for the Tracking Control case.}
    \label{fig:F16_Al5con_KP_track}
\end{figure}

\begin{figure}
    \centering
    {\includegraphics[width=3.3in]{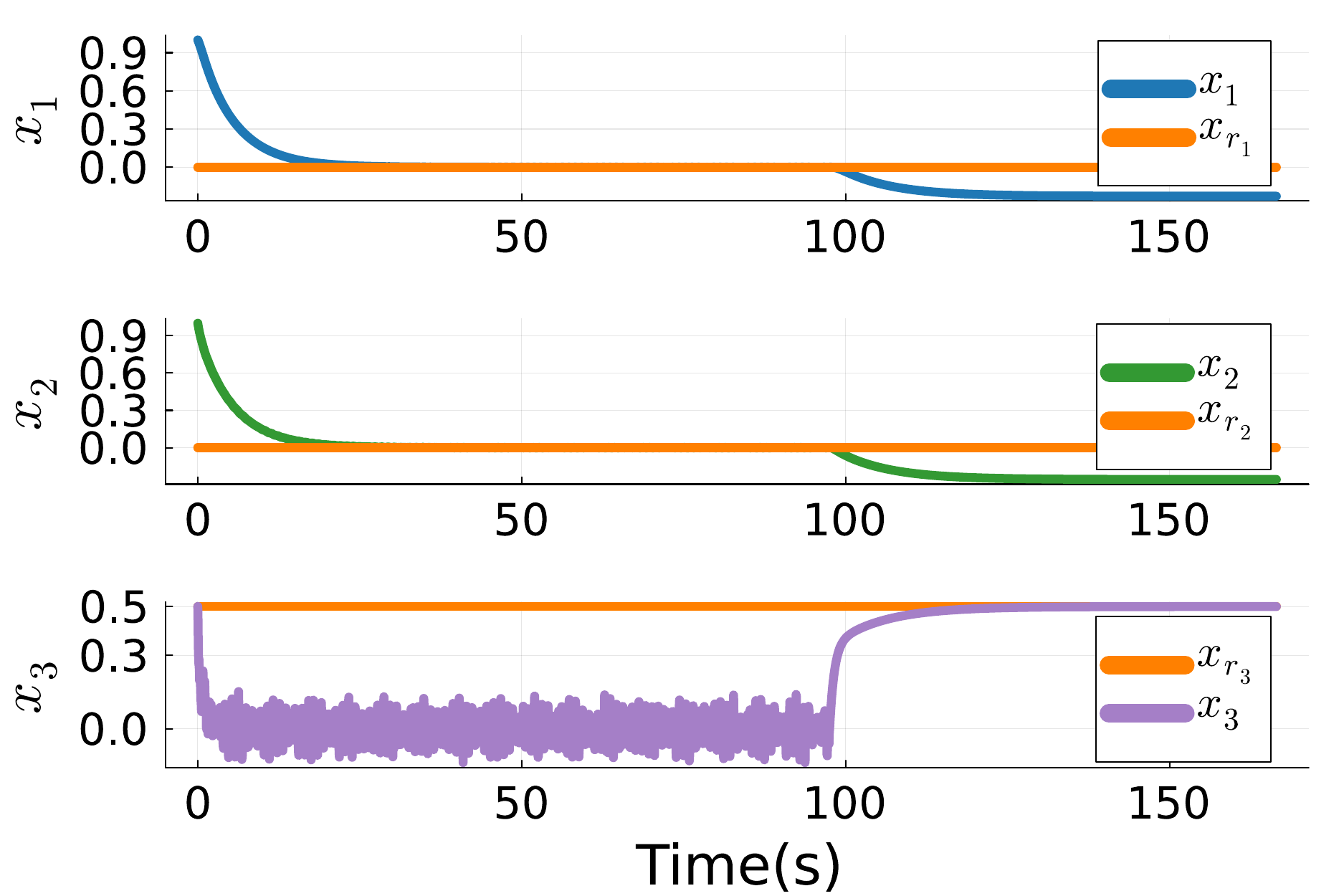}}
    \caption{ F-16 aircraft system state using Algorithm 4 for the Tracking Control case.}
    \label{F16_Al5con_state_track}
\end{figure}

\begin{figure}
    \centering
    {\includegraphics[width=3.3in]{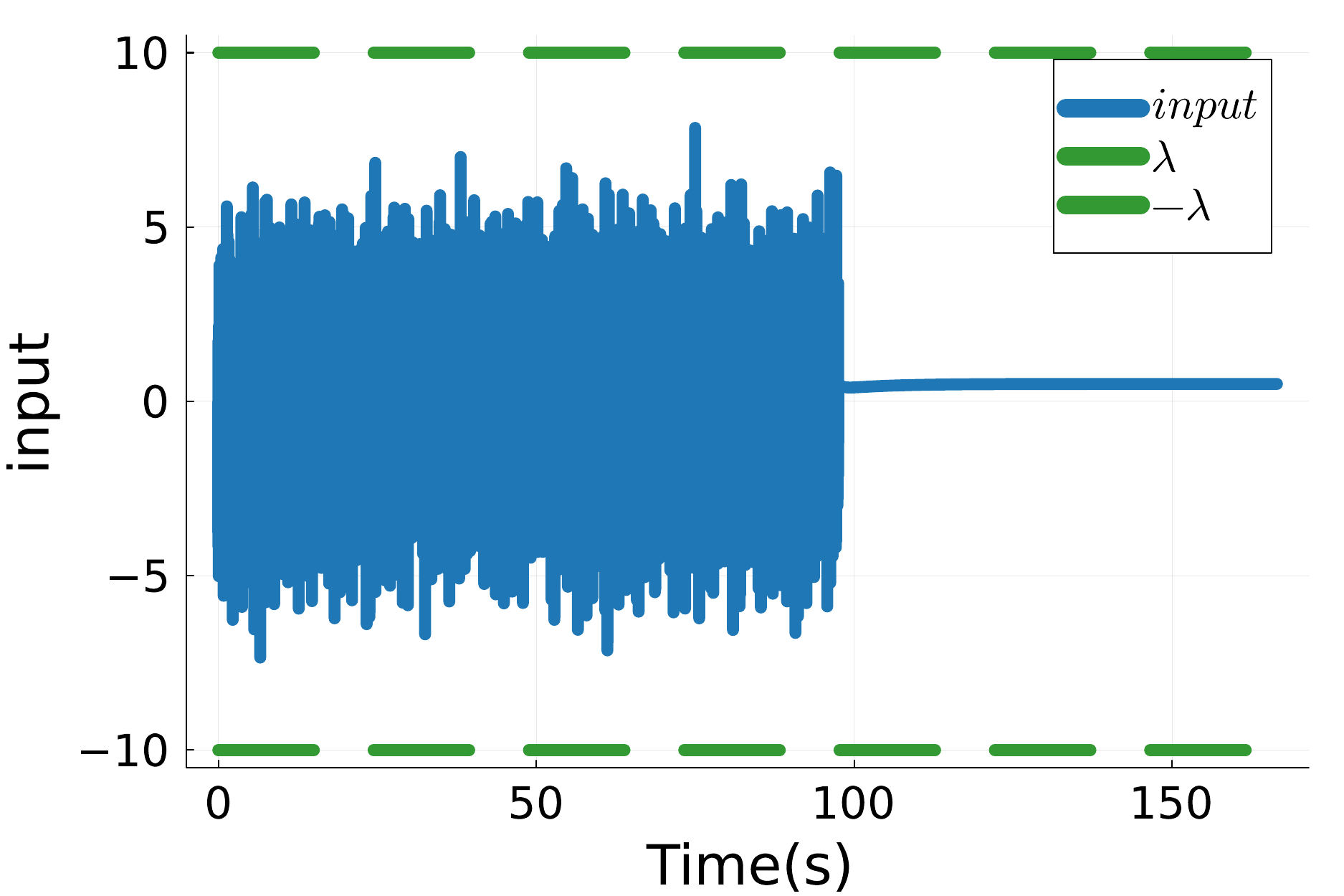}}
    \caption{Input of model using Algorithm 5 for the Tracking Control case.}
    \label{F16_Al5con_input_track}
\end{figure}

\begin{figure}
    \centering
    {\includegraphics[width=3.3in]{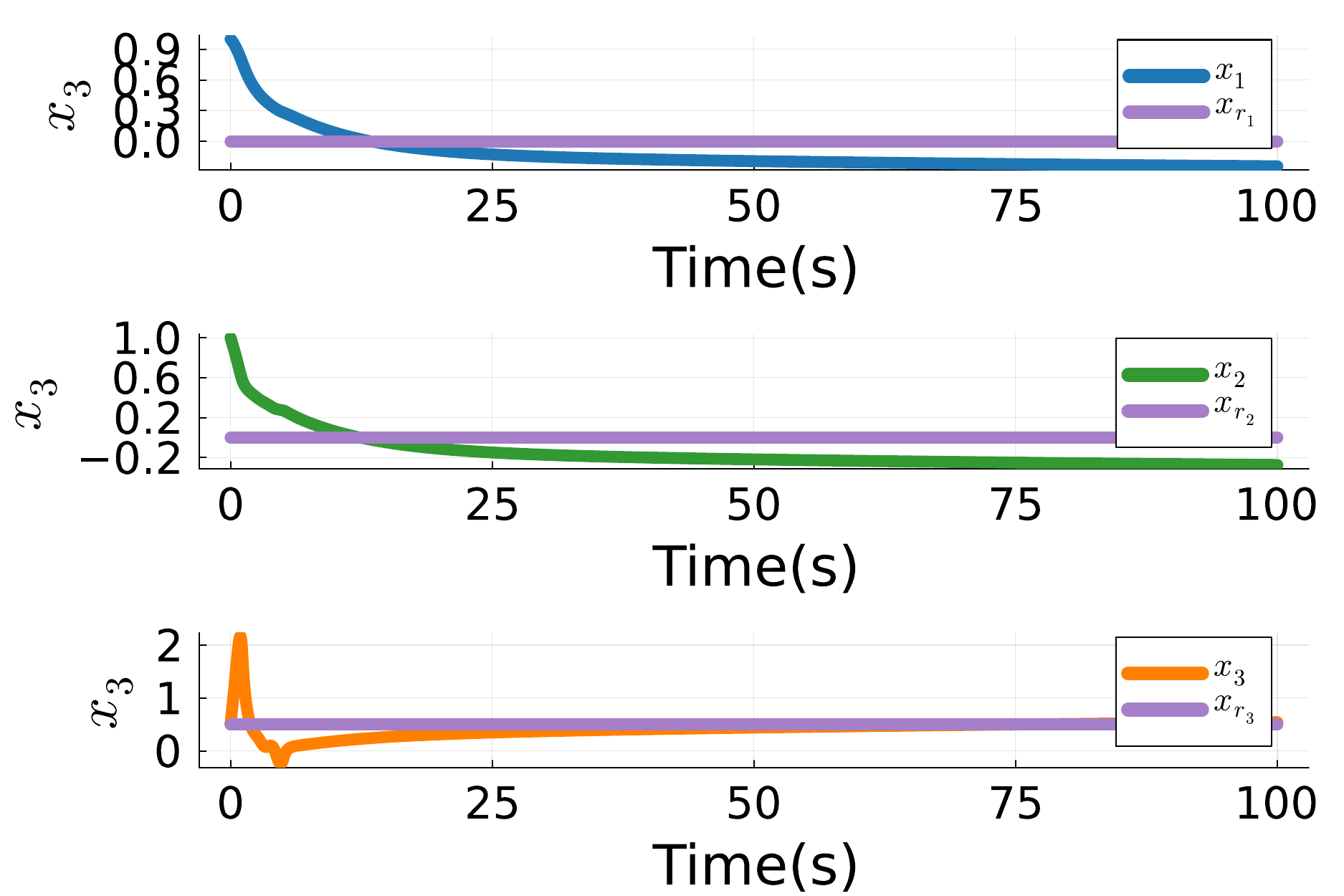}}
    \caption{ F-16 aircraft system state using Online Q-learning \cite{VAMVOUDAKIS201714} for the Tracking Control case.}
    \label{fig:comparation_1_state}
\end{figure}

\begin{figure}
    \centering
    {\includegraphics[width=3.3in]{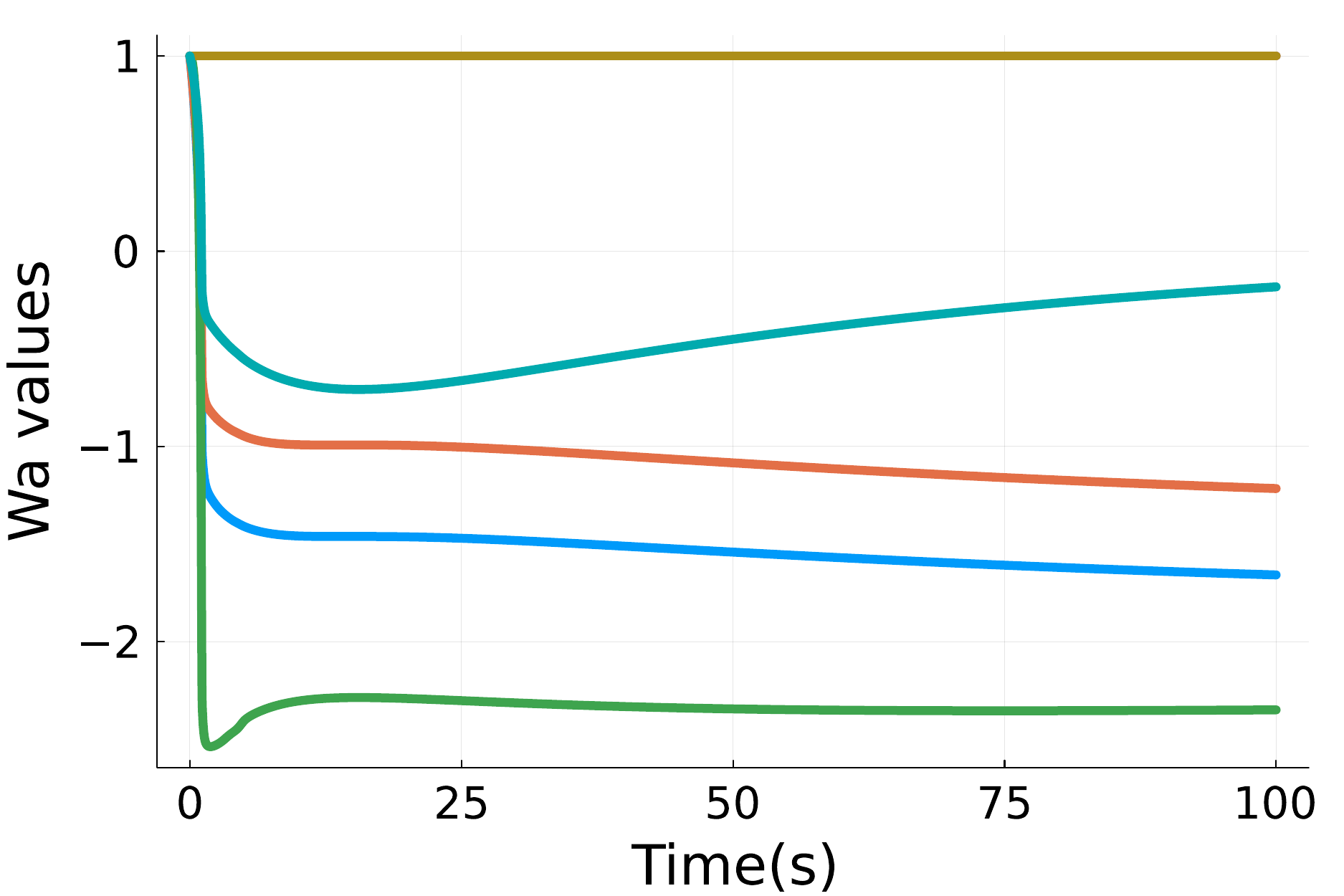}}
    \caption{ Actor weight of Online Q-learning \cite{VAMVOUDAKIS201714} for the Tracking Control case.}
    \label{fig:comparation_1_Wa}
\end{figure}

\begin{figure}
    \centering
    {\includegraphics[width=3.3in]{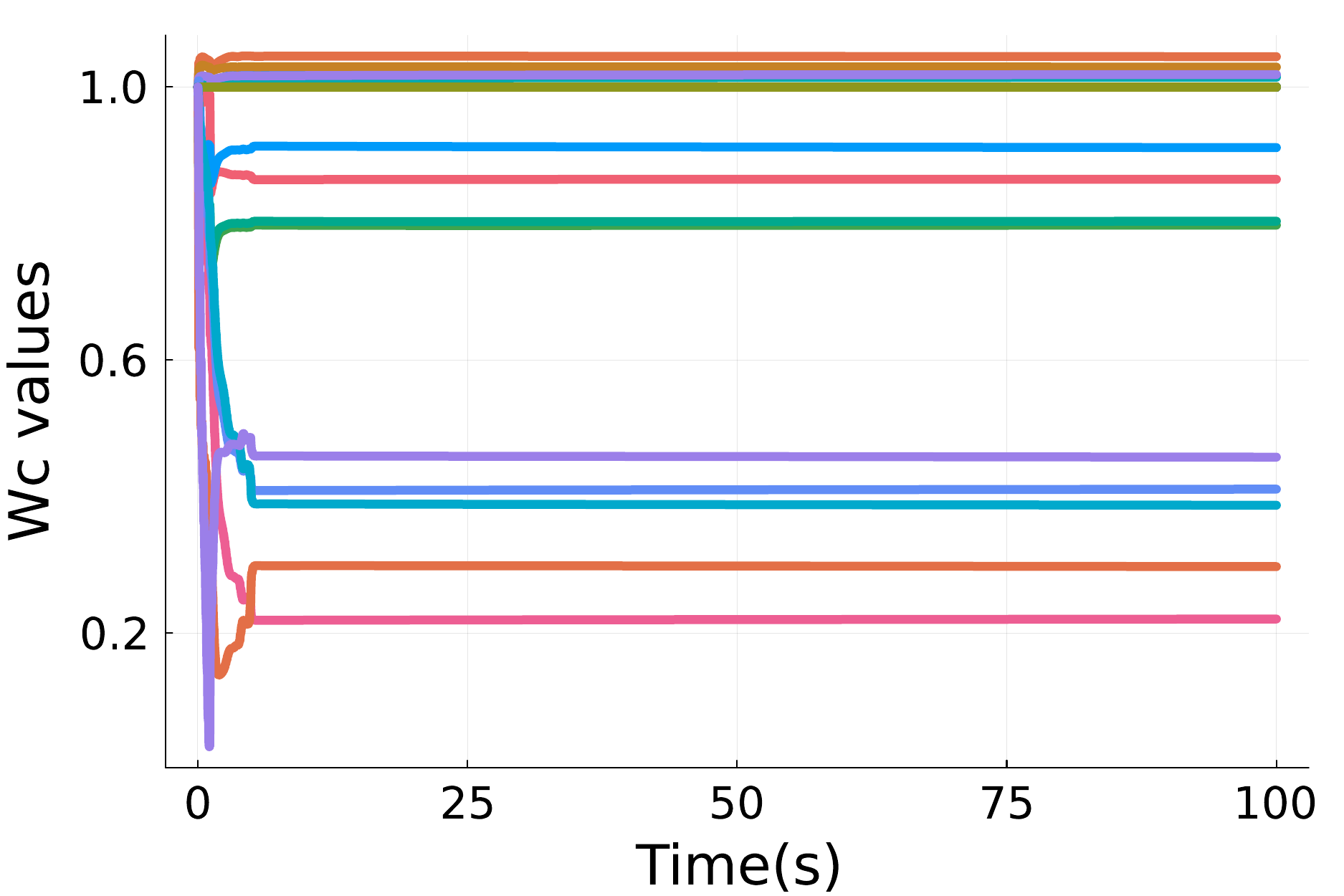}}
    \caption{ Critic weight of Online Q-learning \cite{VAMVOUDAKIS201714} for the Tracking Control case.}
    \label{fig:comparation_1_Wc}
\end{figure}

\begin{figure}
    \centering
    {\includegraphics[width=3.3in]{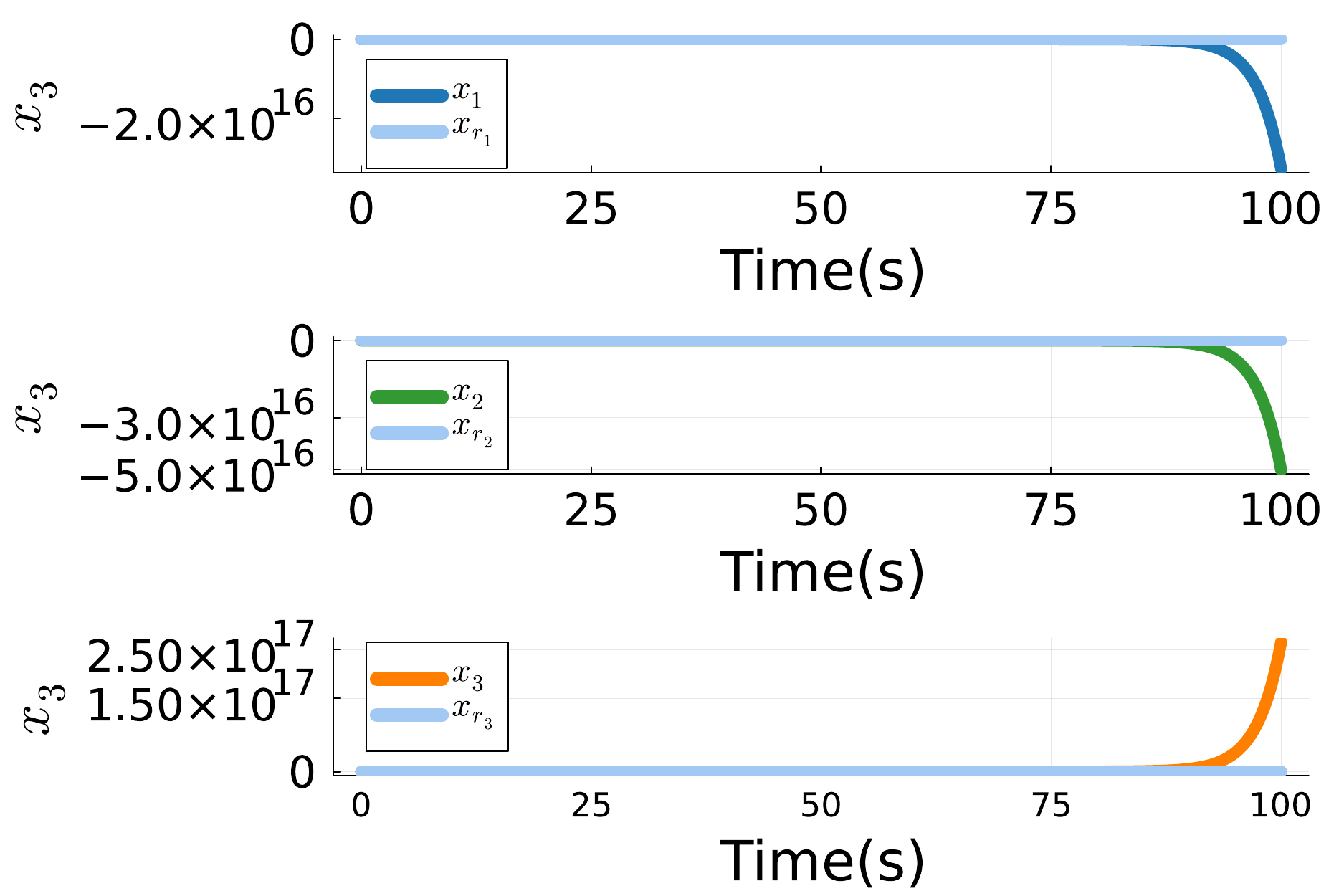}}
    \caption{ F-16 aircraft system state using On-Policy Q-learning \cite{Possieri} for the Tracking Control case.}
    \label{fig:comparation_2_state}
\end{figure}

\begin{figure}
    %\centering
    %\resizebox{\columnwidth}{!}
    \includegraphics[width=3.3in]{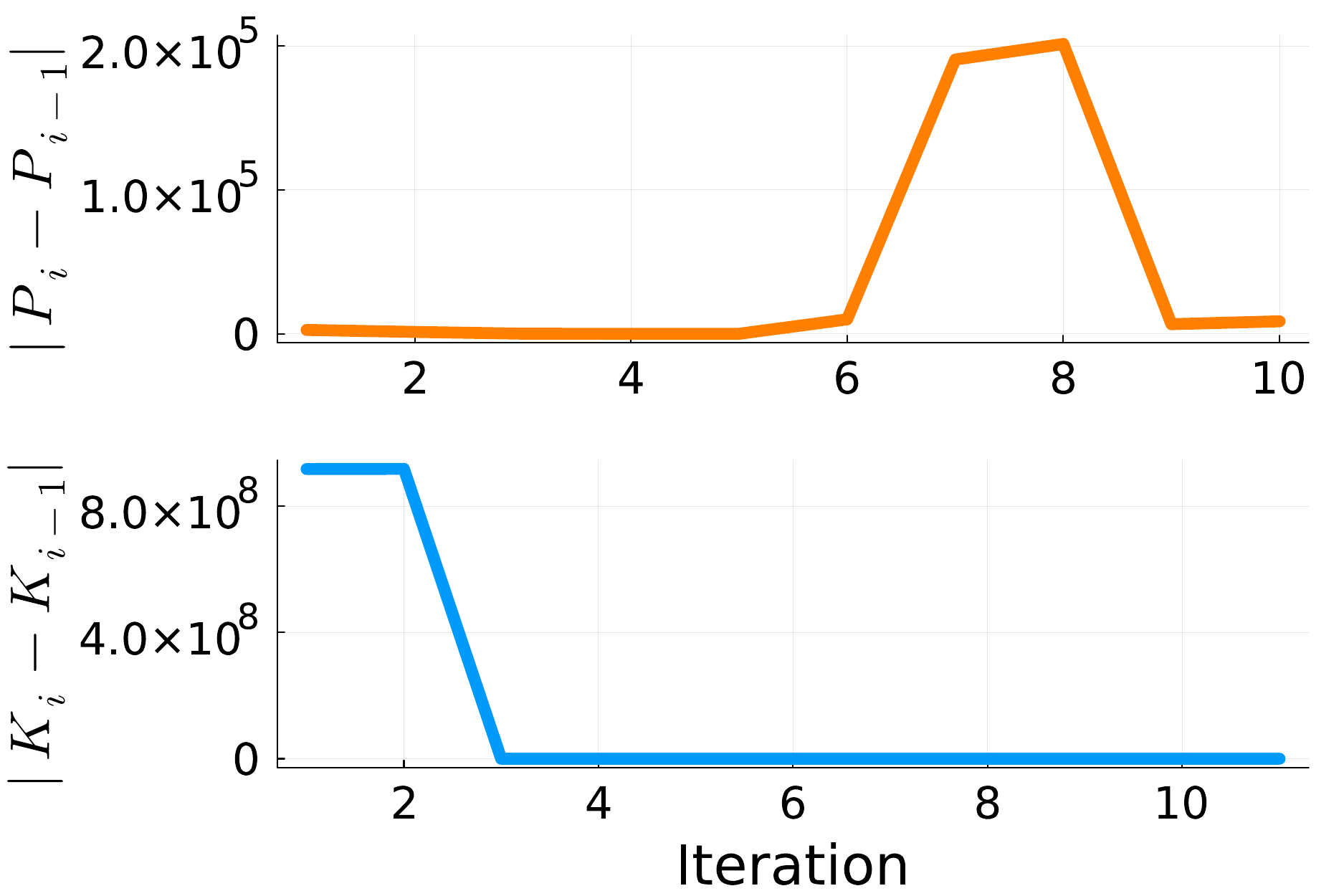}
    \caption{$P^i$ and $K_1^i$ using On-Policy Q-learning \cite{Possieri} for the Tracking Control case.}
    \label{fig:comparation_2_KP}
\end{figure}
It is evident from Figures~\ref{fig:F16_Al4con_KP_track}-\ref{F16_Al5con_input_track} that both Algorithms~\ref{Alg 4} and \ref{Alg 5} exhibit a similar convergence speed of approximately 37.5 seconds. Moreover, both algorithms enable the F-16 aircraft to track the reference signal $x_r(t)$ while ensuring that the control input remains within the specified constraints. As shown in Figures~\ref{F16_Al5con_input_track}-\ref{fig:comparation_2_KP}, the Online Q-learning method has a convergence time of approximately 100 seconds, while the online policy method does not exhibit convergence. Table 1 presents the error between the weight of the optimal value function, $\Delta P= P^*-P$, and the optimal policy, $\Delta K= K^*-K$ obtained from the model-free algorithms examined in this section. The results indicate that the On-policy Q-learning and Online Q-learning algorithms exhibit significant errors, suggesting inherent bias in their solutions. It can be observed that, as noted in Remark 1, this bias causes the On-policy method to fail to converge and results in system instability. On the other hand, although the Online Q-learning method exhibits bias, the system state remains stable. However, this bias indicates that the method policy is not yet optimal.

\begin{table}[]
    \centering
    \begin{tabular}{|c|c c|c c|}
    \hline
       \multirow{2}{*}{Algorithms}  & \multicolumn{2}{c|}{Feedback Control} & \multicolumn{2}{c|}{Tracking Control}  \\
       \cline{2-5}
         & $\|\Delta P\|$ & $\|\Delta K\|$ &$\|\Delta P\|$ & $\|\Delta K\|$\\
         \hline
         Algorithm 2 & $0.0156$ & $0.0011$ &  & \\
         Algorithm 3 & $0.0025$ & $8 \times 10^{-5}$& &\\
         Algorithm 4 & $0.0095$ & $0.0004$ & $0.1769$  & $0.0002$ \\
         Algorithm 5 & $0.0007$ & $8 \times 10^{-6}$ & $0.0536$& $7 \times 10^{-4}$ \\
         Algorithm \cite{VAMVOUDAKIS201714} &$4.7866$& $1.2239$&     $9.0753$& $8.3354$ \\
         Algorithm \cite{Possieri} &$3.0119$& $2.8577$& $8 \times 10^3$&  $8 \times 10^3$\\
         \hline
    \end{tabular}
    \caption{Error between the solutions obtained from the algorithms and the optimal solution}
    \label{table 1}
\end{table}
\section{Conclusion}\label{Section 5}
In this paper, four off-policy Q-learning algorithms designed for linear continuous-time systems are presented to address the LQR problem and the LQT problem with constrained input. A novel concept of the Advantage function for continuous-time systems is also introduced, serving as the core idea for developing these algorithms within a model-free framework. Among the proposed algorithms, we demonstrate two implementation approaches: one iterates over fixed time intervals, while the other is time-iterative. These algorithms can be viewed as a data-driven extension of the well-known Kleinman algorithm, and as a result, inherit its convergence properties. Furthermore, in the LQR problem, an admissible gain can be obtained without requiring the system model by solving the LMI with constraints derived from data collected along the system trajectory. For the LQT problem with constrained input, a heuristic search method has been introduced to reduce the effort required to find the admissible gain. Finally, simulations validate the effectiveness of the proposed methods.

%\clearpage

\printcredits
\section{Declaration of Competing Interest}
Authors declare that they have no conflict of interest.

%% Loading bibliography style file
%\bibliographystyle{model1-num-names_i}
\bibliographystyle{elsarticle-num}
\bibliography{references}
\begin{comment}

\bio{Nam.jpg}
\textbf{Phuong Nam Dao:} He received the Ph.D. degree in Industrial Automation from Hanoi University of Science and Technology, Hanoi, Vietnam in 2013. Currently, he holds the position as Associate Professor at Hanoi University of Science and Technology, Vietnam. His research interests include control of robotic systems and robust/adaptive, optimal control. He is the author/co-author of more than
90 papers (Journals, Conferences, etc.).
His ORCID is: https://orcid.org/0000-0002-8333-5572
\endbio

\bio{Quang.jpg}
\textbf{Van Quang Nguyen:} He is currently a talent program student in Automation and Control Engineering at Hanoi University of Science and Technology, Vietnam. He is currently working toward a PhD degree and his research interests include control and navigation system, optimal control, reinforcement learning and deep learning.
\endbio

\bio{Hoang Anh.jpg}
\textbf{Hoang Anh Nguyen Duc:} He is currently studying the Bachelor degree in talent program of Control Engineering and Automation from Hanoi University of Science and Technology, Hanoi, Vietnam. He is currently working toward a PhD degree and his research interests include optimal control, intelligent control and its application for robotic systems.
\endbio
\end{comment}
% Loading bibliography database

\end{document}